%% file: main-easychair.tex
\documentclass{easychair}

\usepackage{xspace}
\usepackage[font=small,skip=3pt]{caption}
\usepackage{tikz}
\usetikzlibrary{shapes,arrows,positioning}
\usepackage{multicol}
\usepackage{mdframed}
\mdfdefinestyle{codebox}{innermargin=17pt,outermargin=17pt}
\usepackage{graphicx}
\usepackage{amsmath}
\usepackage{latexsym}
\usepackage{amsfonts}
\usepackage{xcolor}
\usepackage{url}
\usepackage{listings}
\usepackage{times}
\usepackage{wrapfig}
\usepackage{amsthm}
\usepackage{enumitem}
\usepackage{algorithmicx}
\usepackage{algpseudocode}
\usepackage{longtable}
\usepackage{booktabs}
\usepackage{pgfplots}
\usepackage{subcaption}

\usepackage{pstricks}
\usepackage{pst-node}
\psset{nodesep=1ex}

\usepackage[T1]{fontenc}
\usepackage[draft]{minted}
\usepackage{upquote}
\AtBeginDocument{%
}

\input{comm-anthony}
\input{comm}

\input{namedenvs}

\input{numbers}

\newenvironment{to-do}
{ \rule{1ex}{1ex}\hspace{\stretch{1}} \bfseries}
{ \hspace{\stretch{1}}\rule{1ex}{1ex} \vspace{1ex}}

\newenvironment{compactenum}{
    \begin{enumerate}
}{
    \end{enumerate}
}

\newtheorem{theorem}{Theorem}[section]
\newtheorem{remark}[theorem]{Remark}

\newtheorem{proposition}[theorem]{Proposition}
\newtheorem{definition}[theorem]{Definition}
\newtheorem{lemma}[theorem]{Lemma}

\newcommand\acmeasychair[2]{#2}

\begin{document}

\title{CSS Minification via Constraint Solving
\\ {\large (Technical Report)}}
\titlerunning{CSS Minification via Constraint Solving}

\author{
    Matthew Hague\inst{1}
    \and
    Anthony W. Lin\inst{2,3}
    \and
    Chih-Duo Hong\inst{3}
}
\authorrunning{M. Hague, A. W. Lin, and C.-H. Hong}
\institute{
  Royal Holloway, University of London, UK
  \and
  Technische Universit\"{a}t Kaiserslautern, Germany
  \and
  University of Oxford, UK
}

\maketitle

\begin{abstract}
    \input{new_abstract}
\end{abstract}

\input{intro}
\input{css_by_example}

\input{prelim}

\input{css2graph}
\input{edgeOrder}
\input{intersection}

\input{graph2maxsat}

\input{experiments}

\input{related}

\input{conclusion}

\subsubsection*{Acknowledgements}
    We are grateful for the support that we received from the Engineering and
    Physical Sciences Research Council [EP/K009907/1], Google
    (Faculty Research Award), and European Research Council (grant agreement no.
    759969).
    We also thank Davood Mazinanian for answering questions about his work, and
    anonymous reviewers for their helpful comments.


\input{main-easychair.bbl}
\appendix
\begin{center}
    \bfseries \huge Appendix
\end{center}

\input{graph2maxsat-appendix}
\input{edgeOrder-appendix}

\input{intersection-appendix}

\input{experiments-appendix}

\end{document}

%% file: comm-anthony.tex
\newif\ifdraft\draftfalse
\ifdraft
\newcommand\anthony[1]{{\color{blue}
[#1 - \textbf{Anthony}]}}
\newcommand\matt[1]{{\color{red}
[#1 - \textbf{Matt}]}}
\newcommand\chihduo[1]{{\color{green}
[#1 - \textbf{Chih-Duo}]}}
\newcommand\alchanged[1]{{\color{blue}{#1}}}
\newcommand\mhchanged[1]{{\color{red}{#1}}}

\else
\newcommand\anthony[1]{}
\newcommand\matt[1]{}
\newcommand\chihduo[1]{}
\newcommand\alchanged[1]{#1}
\newcommand\mhchanged[1]{#1}

\fi

\newcommand{\Seq}{\sigma}

\newcommand{\CSSgraph}{\ensuremath{\mathcal{G}}}
\newcommand{\selNodes}{\ensuremath{S}}
\newcommand{\propNodes}{\ensuremath{P}}
\newcommand{\selsBucket}{\ensuremath{X}}
\newcommand{\propsBucket}{\ensuremath{\overline{Y}}}
\newcommand{\CSSedges}{\ensuremath{E}}
\newcommand{\CSSvertices}{\ensuremath{V}}
\newcommand{\edgeOrder}{\prec}
\newcommand{\nodeWeight}{\text{\texttt{wt}}}
\newcommand{\Weight}{\text{\texttt{wt}}}
\newcommand{\inCSSgraph}{\ensuremath{(\selNodes,\propNodes,\CSSedges,\edgeOrder,\nodeWeight)}}

\newcommand{\covering}{\mathcal{C}}
\newcommand{\bucket}{\overline{B}}
\newcommand{\biclique}{B}
\newcommand{\incovering}[2]{\ensuremath{\{\bucket_i\}_{i=#1}^{#2}}}
\newcommand{\inbucket}{(\selsBucket,\propsBucket)}
\newcommand{\inbiclique}{(\selsBucket,\propsSet)}
\newcommand{\Trim}[1]{{#1}_{\downarrow}}

\newcommand{\inSeq}[3]{\ensuremath{#1[#3 \to #2]}} 

\newcommand{\Index}{\text{index}}

\newcommand{\propOrder}{\vartriangleleft}

\newcommand{\nssatts}{\nspaces\times\atts}

\newcommand{\satcss}{\textsc{SatCSS}\xspace}

%% file: comm.tex
\usepackage{amssymb}
\usepackage{mathtools}

\newcommand{\OMIT}[1]{}

\newcommand\sizeof[1]{\left|#1\right|}
\newcommand\mytt[1]{\text{\textup{\texttt{#1}}}}

\renewcommand\iff\Leftrightarrow
\newcommand{\iffdef}{\ensuremath{\stackrel{\text{def}}{\Leftrightarrow}}}

\newcommand\mathttt[1]{\text{\texttt{#1}}}

\newcommand{\synalt}{\ \bigm\vert\ } 

\newcommand{\Lang}{\mathcal{L}}
\newcommand{\ialphabet}{\Sigma}

\newcommand{\posN}{\mathbb{N}_{>0}}
\newcommand{\N}{\mathbb{N}}
\newcommand{\Z}{\mathbb{Z}}
\newcommand{\posZ}{\Z_{> 0}}

\newcommand{\sem}[1]{[\![#1 ]\!]}
\newcommand\numof{n}
\newcommand\numofalt{m}
\newcommand\runlen{\ell}
\newcommand\brac[1]{\left({#1}\right)}
\newcommand\tup[1]{\left({#1}\right)}
\newcommand\eseq{\varepsilon}
\newcommand\idxi{i}
\newcommand\idxj{j}
\newcommand\cha{a}

\newcommand\set[1]{\left\{#1\right\}}
\newcommand\setcomp[2]{\left\{{#1}\ \middle|\ {#2}\right\}}
\newcommand\ap[2]{{#1}\mathord{\brac{#2}}}
\newcommand\abs[1]{\left|#1\right|}
\newcommand\finfuns[2]{\ap{\mathcal{F}_{\text{fin}}}{#1, #2}}

\newcommand{\defn}[1]{\emph{#1}}

\newcommand{\NP}{\ensuremath{\mathsf{NP}}}

\newcommand\tree{T}
\newcommand\node{\eta}
\newcommand\treedom{D}
\newcommand\treelab{\lambda}
\newcommand\treedir{\iota}
\newcommand\treelabs{\Sigma}
\newcommand\treedescendant{\mathop{\textcircled{\hspace{.2pt}\raisebox{.35pt}{\scriptsize $\gg$}}}}
\newcommand\treechild{\mathop{\textcircled{\hspace{.7pt}\raisebox{.35pt}{\scriptsize $>$}}}}
\newcommand\treeneighbour{\mathop{\textcircled{\raisebox{.35pt}{\scriptsize $+$}}}}
\newcommand\treesibling{\mathop{\textcircled{\raisebox{.35pt}{\scriptsize $\sim$}}}}

\newcommand\treeset[4]{\ap{\mathrm{Trees}}{#1, #2, #3, #4}}

\newcommand\treelabns{\treelab_{\texttt{S}}}
\newcommand\treelabele{\treelab_{\texttt{E}}}
\newcommand\treelabatts{\treelab_{\texttt{A}}}
\newcommand\treelabpclss{\treelab_{\texttt{P}}}

\newcommand\css{\varphi}
\newcommand\csssim{\sigma}
\newcommand\cssconds{\Theta}
\newcommand\csscond{\theta}
\newcommand\csstype{\tau}
\newcommand\csssimnoneg{\ensuremath{\sigma_{\neg}}}
\newcommand\csscondnoneg{\ensuremath{\theta_{\neg}}}
\newcommand\coefa{\alpha}
\newcommand\coefb{\beta}

\newcommand\nspaces{\mathrm{NS}}
\newcommand\atts{A}
\newcommand\attfun{f_A}
\newcommand\eles{\text{\scshape Ele}}
\newcommand\alphabet{\Gamma}
\newcommand\cspace{\text{\textvisiblespace}}
\newcommand\cdash{\mathord{\text{-}}}

\newcommand\id{\iota}
\newcommand\ns{s}
\newcommand\ele{e}
\newcommand\csslang{l}
\newcommand\pcls{p}
\newcommand\pslink{\operatorname{\mytt{:link}}}
\newcommand\psvisited{\operatorname{\mytt{:visited}}}
\newcommand\pshover{\operatorname{\mytt{:hover}}}
\newcommand\psactive{\operatorname{\mytt{:active}}}
\newcommand\psfocus{\operatorname{\mytt{:focus}}}
\newcommand\pstarget{\operatorname{\mytt{:target}}}
\newcommand\psenabled{\operatorname{\mytt{:enabled}}}
\newcommand\psdisabled{\operatorname{\mytt{:disabled}}}
\newcommand\pschecked{\operatorname{\mytt{:checked}}}
\newcommand\psindeterminate{\operatorname{\mytt{:indeterminate}}}
\newcommand\psroot{\operatorname{\mytt{:root}}}
\newcommand\psfirstchild{\operatorname{\mytt{:first-child}}}
\newcommand\pslastchild{\operatorname{\mytt{:last-child}}}
\newcommand\psfirstoftype{\operatorname{\mytt{:first-of-type}}}
\newcommand\pslastoftype{\operatorname{\mytt{:last-of-type}}}
\newcommand\psonlychild{\operatorname{\mytt{:only-child}}}
\newcommand\psonlyoftype{\operatorname{\mytt{:only-of-type}}}
\newcommand\psempty{\operatorname{\mytt{:empty}}}
\newcommand\psfirstline{\operatorname{\mytt{::first-line}}}
\newcommand\psfirstletter{\operatorname{\mytt{::first-letter}}}
\newcommand\psbefore{\operatorname{\mytt{::before}}}
\newcommand\psafter{\operatorname{\mytt{::after}}}
\newcommand\cssdescendant{\ensuremath{\mathop{\mytt{>\!>}}}}
\newcommand\csschild{\ensuremath{\mathop{\mytt{>}}}}
\newcommand\cssneighbour{\ensuremath{\mathop{\mytt{+}}}}
\newcommand\csssibling{\ensuremath{\mathop{\raisebox{-.2em}{\mytt{\~}}}}}
\newcommand\att{a}
\newcommand\attundef{\bot}
\newcommand\attval{v}
\newcommand\attvalalt{u}
\newcommand\isany{\mytt{*}}

\newcommand\genop{o}
\newcommand\selectors{\text{\scshape Sel}}
\newcommand\nodeselectors{\text{\scshape NSel}}
\newcommand\opis{\mathop{\mytt{=}}}
\newcommand\ophas{\mathop{\mytt{\raisebox{.25em}{\texttildelow}=}}}
\newcommand\opbegin{\mathop{\mytt{|\!=}}}
\newcommand\opstrbegin{\mathop{\mytt{\textasciicircum=}}}
\newcommand\opstrend{\mathop{\mytt{\$=}}}
\newcommand\opstrsub{\mathop{\mytt{*=}}}
\newcommand\attop{\mathit{op}}
\newcommand\langatt{\mytt{lang}}
\newcommand\classatt{\mytt{class}}
\newcommand\idatt{\mytt{id}}
\newcommand\csspipe{\mytt{|}}

\newcommand\qele[2]{#1\mathord{:}#2}
\newcommand\qatt[2]{#1\mathord{:}#2}
\newcommand\isclass[1]{\mytt{.}#1}
\newcommand\isid[1]{\mytt{\#}#1}
\newcommand\isanyns[1]{\mytt{(}#1\csspipe\mytt{*)}}
\newcommand\isele[1]{#1}
\newcommand\iselens[2]{\mytt{(}#1\csspipe{#2}\mytt{)}}
\newcommand\pslang[1]{\operatorname{\mytt{:lang}}\mathord{\mytt{(}#1\mytt{)}}}
\newcommand\psnthchild[2]{\operatorname{\mytt{:nth-child(}{#1}\mytt{n + }{#2}\mytt{)}}}
\newcommand\psnthlastchild[2]{\operatorname{\mytt{:nth-last-child(}{#1}\mytt{n + }{#2}\mytt{)}}}
\newcommand\psnthoftype[2]{\operatorname{\mytt{:nth-of-type(}{#1}\mytt{n + }{#2}\mytt{)}}}
\newcommand\psnthlastoftype[2]{\operatorname{\mytt{:nth-last-of-type(}{#1}\mytt{n + }{#2}\mytt{)}}}
\newcommand\cssneg[1]{\operatorname{\mytt{:not(}{#1}\mytt{)}}}
\newcommand\hasatt[1]{\mytt{[}{#1}\mytt{]}}
\newcommand\opatt[3]{\mytt{[}{#1}\ {#2}\ {#3}\mytt{]}}
\newcommand\attis[2]{\opatt{#1}{\opis}{#2}}
\newcommand\atthas[2]{\opatt{#1}{\ophas}{#2}}
\newcommand\attbegin[2]{\opatt{#1}{\opbegin}{#2}}
\newcommand\attstrbegin[2]{\opatt{#1}{\opstrbegin}{#2}}
\newcommand\attstrend[2]{\opatt{#1}{\opstrend}{#2}}
\newcommand\attstrsub[2]{\opatt{#1}{\opstrsub}{#2}}
\newcommand\hasattns[2]{\mytt{[}{#1}\csspipe{#2}\mytt{]}}
\newcommand\opattns[4]{\mytt{[}{#1}\csspipe{#2}\ {#3}\ {#4}\mytt{]}}
\newcommand\attisns[3]{\opattns{#1}{#2}{\opis}{#3}}
\newcommand\atthasns[3]{\opattns{#1}{#2}{\ophas}{#3}}
\newcommand\attbeginns[3]{\opattns{#1}{#2}{\opbegin}{#3}}
\newcommand\attstrbeginns[3]{\opattns{#1}{#2}{\opstrbegin}{#3}}
\newcommand\attstrendns[3]{\opattns{#1}{#2}{\opstrend}{#3}}
\newcommand\attstrsubns[3]{\opattns{#1}{#2}{\opstrsub}{#3}}

\newcommand\cssaut{\mathcal{A}}
\newcommand\astates{Q}
\newcommand\atrans{\Delta}
\newcommand\astate{q}
\newcommand\ainitstate{q^{\text{in}}}
\newcommand\arrchild{\downarrow}
\newcommand\arrneighbour{\rightarrow}
\newcommand\arrsibling{\rightarrow_+}
\newcommand\arrlast{\circ}
\newcommand\arrgen{d}
\newcommand\weakord{\lesssim}
\newcommand\afinstate{q_f}
\newcommand\atrant{t}

\newcommand\atran[2]{\xrightarrow[#2]{#1}}
\newcommand\aaccepts[3]{\tup{#1, #2} \in \ap{\Lang}{#3}}
\newcommand\selaut[1]{\cssaut_{#1}}
\newcommand\midstate[1]{\bullet_{#1}}
\newcommand\selstate[1]{\circ_{#1}}

\newcommand\fakeele{\tau}
\newcommand\nullele{\bot}

\newcommand\nodeset{\mathit{Nodes}}
\newcommand\elemap{\theta}

\newcommand\elesof[1]{\eles_{#1}}

\newcommand\selsof[1]{S_{#1}}

\newcommand\numsels[1]{\sizeof{#1}_\csssim}
\newcommand\finelesof[1]{\finof{\eles_{#1}}}
\newcommand\finof[1]{\ap{\downharpoonleft}{#1}}

\newcommand\pres{\theta}

\newcommand\ptrue{\top}
\newcommand\pfalse{\bot}
\newcommand\xvar{\denotevar{x}}

\newcommand\assignment{\rho}

\newcommand\nvar{\denotevar{n}}
\newcommand\pclss{P}

\newcommand\consistent{\mathrm{Consistent}}
\newcommand\consistentnums{\mathrm{Consistent}_n}

\newcommand\consistentids{\mathrm{Consistent}_i}
\newcommand\consistentpseudo{\mathrm{Consistent}_p}
\newcommand\shiftvar{\denotevar{\delta}}

\newcommand\totnumof{N}

\newcommand\presof[1]{\pres_{#1}}
\newcommand\denotevar[1]{\overline{#1}}
\newcommand\nsvar[1]{\denotevar{s}_{#1}}
\newcommand\elevar[1]{\denotevar{e}_{#1}}

\newcommand\pclsvar[2]{\denotevar{#2}_{#1}}

\newcommand\numvar[1]{\denotevar{n}_{#1}}
\newcommand\numtypevar[2]{\denotevar{n}^{#2}_{#1}}
\newcommand\lnumvar[1]{\denotevar{N}_{#1}}
\newcommand\lnumtypevar[2]{\denotevar{N}^{#2}_{#1}}

\newcommand\csssimpres[2]{\ap{\mathrm{Pres}}{{#1}, {#2}}}

\newcommand\nomatch[3]{\ap{\mathrm{NoMatch}}{#1, #2, #3}}
\newcommand\astatevar[1]{\denotevar{\astate}_{#1}}
\newcommand\tranpres[1]{\ap{\mathrm{Tran}}{#1}}
\newcommand\tranprest[2]{\ap{\mathrm{Tran}}{#1, #2}}
\newcommand\shiftvartype[1]{\denotevar{\delta}_{#1}}
\newcommand\noatts[1]{\ap{\mathrm{NoAtts}}{#1}}

\newcommand\eword{\varepsilon}

\newcommand\mdash{\mbox{-}}
\newcommand\consset{C}
\newcommand\wordposvec{\vec{x}}
\newcommand\attvalbound{N}
\newcommand\numstates{N_s}
\newcommand\numcons{N_c}
\newcommand\aut{\mathcal{A}}
\newcommand\startchar{\hat{\phantom{x}}}
\newcommand\lastchar{\$}
\newcommand\nullch{0}
\newcommand\conssat{S}
\newcommand\consviol{V}
\newcommand\factor{\preceq}

\newcommand\attspres[2]{\ap{\mathrm{AttsPres}}{{#1}, {#2}}}
\newcommand\attspresns[4]{\ap{\mathrm{AttsPres}_{\qatt{#1}{#2}}}{{#3}, {#4}}}
\newcommand\attspresnulls[1]{\ap{\mathrm{Nulls}}{#1}}
\newcommand\wordpos[4]{x_{#3, #4}^{\qatt{#1}{#2}}}
\newcommand\consaut[1]{\aut_{#1}}
\newcommand\consprefixes[1]{\ap{\mathrm{Prefs}}{#1}}
\newcommand\atranch[1]{\xrightarrow{#1}}
\newcommand\consrunst[3]{\tup{#1, #2, #3}}

\newcommand\idcons{\consset_{\idatt}}
\newcommand\poscons{\consset_{\text{pos}}}
\newcommand\tset{\mathrm{Ts}}

\newcommand\btarget{b_{\mathrm{targ}}}
\newcommand\broot{b_{\mathrm{root}}}
\newcommand\bsibling{b_{\mathrm{sib}}}

\newcommand\totnumvar{\denotevar{N}}

\newcommand\attsof[1]{\ap{\text{Atts}}{#1}}
\newcommand\posof[1]{\ap{\text{Pos}}{#1}}
\newcommand\psof[1]{\ap{\text{Pseudo}}{#1}}
\newcommand\totnumtypevar[1]{\denotevar{N}_{#1}}

\newcommand\hasedge[1]{\ap{\mathrm{HasEdge}}{#1}}

\newcommand\vble[1]{\overline{#1}}

\newcommand\selpropord[1]{\rho_{#1}}

\newcommand\excludevble[1]{\vble{x}_{#1}}

\newcommand\orderablesub[2]{\ap{\mathrm{OrderableSub}}{#1, #2}}

\newcommand\orderablesubcand[3]{\ap{\mathrm{OrderableSub}}{#1, #2, #3}}

\newcommand\edgeslast[2]{\CSSedges^{#1}_{#2}}

\newcommand\propord[2]{\ll^{#1}_{#2}}

\newcommand\orderable[2]{\ap{\mathrm{Orderable}}{#1, #2}}

\newcommand\inpos{\vble{j}}

\newcommand\mbiclique{M}

\newcommand\nummbicliques{\mu}
\newcommand\bicliquevble{\vble{i}_M}
\newcommand\numexcludes{\chi}
\newcommand\propsSet{Y}
\newcommand\selsSet{X}
\newcommand\forbidden{\mathcal{F}}
\newcommand\forbiddenfst{\mathcal{F}^{1\text{st}}}

\newcommand\mbicliqueset{\mathcal{M}}
\newcommand\unorderablebicliques{\chi}
\newcommand\orderablebicliques{\Omega}

\newcommand\candnodes{\Delta}

\newenvironment{compactitem}{
    \begin{itemize}
}{
    \end{itemize}
}

\newcommand\softfmlabucket[2]{\fmla^{#1}_{#2}}

\newcommand\fmlanode[2]{\varfmla^{#1}_{#2}}

\newcommand\fmla{\varphi}
\newcommand\varfmla{\psi}
\newcommand\hardcons{\Pi_H}
\newcommand\softcons{\Pi_S}
\newcommand\wght{\omega}
\newcommand\fmlavalid{\fmla_{\mathrm{vld}}}
\newcommand\fmlaordering{\fmla_{\mathrm{ord}}}
\newcommand\softconsbucket{\softcons^{\bucket}}
\newcommand\softconssels{\softcons^{\mathrm{Sels}}}
\newcommand\softconsprops{\softcons^{\mathrm{Props}}}
\newcommand\textand{\text{ and }}

%% file: namedenvs.tex
\newenvironment{namedtheorem}[2]{%
    \expandafter\gdef\csname reftheorem#1\endcsname{%
        Theorem~\ref{#1} (#2)%
    }%
    \begin{theorem}[{#2}] \label{#1}%
}{%
    \end{theorem}%
}
\newcommand\reftheorem[1]{\expandafter\csname reftheorem#1\endcsname}

\newenvironment{namedlemma}[2]{%
    \expandafter\gdef\csname reflemma#1\endcsname{%
        Lemma~\ref{#1} (#2)%
    }%
    \begin{lemma}[{#2}] \label{#1}%
}{%
    \end{lemma}%
}
\newcommand\reflemma[1]{\expandafter\csname reflemma#1\endcsname}

\newcommand\refproperty[1]{\expandafter\csname refproperty#1\endcsname}

\newenvironment{nameddefinition}[2]{%
    \expandafter\gdef\csname refdefinition#1\endcsname{%
        Definition~\ref{#1} (#2)%
    }%
    \begin{definition}[{#2}] \label{#1}%
}{%
    \end{definition}%
}
\newcommand\refdefinition[1]{\expandafter\csname refdefinition#1\endcsname}

\newenvironment{namedproposition}[2]{%
    \expandafter\gdef\csname refproposition#1\endcsname{%
        Proposition~\ref{#1} (#2)%
    }%
    \begin{proposition}[{#2}] \label{#1}%
}{%
    \end{proposition}%
}
\newcommand\refproposition[1]{\expandafter\csname refproposition#1\endcsname}

%% file: numbers.tex
\newcommand\unoptimisedtimeouts{19}
\newcommand\numcompleted{43}
\newcommand\numtimedout{29}
\newcommand\numbenchmarks{72}
\newcommand\nobicliquesnumcompleted{45}
\newcommand\nobicliquesnumtimedout{27}
\newcommand\fullexcnumcompleted{36}
\newcommand\fullexcnumtimedout{36}
\newcommand\unlimbicliquesnumcompleted{42}
\newcommand\unlimbicliquesnumtimedout{30}
\newcommand\nothreadnumcompleted{40}
\newcommand\nothreadnumtimedout{32}
\newcommand\noordnumcompleted{50}
\newcommand\noordnumtimedout{22}
\newcommand\satcssmedian{3.79}
\newcommand\satcssthird{6.90}
\newcommand\miniuptomedian{3.29}
\newcommand\miniuptothird{5.45}
\newcommand\afterminimedian{4.70}
\newcommand\afterminithird{8.26}
\newcommand\cleancssmedian{3.03}
\newcommand\cleancssthird{5.26}
\newcommand\aftercleancssmedian{6.10}
\newcommand\aftercleancssthird{9.50}
\newcommand\gainsmedian{2.80}
\newcommand\gainsthird{5.03}
\newcommand\ratiomedian{48}
\newcommand\ratiothird{136}
\newcommand\pctzerobadbicliques{58.4}
\newcommand\meanbadbicliques{0.39}
\newcommand\maxpctbadbicliques{5.84}

%% file: new_abstract.tex
Minification is a widely-accepted technique which aims at reducing the size of
the code transmitted over the web. This paper concerns the problem of
semantic-preserving minification of Cascading Style Sheets (CSS) --- the de
facto language for styling web documents --- based on merging similar rules.

The cascading nature of CSS makes the semantics of CSS files sensitive to the
ordering of rules in the file. To automatically identify rule-merging
opportunities that best minimise file size, we reduce the rule-merging problem
to a problem concerning ``CSS-graphs'', i.e., node-weighted bipartite graphs
with a dependency ordering on the edges, where weights capture the number of
characters.

Constraint solving plays a key role in our approach. Transforming a CSS file into
a CSS-graph problem requires us to extract the dependency ordering on the edges
(an NP-hard problem), which requires us to solve the selector intersection
problem. To this end, we provide the first full formalisation of CSS3 selectors
(the most stable version of CSS) and reduce their selector intersection problem
to satisfiability of quantifier-free integer linear arithmetic, for which
highly-optimised SMT-solvers are available. To solve the above NP-hard graph
optimisation problem, we show how Max-SAT solvers can be effectively employed.
We have implemented our rule-merging algorithm, and tested it against
approximately $70$ real-world examples (including examples from each of the top 20 most popular websites).
We also used our benchmarks to compare our tool
against six well-known minifiers (which implement other optimisations).
Our experiments suggest that our tool produced larger savings. A substantially 
better minification rate was shown when our tool is used together with these 
minifiers.

%% file: intro.tex
\section{Introduction}

\emph{Minification} \cite[Chapter 12]{Souders-book} is a widely-accepted
technique in the web programming literature that aims at decreasing the size of
the code transmitted over the web, which can directly improve the response-time
performance of a website. Page load time is of course crucial for users'
experience, which impacts the performance of an online business
and is increasingly being included as a ranking factor by search engines
\cite{google-rank}.
Minification bears a resemblance to traditional code compilation.
In particular, it is applied only \emph{once} right before deploying the
website (therefore, its computation time does \emph{not} impact the page load
time). However, they differ in at least two ways. First,
the source and target languages for minification are the same (high-level)
languages.  The code to which minification can be applied is typically
JavaScript or CSS, but it can also be HTML, XML, SVG, etc.
Second, minification
applies various semantic-preserving
transformations with the objective of reducing the size of the code.

This paper concerns the problem of minifying CSS (Cascading Style Sheets),
which is the de facto language for styling web documents (HTML, XML, etc.)
as developed and maintained by the World Wide Web Constortium (W3C)~\cite{CSS}.
We will minify the CSS without reference to the documents it may be designed to style.
We refer the reader to Section~\ref{sec:related} for a discussion of document-independent and document-dependent related work.

A CSS file consists of a list of CSS \emph{rules}, each containing a
list of \emph{selectors} --- each selecting nodes in the \emph{Document Object
Model} (DOM), which is a tree structure representing the document --- and a
list of \emph{declarations}, each assigning values to selected nodes'
display attributes (e.g. \texttt{blue} to the property \texttt{color}).
Real-world CSS files can easily have many rules (in the order of
magnitude of 1000), e.g., see the statistics for popular sites \cite{website-ranking} on
\url{http://cssstats.com/}.
As Souders wrote in his book \cite{Souders-book}, which is widely regarded as
the bible of web performance optimisation:
\begin{quotation}
    The greatest potential for [CSS file] size savings comes from optimizing CSS
    --- merging identical classes, removing unused classes, etc.
    This is a complex problem, given the order-dependent nature of CSS (the
    essence of why it's called \emph{cascading}). This area warrants further
    research and tool development.
\end{quotation}
More and more CSS minifiers 
have been, and are being, developed. These include \texttt{YUI Compressor} \cite{yui},
\texttt{cssnano} \cite{cssnano}, \texttt{minify} \cite{minify},
\texttt{clean-css} \cite{cleancss},
\texttt{csso} \cite{csso}, and \texttt{cssmin} \cite{cssmin}, to name a few.
Such CSS minifiers apply various syntactic transformations, typically
including removing whitespace characters and comments, and using abbreviations
(e.g. \verb+#f60+ instead of \verb+#ff6600+ for the colour orange).

In this paper, we propose a new class of CSS 
transformations
based on
merging similar or duplicate rules in the CSS file
(thereby removing repetitions across multiple rules in
the file) that could reduce file size while preserving the rendering
information.
%
\OMIT{
Refactoring~\cite{refactoring} can be understood as the process of applying
multiple small-step
semantic-preserving transformations on a piece of code typically with
the goal of improving readability, design, and maintanability. In the context
of CSS, refactoring has been
studied by researchers and programmers with the goal of improving maintanability
and reducing the file size (a.k.a. \emph{minification}), e.g.,
see~\cite{Mesbah-refactoring,Geneves-style}.
In this paper, we consider CSS refactoring opportunities which
\emph{merge similar rules} in the CSS file (thereby, removing
repetitions across multiple rules in the file) without affecting how an
HTML document is rendered.
}
\OMIT{
These minification opportunities can in fact be construed as
\defn{rewrite rules} that can be applied in a nondeterministic manner (since
more than one
}
To illustrate the type of transformations that we focus on
in this paper (a more detailed
introduction is given in Section \ref{sec:example}), consider the following simple CSS file.
\begin{center}
\begin{minted}{css}
    #a { color:red; font-size:large }
    #c { color:green }
    #b { color:red; font-size:large }
\end{minted}
\end{center}
The selector \verb+#a+ selects a node with \emph{ID} \verb+a+ (if
it exists).
Note that the first and the third rules above have the same property
declarations. Since one ID can be associated with at most one node in a
DOM-tree, we
can merge these rules resulting in the following equivalent file.
\begin{center}
\begin{minted}{css}
    #a, #b { color:red; font-size:large }
    #c { color:green }
\end{minted}
\end{center}
Identifying such a rule-merging-based minification opportunity ---
which we shall call \defn{merging opportunity} for short --- in general is
non-trivial since a CSS file is sensitive to the ordering
of rules that may match the same node, i.e., the problem mentioned in Souders'
quote above. For example, let us assume that the three selectors
\verb+#a+, \verb+#b+, and \verb+#c+ in our CSS example instead were \verb+.a+,
\verb+.b+,
and \verb+.c+, respectively. The selector \verb+.a+ selects
nodes with \emph{class} \verb+a+. Since a class can be associated with
multiple nodes in a DOM-tree, the above merging opportunity could change how a
page is
displayed and therefore would not be valid, e.g., consider a page with
an element that has two classes, \texttt{.b} and \texttt{.c}.
This element would be
displayed as red by the original file (since red appears later), but as green by the file after applying
the transformation. Despite this, in this case we would still be able to
merge the
\emph{subrules} of the first and the third rules (associated with the
\verb+font-size+ property) resulting in the following smaller equivalent file:
\begin{center}
\begin{minted}{css}
    .a { color:red }
    .c { color:green }
    .b { color:red }
    .a, .b { font-size:large }
\end{minted}
\end{center}
\OMIT{
On the other hand, if \verb+.a+, \verb+.b+, and \verb+.c+ instead

\begin{center}
\begin{minted}{css}
    .a { color:red; font-size:large }
    .c { color:green }
    .b { color:red; font-size:large }
\end{minted}
\end{center}
One may be tempted to merge the first and the third rules, resulting
in the following file.
\begin{center}
\begin{minted}{css}
    .a, .b { color:red; font-size:large }
    .c { color:green }
\end{minted}
\end{center}
However, this could change how a page is displayed, e.g., consider a page with
an element that has the classes \texttt{.b} and \texttt{.c}, which is
displayed as red by the original file, but as green by the file after applying
refactoring. On the other hand, if \verb+.a+, \verb+.b+, and \verb+.c+ instead
were \verb+#a+ (i.e. \emph{the} node with id \verb+a+), \verb+#b+, and
\verb+#c+,
the aforementioned refactoring would be valid since each node in a document
can have at most one id.
}
This example suggests two important issues. First, identifying a
merging opportunity in a general way requires a means of
checking whether two given selectors may \emph{intersect}, i.e., select the
same node in \emph{some} document tree.
Second, a given CSS file could have a large number of merging opportunities at
the same time (the above example has at least two). Which one shall we pick?
There are multiple answers to this question. Note first that merging rules
can be iterated multiple times to obtain increasingly smaller CSS files.
In fact, we can define an optimisation problem which, given a CSS file, computes
a sequence of
applications of semantic-preserving merging rules that yields a \emph{globally
minimal}
CSS file. Finding a provably globally minimal CSS file is a computationally
difficult problem that seems well beyond the reach of current technologies
(constraint solving, and otherwise). Therefore, we propose to use the simple
\emph{greedy strategy}: while there is a merging opportunity that can make
the CSS file smaller, apply an \emph{optimal} one, i.e., a merging opportunity
that reduces the file size the most\footnote{Not to be confused with
the smallest CSS file that is equivalent with the original file, which in
general cannot be obtained by applying a single merging}.
There are now two problems to address. First, can we efficiently find an
optimal merging opportunity for a given CSS file? This paper provides a positive
answer to this question by exploiting a state-of-the-art constraint solving
technology. Second, there is a
potential issue of getting stuck at a \emph{local minimum} (as is
common with any optimisation method based on gradient descent). Does the greedy
strategy produce a meaningful space saving? As we will see in this paper,
the greedy approach could already produce space savings that are beyond
the reach of current CSS minification technologies.

\subsection{Contributions}
We first formulate a \emph{general class of semantic-preserving transformations}
on CSS files that captures the above notion of merging ``similar'' rules in a
CSS file. Such a program
transformation has a clean graph-theoretic formulation (see Section
\ref{sec:css2graph}). Loosely speaking, a CSS rule corresponds to a biclique
(complete bipartite graph) $B$, whose edges connect nodes representing
selectors and nodes representing property declarations.
Therefore, a CSS file $F$ corresponds to a sequence
of bicliques that covers all of the edges in the
bipartite graph $\CSSgraph$ which is constructed by taking the (non-disjoint)
union of all bicliques in $F$.
Due to the cascading nature of CSS, the file $F$ also gives rise to an
(implicit) ordering $\edgeOrder$ on the edges of $\CSSgraph$. Therefore, any
new CSS file $F'$ that we produce from $F$ must also be a valid covering of the
edges of $\CSSgraph$ and respect the order $\edgeOrder$. As we will see,
the above notion of merging opportunity can be defined as a pair $(\bucket,j)$ of
a new rule $\bucket$ and a position $j$ in the file, and that applying
this transformation entails inserting $\bucket$ in position $j$ and removing
all redundant nodes (i.e. either a selector or a property declaration) in
rules at position $i < j$.

Several questions remain. First is how to compute the edge order $\edgeOrder$.
The core difficulty of this problem is to determine
whether two CSS selectors can be matched by the same node in some document tree
(a.k.a.~the \emph{selector intersection problem}).
Second, among the multiple potential merging opportunities, how do we
automatically compute a rule-merging opportunity that best minimises the size
of the CSS file.
We provide solutions to these questions in this paper.

\paragraph{Computing the edge order $\edgeOrder$.} In order to handle
the selector intersection problem, we first provide a complete formalisation of
CSS3 selectors \cite{CSS3sel} (currently the most stable version of
CSS selectors). We then give a polynomial-time reduction from the selector
intersection problem to satisfiability over quantifier-free theory of
integer linear arithmetic, for which highly optimised SMT-solvers (e.g. Z3 \cite{Z3})
are available. To achieve this reduction, we provide a chain of polynomial-time
reductions. First, we develop a new class of
\defn{automata over data trees} \cite{data-aut-survey}, called \emph{CSS automata},
which can capture the expressivity of CSS selectors. This reduces the
selector intersection problem to the language intersection problem
for CSS automata.
Second, unlike the case for CSS selectors, the languages recognised by CSS
automata
enjoy closure under intersection. (\verb+.b .a+ and \verb+.c .a+ are
individually CSS selectors, however their conjunction is not a CSS selector\footnote{%
    The conjunction can be expressed by the \emph{selector group}
    \texttt{.b .c .a, .c .b .a, .b.c .a}.
}.)
This reduces language intersection of CSS automata to language non-emptiness
of CSS automata. To finish this chain of reductions, we provide a reduction
from the problem of language non-emptiness of CSS automata to satisfiability
over quantifier-free theory of integer linear arithmetic. The last reduction
is technically involved and requires insights from logic and automata theory,
which include several small model lemmas (e.g. the sufficiency of considering
trees with a small number of node labels that can be succinctly encoded).
This is despite the fact that CSS selectors may force the smallest satisfying
tree to be exponentially big.

\paragraph{Formulation of the ``best'' rule-merging opportunity and how to find
it.} Since our objective is to minimise file size, we may equip
the graph $\CSSgraph$ by a \emph{weight function} $\nodeWeight$
which maps each node to the number of characters used to define it (recall that
a node is either a selector or a property declaration). Hence, the function
$\nodeWeight$ allows us to define the size of a CSS file $F$ (i.e. a covering of
$\CSSgraph$ respecting $\edgeOrder$) by taking the sum of weights
$\nodeWeight(v)$ ranging over all nodes $v$ in $F$. The goal, therefore,
is to find a merging opportunity $(\bucket,j)$ of $F$ that produces $F'$ with a
minimum file size. We show how this problem can be formulated as a (partially
weighted) Max-SAT instance \cite{maxsat_benchmarks} in such a way that several existing
Max-SAT solvers (including Z3 \cite{Z3maxsat} and MaxRes \cite{maxres}) can
handle it efficiently.
This Max-SAT encoding is non-trivial: the naive encoding causes Max-SAT solvers to require prohibitively long run times even on small examples.
A naive encoding would allow the Max-SAT solver to consider any rule constructed from the nodes of the CSS file, and then contain clauses that prohibit edges that do not exist in the original CSS file (as these would introduce new styling properties).
Our encoding forces the Max-SAT solver to only consider rules that do not introduce new edges.
We do this by enumerating all \emph{maximal} bicliques in the graph $\CSSgraph$ (maximal with respect to set inclusion) and developing an encoding that allows the Max-SAT solver to only consider rules that are contained within a maximal biclique.
We employ the algorithm from \cite{K10} for enumerating all
maximal bicliques in a bipartite graph, which runs in time polynomial in the
size of the input and output. Therefore, to make this algorithm run efficiently,
the number of maximal bicliques in the graph $\CSSgraph$ cannot be very large.
Our benchmarking (using approximately $70$ real-world examples including CSS files from each of the top
20 websites \cite{website-ranking}) suggests that this is the case for graphs $\CSSgraph$
generated by CSS files (with the maximum ratio between the number of bicliques
and the number of rules being $2.05$).
Our experiments suggest that the
combination of the enumeration
algorithm of \cite{K10} and Z3 \cite{Z3maxsat} makes the problem of finding
the best merging opportunity for CSS files practically feasible.

\paragraph{Experiments} We have implemented our CSS minification algorithm in
the tool \satcss which
greedily searches for and applies the best merging opportunity
to a given CSS file until no more rule-merging can reduce file size.
The source code, experimental data, and a disk image is available on Figshare~\cite{HHL18}.
The source code is also available at the following URL.
\begin{center}
    \url{https://github.com/matthewhague/sat-css-tool}
\end{center}
Our tool utilises Z3 \cite{Z3,Z3maxsat} as a backend solver for Max-SAT and SMT
over integer linear arithmetic. We have tested our tool on around $70$ examples from
real-world websites (including examples from each of the top 20 websites~\cite{website-ranking})
with promising experimental results. We found that \satcss
(which only performs rule-merging) yields larger savings on our benchmarks
in comparison to six popular CSS minifiers~\cite{minifier-list}, which support many other
optimisations but not rule-merging.
More precisely, when run individually, \satcss reduced the file size by
a third quartile of \satcssthird\% and a median value of \satcssmedian\%.
The six mainstream minifiers achieved savings with
third quantiles and medians up to \miniuptothird\% and \miniuptomedian\%, respectively.
Since these are orthogonal minification opportunities, one might suspect that
applying these optimisations together could further improve the minification
performance, which is what we discover during our experiments.
More precisely, when we run our tool after running any one of these minifiers,
the third quartile of the savings can be increased to \afterminithird\% and the median to \afterminimedian\%.
The additional gains obtained by our tool on top of the six minifiers
(as a percentage of the original file size)
have a third quartile of \gainsthird\% and a median value of \gainsmedian\%,
\anthony{which supports the hypothesis that our optimisation method is
orthogonal to the optimisations performed by existing minifiers.}
Moreover, the ratios of the percentage of savings made by \satcss to
the percentage of savings made by the six minifiers have third quartiles of
at least \ratiothird\% and medians of at least \ratiomedian\%.
In fact, in the case of \texttt{cleancss} which has a third quartile saving of
\cleancssthird\% and median saving of \cleancssmedian\%, applying \satcss thereafter results in
a third quartile saving of \aftercleancssthird\% and median saving of \aftercleancssmedian\%.
These figures clearly indicate a
substantial proportion of extra space savings made by \satcss.
See Table~\ref{tab:percentile-tables} and Figure~\ref{fig:box-plots}
for more detailed statistics.
\subsection{Organisation}
Section~\ref{sec:example} provides a gentle and more detailed introduction (by
example) to the problem of CSS minification via rule-merging. Preliminaries
(notation, and basic terminologies) can be found in Section~\ref{sec:prelim}.
In Section~\ref{sec:css2graph} we formalise the rule-merging problem in terms
of what we will call CSS-graphs. 
In Section~\ref{sec:edgeOrder}, we provide a formalisation of CSS3 selectors
and an outline of an algorithm for solving their intersection problem, from
which we can extract the edge order of a CSS-graph
that is required when reducing the rule-merging problem to the
edge-covering problem of CSS-graphs.
Since the algorithm solving the selector intersection problem is rather
intricate, we dedicate one full section (Section \ref{sec:intersection}) for it.
In Section
\ref{sec:graph2maxsat} we show how Max-SAT solvers can be exploited to solve
the rule-merging problem of CSS-graphs.
We describe our experimental results in
Section~\ref{sec:experiments}. 
Note that a subset of Max-SAT instances generated using our tool was used in
MaxSAT Evaluation 2017 \cite{HH17}.
In Section~\ref{sec:related} we
give a detailed discussion of related work.
Finally, we conclude in Section~\ref{sec:conclusion}.
Additional details can be found in
the appendix.
The Python code and benchmarks for our tool have been included in the supplementary material, with a brief user guide.
A full artefact, including virtual machine image will be included when appropriate.
These are presently available from the URLs above.

%% file: css_by_example.tex
\section{CSS Rule-Merging: Example and Outline}
\label{sec:example}

In this section, we provide a gentler and more detailed introduction to CSS
minification via merging rules, 
while elaborating on the difficulties of the problem.
We also give a general overview of the algorithm, which may serve as a guide to the article.
We begin by giving a basic introduction to HTML and CSS.
Our formal models of HTML and CSS are given in Section~\ref{sec:edgeOrder}.

\subsection{HTML and CSS}
\label{sec:html-example}

In this section we give a simple example to illustrate HTML and CSS.
Note, in this article our analysis takes as input a CSS file only.
We cover HTML here to aid the description of CSS.

An HTML document is given in Figure~\ref{fig:simple-html}.
The format is similar to (but not precisely) XML, and is organised into a tree structure.
The root node is the \texttt{html} \emph{element} (or \emph{tag}), and its two children are the \texttt{head} and \texttt{body} nodes.
These contain page header information (particularly the title), which is not directly displayed, and the main page contents respectively.

The \texttt{body} node contains two children which represent the page content.
The first child is a \texttt{div} element, which is used to group together parts of the page for the convenience of the developer.
In this case, the developer has chosen to collect all parts of the page containing the displayed heading of the page into one \texttt{div}.
This \texttt{div} has an \emph{id} which takes the value \texttt{heading}.
Values specified in this way are called \emph{attributes}.
\defn{ID} attributes uniquely identify a given node in the HTML, that is, no two nodes should have the same ID.
The \texttt{div} has two children, which are an image (\texttt{img}) element and a textual heading (\texttt{h1}).
The source of the image is given by \texttt{src}, which is also an attribute.
The image also has a \texttt{class} attribute which takes the value \texttt{banner}.
The class attribute is used to distinguish nodes which have the same tag.
For example, the developer may want an image with the class \texttt{banner} to be displayed differently to an image with the class \texttt{illustration}.

The second child of the \texttt{body} node is an unordered list (\texttt{ul}).
This lists contains two food items, with appropriate ID and class attributes.

\begin{figure}
    \begin{minted}{html}
        <html>
            <head><title>Shopping List</title></head>
            <body>
                <div id="heading">
                    <img class="banner" src="banner.jpg"/>
                    <h1>An Example HTML Document</h1>
                </div>
                <ul>
                    <li id="apple" class="fruit">Apple</li>
                    <li id="broccoli">Broccoli</li>
                </ul>
            </body>
        </html>
    \end{minted}
    \caption{\label{fig:simple-html}A simple HTML document.}
\end{figure}

To emphasize the tree structure, we give in Figure~\ref{fig:html-tree} the same HTML document drawn as a tree.
We will give a refined version of this diagram in Section~\ref{sec:edgeOrder} when defining DOM trees formally.

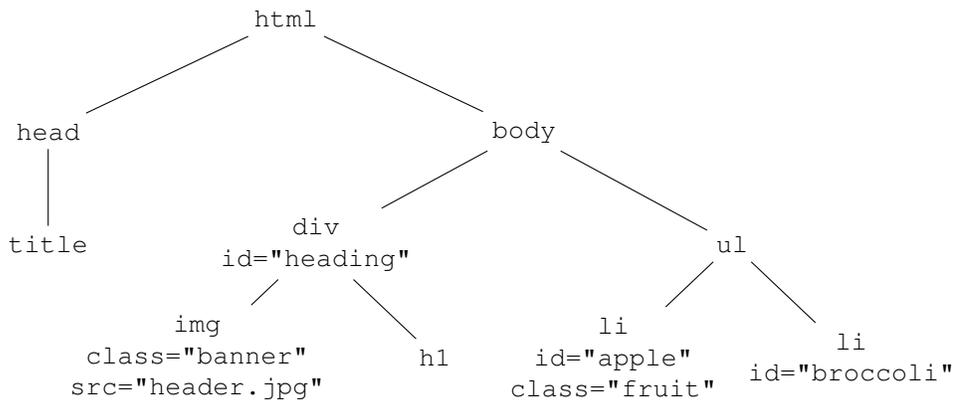
\begin{figure}
    \begin{tikzpicture}[level 1/.style={sibling distance=40ex},
                        level 2/.style={sibling distance=35ex},
                        level 3/.style={sibling distance=20ex},
                        every node/.style={align=center}]
        \node {\texttt{html}}
            child {node {\texttt{head}}
                child {node {\texttt{title}}}
            }
            child {node {\texttt{body}}
                child {node {\texttt{div} \\
                             \texttt{id="heading"}}
                    child  {node {\texttt{img} \\
                                  \texttt{class="banner"} \\
                                  \texttt{src="header.jpg"}}}
                    child  {node {\texttt{h1}}}
                }
                child {node {\texttt{ul}}
                    child {node {\texttt{li} \\
                                 \texttt{id="apple"} \\
                                 \texttt{class="fruit"}}}
                    child {node {\texttt{li} \\
                                 \texttt{id="broccoli"}}}
                }
            };
    \end{tikzpicture}
    \caption{\label{fig:html-tree}The HTML file in Figure~\ref{fig:simple-html} drawn as a tree.}
\end{figure}

A CSS file consists of a sequence of \defn{rules}, each of the form
\begin{center}
    \emph{selectors} \{ \emph{declarations} \}
\end{center}
where \emph{selectors} contains one or more (node-)selectors (separated by
commas) and
\emph{declarations} contains a sequence of (visual) property declarations
(separated by semicolons).
An example CSS file containing two rules is given in Figure~\ref{fig:css-example}.

\begin{figure}
    \begin{minted}{css}
        img.banner { width: 100% }
        #heading h1 { font-size: 30pt; font-weight: bold }
    \end{minted}
    \caption{\label{fig:css-example}An example CSS file.}
\end{figure}

The semantics of a rule is simple: if a node can be matched by at least one
of the selectors, then label the node by all the visual properties
in the rule.
Both rules in Figure~\ref{fig:css-example} have one selector.
The selector in the first rule \texttt{img.banner} matches all \texttt{img} elements which have the class \texttt{banner}.
The notation \texttt{.} is used to indicate a class.
Thus, this rule will match the \texttt{img} node in our example document and will specify that the image is to be wide.

The second selector demonstrates that selectors may reason about the structure of the tree.
The \texttt{\#heading} means that the match should begin at the node with ID \texttt{heading}.
Then, the space means the match should move to any descendent of the \texttt{\#heading} node, i.e., any node contained within the \texttt{div} with ID \texttt{heading} in our example page.
Next, the \texttt{h1} is used to choose \texttt{h1} elements.
Thus, this selector matches any node with element \texttt{h1} that occurs below a node with ID \texttt{heading}.
The text of any such match should be rendered in a $30$pt bold font.

There are many other features of selectors.
These are described in full in Section~\ref{sec:edgeOrder}.

\subsection{CSS Rule-Merging by Example}

We will now discuss more advanced features of CSS files and give an example of 
rule-merging.
Figure \ref{fig:simpleEx} contains a simple CSS file (with four rules).

\begin{figure}[h]
    \begin{mdframed}[style=codebox]
    \begin{minipage}{.75\linewidth}
    \begin{center}
    \begin{minted}{css}
#apple { color:blue; font-size:small }
.fruit, #broccoli { color:red; font-size:large }
#orange { color:blue }
#tomato { color:red; font-size:large;
          background-color:lightblue }
    \end{minted}
    \end{center}
    \end{minipage}
    \end{mdframed}
    \caption{A simple example of a CSS file.\label{fig:simpleEx}}
\end{figure}

The sequence of rules in a file is applied to a node $v$
in a ``cascading'' fashion (hence the name Cascading Style Sheets). That is,
read the rules from top to bottom and
check if $v$ can be matched by at least one selector in the rule.
If so, assign the properties to $v$ (perhaps overriding previously
assigned properties, e.g., color) provided that
the selectors matching $v$ in the current rule have higher ``values'' (a.k.a.
\defn{specificities}) than the
selectors previously matching $v$. Intuitively, the specificity
of a selector \cite{CSS3sel}
can be calculated by taking a weighted sum of the number of
classes, IDs, tag names, etc. in the selector, where IDs have higher
weights than classes, which in turn have higher weights than tag names.

For example, let us apply this CSS file to a node matching \texttt{.fruit}
and \texttt{\#apple}. In this case, the selectors in the first
(\texttt{\#apple}) and the second rules (\texttt{.fruit}) are applicable.
However, since \texttt{\#apple}
has a higher specificity than \texttt{.fruit}, the node gets labels
\texttt{color:blue} and \texttt{font-size:small}.
%

Two syntactically different CSS files could be ``semantically equivalent'' in
the sense that, on each DOM tree $T$, both CSS files will precisely yield the
same
tree $T'$ (which annotates $T$ with visual information). In this paper,
we only look at semantically equivalent CSS files that are obtained by
merging similar rules.
Given a CSS rule $R$, we say that it is \defn{subsumed} by a CSS file $F$
if each possible selector/property combination in $R$ occurs in some rule in
$F$. Notice that a rule can be subsumed in a CSS file $F$ without occurring
in $F$ as one of the rules.
For example, the rule $R_1$
\begin{center}
\begin{minted}{css}
    .fruit, #broccoli, #tomato { color:red; font-size:large }
\end{minted}
\end{center}
is subsumed in our CSS file example (not so if
\texttt{background-color:lightblue} were added to $R_1$). A \defn{(rule-)merging
opportunity} consists
of a CSS rule subsumed in the CSS file and a position in the
file to insert the rule into.
An example of a merging opportunity in our CSS file example
is the rule $R_1$ and the end of the file as the insertion position. This
results in a new (bigger) CSS file show in Figure~\ref{fig:bigger-css}.

\begin{figure}
    \begin{subfigure}{\linewidth}
        \begin{mdframed}[style=codebox]
        \begin{minipage}{.75\linewidth}
        \begin{center}
        \begin{minted}{css}
#apple { color:blue; font-size:small }
.fruit, #broccoli { color:red; font-size:large }
#orange { color:blue }
#tomato { color:red; font-size:large;
          background-color:lightblue }
.fruit, #broccoli, #tomato { color:red; font-size:large }
        \end{minted}
        \end{center}
        \end{minipage}
        \end{mdframed}
        \caption{\label{fig:bigger-css}The CSS file in Figure~\ref{fig:simpleEx} with a rule inserted at the end.}
    \end{subfigure}
    \begin{subfigure}{\linewidth}
        \vspace{2ex}
        \begin{mdframed}[style=codebox]
        \begin{minipage}{.75\linewidth}
        \begin{center}
        \begin{minted}{css}
#apple { color:blue; font-size:small }
#orange { color:blue }
#tomato { background-color:lightblue }
.fruit, #broccoli, #tomato { color:red; font-size:large }
        \end{minted}
        \end{center}
        \end{minipage}
        \end{mdframed}
        \caption{\label{fig:trimmed-css}The result of trimming the CSS file in Figure~\ref{fig:bigger-css}.}
    \end{subfigure}
    \caption{\label{fig:trim-example}An example of insertion and trimming.}
\end{figure}

We then ``trim'' this resulting CSS file by eliminating ``redundant''
subrules. For example, the second rule is redundant since it is completely
contained in the new rule $R_1$. Also, notice that the subrule:
\begin{center}
\begin{minted}{css}
    #tomato { color:red; font-size:large }
\end{minted}
\end{center}
of the fourth rule in the original file is also completely redundant since it
is contained in $R_1$. Removing these, we obtain the trimmed version of
the CSS file, which is shown in Figure~\ref{fig:trimmed-css}.
This new CSS file contains fewer bytes.

We may apply another round of rule-merging by
inserting the rule
\begin{center}
\begin{minted}{css}
    #apple, #orange { color:blue }
\end{minted}
\end{center}
at the end
of the second rule (and trim). The resulting file is even smaller.
Computing such merging opportunities
(i.e.\ which yields maximum space saving) is difficult.
\OMIT{
We provide a clean formalisation of the problem in terms of node-weighted
edge-ordered bipartite graph covering in Section~\ref{sec:css2graph}, and show
how Max-SAT solvers can be employed to effectively tackle this problem in
Section~\ref{sec:graph2maxsat}.
}

Two remarks are in order. Not all merging opportunities yield
a smaller CSS file. For example, if \texttt{\#tomato} is replaced by the fruit
vegetable
\begin{center}
    \texttt{\#vigna\_unguiculata\_subsp\_sesquipedalis}
\end{center}
the first
merging opportunity above would have resulted in a larger file, which we should
\emph{not} apply. The second remark is related to the issue of
\defn{order dependency}. Suppose we add the selector \texttt{.vegetable} to
the third rule in the original file, resulting in the following rule
\begin{center}
\begin{minted}{css}
    .vegetable, #orange { color:blue }
\end{minted}
\end{center}
Any node 
labeled by
\texttt{.fruit} and \texttt{.vegetable} but no IDs will be assigned
\texttt{color:blue}. However, using the first merging opportunity
above yields
\OMIT{
the file:
\begin{center}
\begin{minted}{css}
    #apple { color:blue; font-size:small }
    #orange, .vegetable { color:blue }
    #tomato { background-color:lightblue }
    .fruit, #broccoli, #tomato { color:red; font-size:large }
\end{minted}
\end{center}
}
a CSS file
which assigns \texttt{color:red} to fruit vegetables, i.e.,
\emph{not} equivalent to the original file.
To tackle the issue of order
dependency, we need to account for selectors specificity and
whether
two selectors with the same specificity
can \defn{intersect} (i.e.\ can be matched by a node in some DOM tree).
Although for our examples the problem of selector intersection is quite obvious,
this is not the case for CSS selectors in general.
For example, the following two selectors are from a real-world CSS example found on The Guardian website.
\begin{verbatim}
    .commercial--masterclasses .lineitem:nth-child(4)
    .commercial--soulmates:nth-child(n+3)
\end{verbatim}
Interested readers unfamiliar with CSS may refer to Section~\ref{sec:css-selectors-def} for a complete definition of selectors.
The key feature is the use of \texttt{:nth-child}.
In the first selector it is used to only match nodes that are the fourth child of some node.
For a node to match the second selector, there must be some $n \geq 0$ such that the node is the $(n + 3)$th child of some node.
I.e.~the third, fourth, fifth, etc.
These two selectors have a non-empty intersection, while changing \texttt{n+3} to \texttt{2n+3} would yield the
intersection empty!
\OMIT{
To make matters worse, there are
no sufficiently complete formalisations of CSS selectors, let alone
automatic methods for checking selector intersection. We address these problems
in Section \ref{sec:edgeOrder}
by providing a complete formalisation of the CSS3 selector language and an
efficient reduction of the intersection problem to SMT over integer linear
arithmetic.
}

\subsection{CSS Rule-Merging Outline}

The component parts of our algorithm are given in the schematic diagram in Figure~\ref{fig:schematic}.

From an input CSS file, the first step is to construct a formal model of the CSS that we can manipulate.
This involves extracting the selector and property declaration pairs that appear in the file.
Then the edge ordering, which records the order in which selector/declaration pairs should appear, needs to be built.
To compute the edge ordering, it is important to be able to test whether two selectors may match the same node in some document tree.
Thus, computing the edge ordering requires an algorithm for computing whether the intersection of two selectors is empty.
This process is described in detail in Section~\ref{sec:edgeOrder} and Section~\ref{sec:intersection}.

Once a model has been constructed, it is possible to systematically search for
semantics-preserving rule-merging opportunities that can be applied to the file to reduce its overall size.
The search is formulated as a MaxSAT problem and a MaxSAT solver is used to find
the merging opportunity that represents the largest size saving.
This encoding is described in Section~\ref{sec:graph2maxsat}.
If a merging opportunity is found, it is applied and the search begins for 
another merging opportunity.
If none is found, the minimised CSS file is output.

\begin{figure}
\centering
    \tikzstyle{decision} = [diamond, draw, text width=4.5em, text badly centered]
    \tikzstyle{block} = [rectangle, draw, text width=5em, text centered, rounded corners, minimum height=4em]
    \tikzstyle{line} = [draw, -latex']
    \tikzstyle{cloud} = [draw, ellipse, minimum height=2em]

    \begin{tikzpicture}[node distance=.75cm and 1.5cm, auto]
        \node [cloud] (cssfile) {CSS File};
        \node [block, below=of cssfile] (buildmodel) {Build CSS Model};
        \node [block, below=of buildmodel] (findmerger) {Find Merging Opp.};
        \node [block, right=of buildmodel] (edgeorder) {Construct Edge Order};
        \node [block, right=of edgeorder] (intersection) {Selector Intersection};
        \node [decision, below=of findmerger] (isgood) {Merging Opp. Exists?};
        \node [block, right=of findmerger] (maxsat) {MaxSAT};
        \node [block, left=of isgood] (apply) {Apply Merging Opp.};
        \node [cloud, below=of isgood] (output) {Minified CSS File};
        \path [line] (cssfile) -- (buildmodel);
        \path [line,dashed] (edgeorder) -- node {uses} (buildmodel);
        \path [line,dashed] (intersection) -- node {uses} (edgeorder);
        \path [line] (buildmodel) -- (findmerger);
        \path [line,dashed] (maxsat) -- node {uses} (findmerger);
        \path [line] (findmerger) -- (isgood);
        \path [line] (isgood) -- node {yes} (apply);
        \path [line] (isgood) -- node {no} (output);
        \path [line] (apply) |- (findmerger);
    \end{tikzpicture}
    \caption{\label{fig:schematic} A Schematic Diagram of Our Approach}
\end{figure}
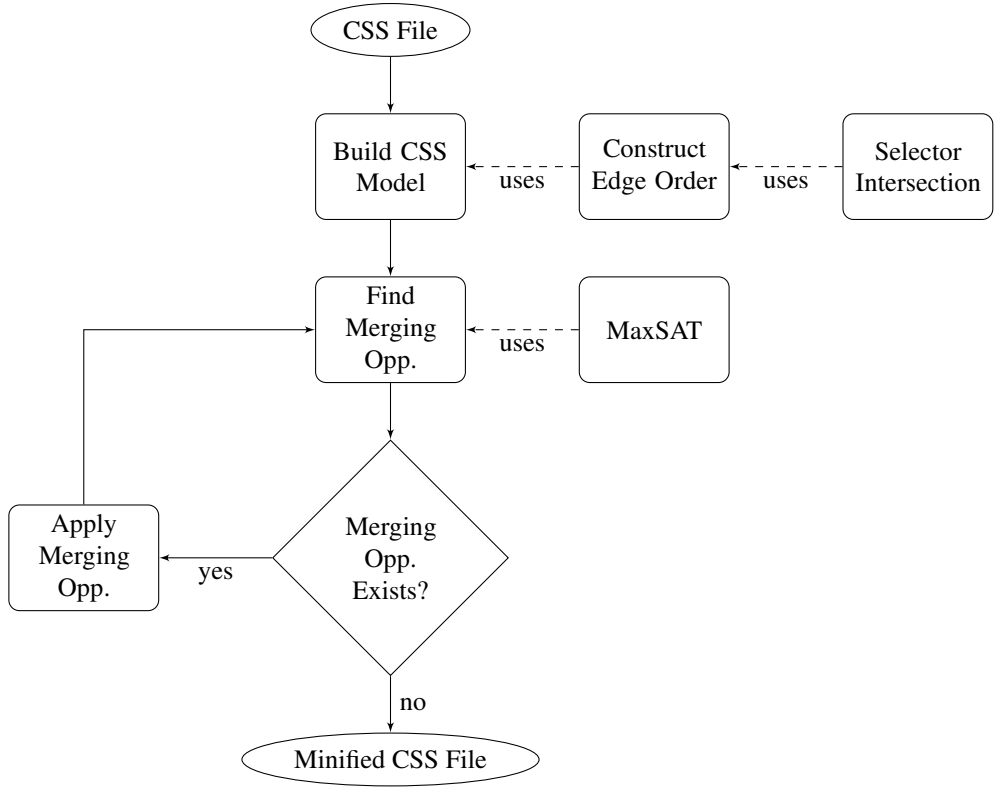

%% file: prelim.tex
\section{Preliminaries}
\label{sec:prelim}

\subsection{Maths}
As usual, $\Z$ denotes the set of all integers. We use $\N$ to denote the set
$\{0,1,\ldots,\}$ of all natural numbers. Let $\posN = \N \setminus\{0\}$
denote the set of all positive integers.
For an integer $x$ we define $\abs{x}$ to be the absolute value of $x$.
For two integers $i,j$, we write $[i,j]$ to denote the set $\{i,\ldots,j\}$.
Similar notations for open intervals  will also be used for
integers (e.g. $(i,j)$ to mean $\{i+1,\ldots,j-1\}$).
For a set $S$, we will write $S^*$ (resp.~$S^+$)
to denote the set of sequences (resp.~non-empty sequences) of elements from
$S$. When the meaning is clear, if $S$ is a singleton $\{s\}$, we will denote
$\{s\}^*$ (resp.~$\{s\}^+$) by $s^*$ (resp.~$s^+$).
Given a (finite) sequence $\Seq = s_1,\ldots,s_n$, $i \in [0,n]$, and a
new element $s$, we write $\inSeq{\Seq}{i}{s}$ to denote the new sequence
$s_1,\ldots,s_i,s,s_{i+1},\ldots,s_n$, i.e., inserting the element $s$ right
after the position $i$ in $\Seq$.

\subsection{Trees}
\label{sec:prelim-trees}

We review a standard formal definition of
\emph{(rooted) unranked ordered trees}~\cite{GS97,Libkin-survey,Neven-survey}
from the database research community, which use it to model XML.
We will use this in Section~\ref{sec:edgeOrder} to define document trees.
``Unranked'' means that there is no restriction on the number of
children of a node, while ``ordered'' means that the children of each node are linearly
ordered from the left-most child to the right-most child.
An unranked ordered tree consists of a tree domain and a labelling, which we define below.

We identify a node by the unique path in the tree to the given node.
A \defn{tree domain} defines the set of nodes in the tree and is a set of sequences of natural numbers
$\treedom \subseteq \brac{\posN}^\ast$.
The empty sequence is the root node.
The sequence $1$ would be the first child of the root node, while $2$ would be the second child.
The sequence $21$ would denote the first child of the second child of the root node, and so on.
The tree in Figure~\ref{fig:html-tree} has
$\treedom = \set{\eseq, 1, 2, 11, 21, 22, 211, 212, 221, 222}$.

We require that $\treedom$ is both \defn{prefix-closed} and \defn{preceding-sibling closed}.
By prefix-closed we formally mean
$\node\treedir \in \treedom$
implies
$\node \in \treedom$;
this says that the parent of each node is also a node in the tree.
By preceding-sibling closed we formally mean
$\node\treedir \in \treedom$
implies
$\node\treedir' \in \treedom$
for all $\treedir' < \treedir$;
for example, this means that if a node has a second child, it also has a first.
Observe, we write $\node$ for a tree node (element of $\treedom$) and $\treedir$ for an element of $\posN$.

Our trees will also be labelled by items such as element names and attribute values.
In general, a \defn{$\treelabs$-labelled tree} is a pair
$\tree = \tup{\treedom, \treelab}$
where
    $\treedom$ is a tree domain, and
    $\treelab : \treedom \rightarrow \treelabs$ is a labelling function of the nodes of $\tree$ with items from a set of labels $\treelabs$.
A simple encoding of the tree in Figure~\ref{fig:html-tree} will have the labelling
\[
    \begin{array}{rcl}
        \ap{\treelab}{\eseq} &=& \text{\texttt{html}} \\
        \ap{\treelab}{1} &=& \text{\texttt{head}} \\
        \ap{\treelab}{2} &=& \text{\texttt{body}} \\
        \ap{\treelab}{11} &=& \text{\texttt{title}} \\
        \ap{\treelab}{21} &=& \text{\texttt{div id="heading"}} \\
        \ap{\treelab}{22} &=& \text{\texttt{ul}} \\
        \ap{\treelab}{211} &=& \text{\texttt{img class="banner" src="header.jpg"}} \\
        \ap{\treelab}{212} &=& \text{\texttt{h1}} \\
        \ap{\treelab}{221} &=& \text{\texttt{li id="apple" class="fruit"}} \\
        \ap{\treelab}{222} &=& \text{\texttt{li id="broccoli"}} \ .
    \end{array}
\]
Note, in Section~\ref{sec:edgeOrder}, we will use a more complex labelling of DOM trees to fully capture all required features.

Next we recall terminologies for relationships between nodes in trees.
To avoid notational clutter, we deliberately choose notation that resembles the
syntax of CSS, which we define in Section~\ref{sec:edgeOrder}.
In the following, take $\node, \node' \in \treedom$.
We write $\node \treedescendant \node'$
if $\node$ is a (strict) ancestor of $\node'$, i.e., there is some
$\node'' \in \posN^+$ such that $\node' = \node\node''$.
We write $\node \treechild \node'$ if $\node$ is the parent of $\node'$, i.e.,
there is some $\treedir \in \posN$ such that $\node' = \node \treedir$.
We write $\node \treeneighbour \node'$
if $\node$ is the direct preceding sibling of $\node'$, i.e., there is some
$\node''$ and $\treedir \in \posN$
such that
$\node = \node'' (\treedir - 1)$
and
$\node' = \node'' \treedir$.
We write
$\node \treesibling \node'$
if $\node$ is a preceding sibling of $\node'$, i.e.,
there is some $\node''$ and
$\treedir, \treedir' \in \posN$
with
$\treedir < \treedir'$
such that
$\node = \node'' \treedir$
and
$\node' = \node'' \treedir'$.

\OMIT{
A tree domain is a set
$\treedom \subseteq \brac{\N \setminus \set{0}}^\ast$
that is both prefix-closed and preceding-sibling closed.
That is
$\node\treedir \in \treedom$ implies both
$\node \in \treedom$,
and
$\node\treedir' \in \treedom$
for all $\treedir' < \treedir$.

A $\treelabs$-labelled tree is a pair
$\tree = \tup{\treedom, \treelab}$
where
    $\treedom$ is a tree domain, and
    $\treelab : \treedom \rightarrow \treelabs$ is a labelling function of the nodes of $\tree$.

In the following, take $\node, \node' \in \treedom$.
We write
$\node \treedescendant \node'$
if $\node$ is a (strict) ancestor of $\node'$.
That is, there is some
$\node'' \in \N^+$
such that
$\node' = \node\node''$.
We write
$\node \treechild \node'$
if $\node$ is the parent of $\node'$.
That is, there is some
$\treedir \in \N$
such that
$\node' = \node \treedir$.
We write
$\node \treeneighbour \node'$
if $\node$ is the direct previous sibling of $\node'$.
That is, there is some $\node''$ and
$\treedir \in \N$
such that
$\node = \node'' \treedir$
and
$\node' = \node'' (\treedir - 1)$.
We write
$\node \treesibling \node'$
if $\node$ is a previous sibling of $\node'$.
That is, there is some $\node''$ and
$\treedir, \treedir' \in \N$
with
$\treedir < \treedir'$
such that
$\node = \node'' \treedir'$
and
$\node' = \node'' \treedir$.
}

\subsection{Max-SAT}
In this paper,
we will reduce the problem of identifying an optimal merging opportunity to partial weighted Max-SAT \cite{maxsat_benchmarks}.
Partial weighted
Max-SAT formulas are boolean formulas in CNF with \emph{hard constraints} (a.k.a.\ clauses that must be satisfied) and \emph{soft constraints} (a.k.a. clauses that may be violated with a specified cost or weight).
A minimal-cost solution is the goal.
Note that our clauses will not be given in CNF, but standard
satisfiability-preserving conversions to CNF exist (e.g.\ see \cite{Bradley-book})
which straightforwardly extend to partial weighted Max-SAT.

We will present Max-SAT problems in the following form
$(\hardcons, \softcons)$
where
\begin{itemize}
\item
    $\hardcons$ are the hard constraints
    -- that is, a set of boolean formulas that \emph{must} be satisfied --
    and
\item
    $\softcons$ are the soft constraints
    -- that is a set of pairs
       $(\fmla, \wght)$
       where $\fmla$ is a boolean formula and $\wght \in \N$ is the weight of the constraint.
\end{itemize}
Intuitively, the weight of a soft constraint is the cost of not satisfying the constraint.
The partial weighted Max-SAT problem is to find an assignment to the boolean variables that satisfies all hard constraints and minimises the sum of the weights of unsatisfied soft constraints.

\subsection{Existential Presburger Arithmetic}

In this paper, we present several encodings into \emph{existential Presburger arithmetic}, also known as the \emph{quantifier-free theory of integer linear arithmetic}.
Here, we use extended existential Presburger formulas
$\exists x_1,\ldots,x_k.\varphi$
where $\varphi$ is a boolean combination of expressions
$\sum_{i=1}^k a_ix_i \sim b$
for constants
$a_1,\ldots,a_k,b\in \Z$
and
$\sim \in \set{\leq,\geq,<,>,=}$
with constants represented in binary.
A formula is satisfiable if there is an assignment of a non-negative integer to each variable
$x_1, \ldots, x_k$
such that $\varphi$ is satisfied.
For example, a simple existential Presburger formula is shown below
\[
    \exists x, y, z\ .\ 0 > 2y + z - x \land 0 > z - y
\]
which is satisfied by any assignment to the variables $x$, $y$, and $z$ such that
$x < 2y + z$ and $y > z$.
The assignment $x = 2, y = 3, z = 0$ is one such satisfying assignment.
Note, when writing existential Presburger formulas, we will allow formulas that do not strictly match the format specified above.
This is to aid readability and we will always be able to rewrite the formulas into the correct format.
The above example formula may be written
\[
   \exists x, y, z\ .\ x < 2y + z \land y < z \ .
\]

It is well-known that satisfiability of existential Presburger formulas is $\NP$-complete even with the above extensions (cf.~\cite{Sca84}).
Such problems can be solved efficiently by SMT solvers such as Z3.

%% file: css2graph.tex
\section{Formal definition of CSS rule-merging and minification}
\label{sec:css2graph}

The semantics of a CSS file can be formally modelled as a CSS-graph.
A \defn{CSS-graph} is a 5-tuple $\CSSgraph = \inCSSgraph$, where
$(\selNodes,\propNodes,\CSSedges)$ is a bipartite graph (i.e.\ the set
$\CSSvertices$ of vertices is partitioned into two sets $\selNodes$ and
$\propNodes$ with $\CSSedges \subseteq \selNodes \times \propNodes$),
$\edgeOrder\ \subseteq \CSSedges \times \CSSedges$ gives the order dependency
on the edges, and $\nodeWeight: \selNodes \cup \propNodes \to \posZ$ is a weight function
on the set of vertices. Vertices in $\selNodes$ are called \defn{selectors},
whereas vertices in $\propNodes$ are called \defn{properties}.
For example,
the CSS graph corresponding to the CSS file in Figure~\ref{fig:simpleEx} with the selector \texttt{.vegetable} added to the third rule is the following bipartite graph
\begin{center}
\bigskip
\begin{psmatrix}[rowsep=0.5ex]
    \rnode[r]{apple}{\texttt{\#apple}}  &  \rnode[l]{blue}{\texttt{color:blue}} \\
    \rnode[r]{fruit}{\texttt{.fruit}}  &  \rnode[l]{small}{\texttt{font-size:small}} \\
    \rnode[r]{broccoli}{\texttt{\#broccoli}}  &  \rnode[l]{red}{\texttt{color:red}} \\
    \rnode[r]{orange}{\texttt{\#orange}}  &  \rnode[l]{large}{\texttt{font-size:large}} \\
    \rnode[r]{vegetable}{\texttt{.vegetable}}  &  \rnode[l]{lightblue}{\texttt{background-color:lightblue}} \\
    \rnode[r]{tomato}{\texttt{\#tomato}}
    \ncline{-}{apple}{blue}
    \ncline{-}{apple}{small}
    \ncline{-}{fruit}{red}
    \ncline{-}{fruit}{large}
    \ncline{-}{broccoli}{red}
    \ncline{-}{broccoli}{large}
    \ncline{-}{orange}{blue}
    \ncline{-}{vegetable}{blue}
    \ncline{-}{tomato}{red}
    \ncline{-}{tomato}{large}
    \ncline{-}{tomato}{lightblue}
\end{psmatrix}
\bigskip
\end{center}
such that the weight of a node is $1+(\text{length of the text})$, e.g.,
$\nodeWeight(\text{\texttt{\#orange}}) = 1+7 = 8$. The reason for the extra $+1$ is
to account for the selector/property separators (i.e.\ commas or semi-colons),
as well as the character `\verb+{+' (resp.~`\verb+}+') at the end of the
sequence
of selectors (resp.~properties).
That is, in a rule, selectors are followed by a comma if another selector follows, or `\verb+{+' if it is the last selector, and properties are followed by a semi-color if another property follows, or `\verb+}+' if it is the last property declaration.
We refer to $\edgeOrder$ as the \defn{edge order} and it intuitively states that one edge should appear strictly before the other in any CSS file represented by the graph.
In this case we have
$(\texttt{.fruit}, \texttt{color:red}) \edgeOrder (\texttt{.vegetable}, \texttt{color:blue})$
because any node labeled by \texttt{.fruit} and \texttt{.vegetable} but no IDs should be assigned the property \texttt{color:blue}.
There are no other orderings since each node can have at most one ID\footnote{Strictly speaking,
this is only true if we are only dealing with namespace \texttt{html} (which
is the case almost always in web programming and so is a reasonable assumption
unless the user specifies otherwise). A node could have multiple IDs, each with
a different namespace. See Section \ref{sec:edgeOrder}.}
and \texttt{.fruit} and \texttt{.vegetable} are the selectors of the lowest specificity in the file.
More details on how to compute $\edgeOrder$ from a CSS file are given in Section~\ref{sec:edge-order-from-css}.

A \defn{biclique} in $\CSSgraph$ is a complete bipartite subgraph, i.e.,
a pair $\biclique = \inbiclique$ of a nonempty set $\selsBucket
\subseteq \selNodes$ of selectors and a nonempty set $\propsSet
\subseteq \propNodes$ of properties such that $\selsBucket \times \propsSet
\subseteq \CSSedges$ (i.e. each pair of a selector and a property in the rule
is an edge).
A \defn{(CSS) rule} is a pair $\bucket = (\biclique,\propOrder)$ of a
biclique and a total order $\propOrder$ on the set of properties.
The reason for the order on the properties, but not on the selectors, is
illustrated by the following example of a CSS rule:
\begin{center}
    \begin{minted}{css}
        .a, .b { color:red; color:rgba(255,0,0,0.5) }
    \end{minted}
\end{center}
That is, nodes matching \texttt{.a} \emph{or} \texttt{.b}
are assigned a semi-transparent red with solid red being defined as a
\emph{fallback} when the semi-transparent red is not supported by the document
reader. Therefore, swapping the order of the properties changes the semantics
of the rule, but swapping the order of the selectors does not.
We will often denote a rule as $(\selsBucket,\propsBucket)$ where
$\propsBucket = \{p_i\}_{i=1}^m$
if
$\propsSet = \{p_1,\ldots,p_m\}$ and $p_1 \propOrder \cdots \propOrder p_m$.

A \defn{covering} $\covering$ of $\CSSgraph$ is a sequence of rules that
\emph{covers} $\CSSgraph$ (i.e. the union of all the edges in $\covering$
equals $\CSSedges$). Given an edge $e \in \CSSedges$, the \defn{index}
$\Index(e)$ of $e$ is defined to be the index of the \emph{last} rule in the
sequence $\covering$ that contains $e$. We say that $\covering$ is \defn{valid}
if, for all two edges $e = (s,p), e' = (s',p')$ in $\CSSedges$
with $e \edgeOrder e'$, either of the following holds:
\begin{itemize}
    \item $\Index(e) < \Index(e')$
    \item 
        $\Index(e) = \Index(e')$ and, if $(\selsBucket,\{p_i\}_{i=1}^m)$ is
        the rule at position $\Index(e)$ in $\covering$, it is the
        case that $p = p_j$ and $p' = p_k$ with $j \leq k$.
\end{itemize}
In the example of Figure~\ref{fig:simpleEx} with the selector \texttt{.vegetable} in the third rule, we can verify that it is indeed a valid covering of itself by observing the only ordering is
$(\texttt{.fruit}, \texttt{color:red}) \edgeOrder (\texttt{.vegetable}, \texttt{color:blue})$
and we have
\[
    \Index((\mathttt{.fruit}, \mathttt{color:red})) = 2
    \quad\text{and}\quad
    \Index((\mathttt{.vegetable}, \mathttt{color:blue})) = 3 \ .
\]

This \emph{last-occurrence} semantics reflects the cascading aspect of
CSS. To relate this to the world of CSS, the original CSS file $F$ may be
represented by a CSS-graph $\CSSgraph$, but $F$ also turns out to be a valid
covering of $\CSSgraph$. In fact, the set of valid coverings of $\CSSgraph$
correspond to CSS files that are equivalent (up to reordering of selectors and
property declarations) to the original CSS file.
That is, if two files cover the same graph, then they will specify the same property declarations on any node of any given DOM.

To define the optimisation problem, we need to define the weight of a rule
and the weight of a covering. To this end, we extend the function
$\nodeWeight$ to rules and coverings by summing the weights of the nodes.
More precisely, given a rule $\bucket = \inbucket$, define
\[
    \Weight(\bucket) = \sum_{w \in \selsBucket \cup \propsBucket} \nodeWeight(w).
    \]
Similarly, given a covering $\covering = \incovering{1}{m}$, the weight
$\Weight(\covering)$ of $\covering$ is $\sum_{i=1}^m \Weight(\bucket_i)$.
It is easy to verify that the weight of a rule (resp.~covering) corresponds
to the number of non-whitespace characters in a CSS rule (resp.~file).
The \defn{minification problem} is, given a CSS-graph $\CSSgraph$, to compute a
valid covering with the minimum weight.

\paragraph{(Optimal) Rule-Merging Problem}
Given a CSS-graph $\CSSgraph$ and a covering $\covering$,
we define the \defn{trim} $\Trim{\covering}$ of $\covering$ to be the covering
$\covering'$ obtained by removing from each rule $\bucket = \inbucket$ (say
at position $i$) in
$\covering$ all nodes $v \in \selsBucket \cup \propsBucket$ that are not
incident
to at least one edge $e$ with $\Index(e) = i$ (i.e. the last occurrence of
$e$ in $\covering$). 
Such nodes $v$ may be removed since they do not affect
the validity of the covering $\covering$.

We can now explain formally the file size reduction shown in Figure~\ref{fig:trim-example}.
First observe in Figure~\ref{fig:trimmed-css} that the second rule
\begin{center}
\begin{minted}{css}
    .fruit, #broccoli { color:red; font-size:large }
\end{minted}
\end{center}
has been removed.
Consider the node \texttt{.fruit} and its incident edges
$(\mathttt{.fruit}, \mathttt{color:red})$
and
$(\mathttt{.fruit}, \mathttt{font-size:large})$.
Both of these edges have \Index\ $5$ ($\neq 2$) since they also appear in the last rule of Figure~\ref{fig:bigger-css}.
Thus we can remove the \texttt{.fruit} node from this rule.
A similar argument can be made for all nodes in this rule, which means that we remove them, leaving an empty rule (not shown).
In the fourth rule, the node \texttt{color:red} is incident only to
$(\mathttt{\#tomato}, \mathttt{color:red})$
which has \Index\ $5$ ($\neq 4$).
The situation is the same for the node \texttt{font-size:large}, thus both of these nodes are removed from the rule.

The trim $\Trim{\covering}$ can be
computed from $\covering$ and $\CSSgraph$ in polynomial time in a
straightforward way. More precisely, first build a hashmap for the index
function.  Second, go through all rules $\bucket$ in the covering $\covering$,
each node $v$ in $\bucket$, and each edge $e$
incident with $v$ checking if at least one such $e$ satisfies $\Index(e) = i$,
where $\bucket$ is the $i$th rule in $\covering$.
Note that
$\Trim{\covering}$ is
uniquely defined given $\covering$.

We define a \defn{(rule)-merging opportunity}
to
be a pair $(\bucket,j)$ of rule $\bucket$ and a number $j \in (0,|\covering|)$
such that $\inSeq{\covering}{j}{\bucket}$ is a valid covering
of $\CSSgraph$. The \defn{result} of applying this merging opportunity
is the covering $\Trim{\inSeq{\covering}{j}{\bucket}}$ obtained by trimming
$\inSeq{\covering}{j}{\bucket}$. The \defn{(rule)-merging problem} can be 
defined as
follows: given a CSS-graph $\CSSgraph$ and a valid covering $\covering$, find
a merging opportunity that results in a covering with the minimum weight.
This rule-merging problem is NP-hard even in the non-weighted version
\cite{Yan78,Peeters03}.


%% file: edgeOrder.tex
\section{CSS Selector Formalisation and its Intersection Problem}
\label{sec:edgeOrder}

In this section, we will show how to efficiently compute a
CSS-graph $\CSSgraph = \inCSSgraph$ from a given CSS file with the help of a
fast solver of quantifier-free theory of integer linear arithmetic, e.g., Z3
\cite{Z3}.
The key challenge is how to extract the order dependency $\edgeOrder$ from a
CSS file, which requires an algorithm for the
\defn{(selector-)intersection problem}, i.e., to check whether two given
selectors can be matched by the same element in \emph{some} document.
To address this, we provide a full formalisation of CSS3
selectors~\cite{CSS3sel} and a fast algorithm for the intersection
problem. Since our algorithm for the intersection problem is technically very
involved, we provide a short
intuitive explanation behind the algorithm in this section and leave the
details to Section \ref{sec:intersection}.

\subsection{Definition of Document Trees}
We define the semantics of CSS3 in terms of Document Object Models (DOMs),
which we also refer to as document trees.
The reader may find it helpful to recall the definition of trees from
Section~\ref{sec:prelim-trees}.

A document tree consists of a number of elements, which in turn may have sub-elements as children.
Each node has a \emph{type} consisting of an element name
(a.k.a. tag name) and a namespace. For example, an element
\texttt{p} in the default \texttt{html} namespace is a paragraph.
Namespaces commonly feature in programming languages (e.g. C++) to allow the use of multiple libraries whilst minimising the risk of overloading names.
For example, the HTML designers introduced a \texttt{div} element to represent a section of an HTML document.
Independent designers of an XML representation of mathematical formulas may represent division using elements also with the name \texttt{div}.
Confusion is avoided by the designers specifying a namespace in addition to the required element names.
In this case, the HTML designers would introduce the \texttt{div} element to the \texttt{html} namespace, while the mathematical \texttt{div} may belong to a namespace \texttt{math}.
Hence, there is no confusion between \texttt{html:div} and \texttt{math:div}.
As an aside, note that an HTML file may contain multiple namespaces, e.g., see \cite{html-namespace}.

Moreover, nodes may also be labelled by attributes, which take string values.
For example, an HTML \verb+img+ element has a \verb+src+ attribute specifying the source of the image.
Finally, a node may be labelled by a number of \emph{pseudo-classes}.
For example $\psenabled$ means that the node is enabled and the user may interact with it.
The set of pseudo-classes is fixed by the CSS specification.

We first the formal definition before returning to the example in Section~\ref{sec:html-example}.

\subsubsection{Formal Definition}

In the formal definition below we permit a possibly infinite set of element,
namespace, and attribute names. CSS stylesheets are \emph{not} limited to HTML
documents (which have a fixed set of element names, but not attribute names
since you can create custom \verb+data-*+ attributes), but they can also be
applied to documents of other types (e.g. XML) that permit custom element names,
namespaces, and attribute names.
Thus, the sets of possible names \emph{cannot} be fixed to finite sets from the
point of view of CSS selectors, which may be applied to any document.

When it is known that the set of elements or attribute names is fixed (e.g. when only considering HTML), it is possible to adapt our approach.
In particular, the small model property in Proposition~\ref{prop:boundedtypes} may be avoided, slightly simplifying the technique.

We denote the set of \defn{pseudo-classes} as
\[
    \pclss = %
    \set{ %
        \begin{array}{c} %
            \pslink, %
            \psvisited, %
            \pshover, %
            \psactive, %
            \psfocus, %
            \pstarget, %
            \\ %
            \psenabled, %
            \psdisabled, %
            \pschecked, %
            \psroot, %
            \psempty %
        \end{array} %
    } \ . %
\]
Then, given
    a possibly infinite set of \defn{namespaces} $\nspaces$,
    a possibly infinite set of \defn{element names} $\eles$,
    a possibly infinite set of \defn{attribute names} $\atts$,
    and a finite alphabet\footnote{See the notes at the end of
    the section.} $\alphabet$ containing the special characters $\cspace$ and $\cdash$ (space and dash respectively),
a \defn{document tree} is a $\ialphabet$-labelled tree
$\tup{\treedom, \treelab}$, where
\[
    \ialphabet := \brac{\nspaces \times \eles \times
                  \finfuns{\nspaces \times \atts}
                          {\alphabet^\ast}
                  \times
                  2^{\pclss}} \ .
\]
Here the notation
$\finfuns{\nssatts}{\alphabet^\ast}$
denotes the set of partial functions from $\brac{\nssatts}$
to $\alphabet^\ast$ whose domain is finite.
In other words,
each node in a document tree is labeled by a namespace, an element, a function
associating a finite number of namespace-attribute pairs with attribute values
(strings), and zero or more of the pseudo-classes.
For a function $\attfun \in \finfuns{\nssatts}{\alphabet^\ast}$
we say $\ap{\attfun}{\ns,\att} = \attundef$
when $\attfun$ is undefined over $\ns \in \nspaces$ and $\att \in \atts$,
where $\attundef \notin \alphabet^\ast$ is a special undefined value.
Furthermore, we assume special attribute names $\classatt, \idatt \in
\atts$ that will be used to attach classes (see later) and IDs to nodes.

When
$\ap{\treelab}{\node} = \tup{\ns,\ele, \attfun, \pclss}$
we define the following projections of the labelling function
\[
    \begin{array}{rcl}
        \ap{\treelabns}{\node} &=& \ns,\\ %
        \ap{\treelabele}{\node} &=& \ele,\\ %
        \ap{\treelabatts}{\node} &=& \attfun,\ \text{and} \\ %
        \ap{\treelabpclss}{\node} &=& \pclss .
    \end{array}
\]
We will use the following standard XML notation:
for an element name $\ele$, a namespace $\ns$, 
and an attribute name $\att$, let $\qele{\ns}{\ele}$ (resp.~$\qele{\ns}{\att}$)
denote the pair $(\ns,\ele)$ (resp.~$(\ns,\att)$). The notation helps to clarify
the role of namespaces as a means of providing contextual or scoping information
to an element/attribute.

There are several consistency constraints on the node labellings.
\begin{compactitem}
\item
    For each $\ns \in \nspaces$, there are \emph{no} two nodes in the tree
    with the same value of $\qele{\ns}{\idatt}$.
\item
    A node cannot be labelled by both $\pslink$ and $\psvisited$.
\item
    A node cannot be labelled by both $\psenabled$ and $\psdisabled$.
\item
    Only one node in the tree may be labelled $\pstarget$.
\item
    A node contains the label $\psroot$ iff it is the root node.
\item
    A node labelled $\psempty$ must have no children.
\end{compactitem}
From now on, we will tacitly assume that document trees satisfy these
consistency constraints. We write $\treeset{\nspaces}{\eles}{\atts}{\alphabet}$
for the set of such trees.

\subsubsection{Example}

Consider the HTML file in Figure~\ref{fig:simple-html}.
This file can be represented precisely by the tree in Figure~\ref{fig:example-dom-tree}.
In this tree, each node is first labelled by its type, which consists of its namespace and its element name.
In all cases, the namespace is the \texttt{html} namespace, while the element names are exactly as in the HTML file.
In addition, we label each node with the attributes and pseudo-classes with which they are labelled.
The \texttt{html:html} node is labelled with $\psroot$ since it is the root element of the tree.
The \texttt{html:div} node is labelled with the ID \texttt{heading}.
The \texttt{html:img} node is labelled with the attribute \texttt{class} with value \texttt{banner} and the attribute \texttt{html:src} with value \texttt{banner.jpg}, as well as the pseudo-class $\psempty$ indicating that the node has no contents.
The remaining nodes however are not labelled $\psempty$ since even the leaf nodes contain some text (which is not represented in our tree model as it is not matchable by a selector).
Hence, a node with no child may still be non-empty.

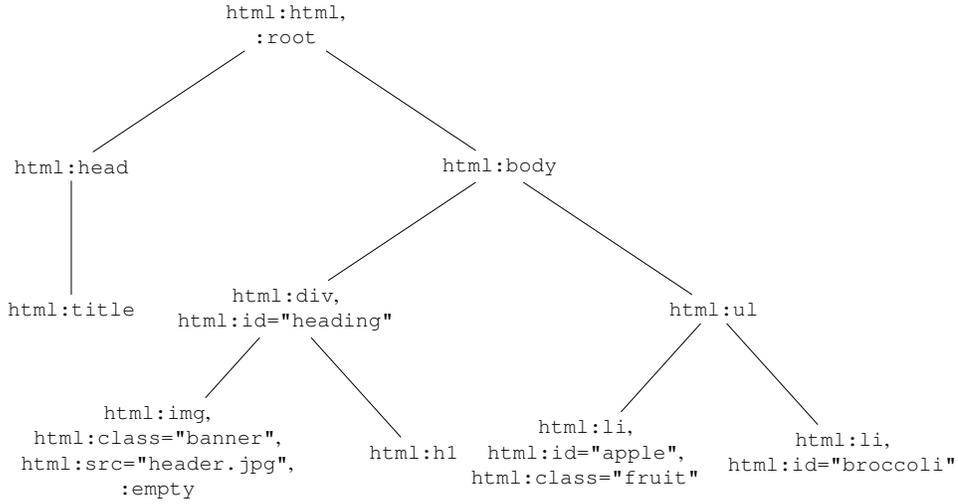
\begin{figure}
    \scalebox{.8}{
    \begin{tikzpicture}[level 1/.style={sibling distance=45ex},
                        level 2/.style={sibling distance=45ex},
                        level 3/.style={sibling distance=27ex},
                        level distance=15ex,
                        every node/.style={align=center}]
        \node {\texttt{html:html},\\
               \texttt{:root}}
            child {node {\texttt{html:head}}
                child {node {\texttt{html:title}}}
            }
            child {node {\texttt{html:body}}
                child {node {\texttt{html:div},\\
                             \texttt{html:id="heading"}}
                    child {node {\texttt{html:img},\\
                                 \texttt{html:class="banner"},\\
                                 \texttt{html:src="header.jpg"},\\
                                 \texttt{:empty}}}
                    child {node {\texttt{html:h1}}}
                }
                child {node {\texttt{html:ul}}
                    child {node {\texttt{html:li},\\
                                 \texttt{html:id="apple"},\\
                                 \texttt{html:class="fruit"}}}
                    child {node {\texttt{html:li},\\
                                 \texttt{html:id="broccoli"}}}
                }
            };
    \end{tikzpicture}
    }
    \caption{\label{fig:example-dom-tree}A DOM tree representation of the HTML file.}
\end{figure}

Formally, the DOM tree is
$\tup{\treedom, \treelab}$
where
$\treedom = \set{\eseq, 1, 2, 11, 21, 22, 211, 212, 221, 222}$
and
\[
    \begin{array}{rcl}
        \ap{\treelab}{\eseq} &=& \tup{\mathttt{html},
                                      \mathttt{html},
                                      \emptyset,
                                      \set{\psroot}} \\
        \ap{\treelab}{1} &=& \tup{\mathttt{html},
                                  \mathttt{head},
                                  \emptyset,
                                  \emptyset} \\
        \ap{\treelab}{2} &=& \tup{\mathttt{html},
                                  \mathttt{body},
                                  \emptyset,
                                  \emptyset} \\
        \ap{\treelab}{11} &=& \tup{\mathttt{html},
                                   \mathttt{title},
                                   \emptyset,
                                   \emptyset} \\
        \ap{\treelab}{21} &=& \tup{\mathttt{html},
                                   \mathttt{div},
                                   \tup{\mathttt{html}, \mathttt{id}}
                                   \mapsto
                                   \mathttt{heading},
                                   \emptyset} \\
        \ap{\treelab}{22} &=& \tup{\mathttt{html},
                                   \mathttt{ul},
                                   \emptyset,
                                   \emptyset} \\
        \ap{\treelab}{211} &=&  \tup{\mathttt{html},
                                     \mathttt{img},
                                     \tup{\begin{array}{l}
                                              \tup{\mathttt{html}, \mathttt{class}}
                                              \mapsto
                                              \mathttt{banner}, \\
                                              \tup{\mathttt{html}, \mathttt{src}}
                                              \mapsto
                                              \mathttt{header.jpg}
                                     \end{array}},
                                     \set{\psempty}} \\
        \ap{\treelab}{212} &=& \tup{\mathttt{html},
                                    \mathttt{h1},
                                    \emptyset,
                                    \emptyset} \\
        \ap{\treelab}{221} &=&  \tup{\mathttt{html},
                                     \mathttt{li},
                                     \tup{\begin{array}{l}
                                              \tup{\mathttt{html}, \mathttt{id}}
                                              \mapsto
                                              \mathttt{apple}, \\
                                              \tup{\mathttt{html}, \mathttt{class}}
                                              \mapsto
                                              \mathttt{fruit}
                                     \end{array}},
                                     \emptyset} \\
        \ap{\treelab}{222} &=& \tup{\mathttt{html},
                                    \mathttt{li},
                                    \tup{\mathttt{html}, \mathttt{id}}
                                    \mapsto
                                    \mathttt{broccoli},
                                    \emptyset} \ .
    \end{array}
\]

\subsection{Definition of CSS3 selectors}
\label{sec:css-selectors-def}

In the following sections we define CSS selectors syntax and semantics.
Informally, a CSS selector consists of \emph{node selectors} $\csssim$ --- which
match individual nodes in the tree --- combined using the operators
$\cssdescendant$, $\csschild$, $\cssneighbour$, and $\csssibling$.
These operators express the descendant-of, child-of, neighbour-of, and sibling-of relations respectively.
Note that the blank space character is used instead of $\cssdescendant$ in
CSS3, though we opt for the latter in the formalisation for the sake of
readability. So, for example, we use $\text{\texttt{.journal}} \cssdescendant
\text{\texttt{.science}}$ (i.e. choose all nodes with class \texttt{.science}
that is a descendant of nodes with class \texttt{.journal}) instead of
the standard syntax \verb+.journal .science+.
In addition, in order to distinguish syntax from meaning, we use slightly
different notation to their counterpart semantical operators
$\treedescendant$, $\treechild$, $\treeneighbour$, and $\treesibling$.

We remark that a comma (\texttt{,}) is not an operator in the CSS selector syntax.
Instead a \emph{selector group} is a comma-separated list of selectors that is matched if any of its selectors is matched.
A CSS rule thus consists of a selector group and a list of property declarations.
For the purposes of rule-merging it is desirable to treat selectors individually as it allows the most flexibility in reorganising the CSS file.
Hence we treat a selector group simply as a set of selectors that can be separated if needed.

A node selector $\csssim$ has the form $\csstype\cssconds$
where $\csstype$ constrains the \emph{type} of the node.
That is, $\csstype$ places restrictions on the element label of the node,
e.g., \verb+p+ for paragraph elements and \verb+*+ (or an empty string)
for all elements.
The rest of the selector is a set $\cssconds$ of \emph{simple}
selectors (written as a concatenation of strings representing these simple
selectors) that assert atomic properties of the node.
There are four types of simple selectors.
\OMIT{
The first form includes class and ID selectors $\isclass{\attval}$ and
$\isid{\attval}$, which assert that the node has a class or ID $\attval$
respectively. Note, a node may have several classes (given as a string by a
space-separated list of classes) but only one ID.
}

\smallskip
\noindent
\underline{\emph{Type 1}:} attribute selectors of the form
$\opattns{\ns}{\att}{\attop}{\attval}$
for some namespace $\ns$, attribute $\att$, operator
$\attop \in \set{\opis, \ophas, \opbegin,
                 \opstrbegin, \opstrend, \opstrsub}$,
and some string $\attval \in \alphabet^\ast$.
We may write $\opatt{\att}{\attop}{\attval}$ to mean that a node can be matched
by $\opattns{\ns}{\att}{\attop}{\attval}$ for some $\ns$.
The operators $\opis, \opstrbegin, \opstrend$, and $\opstrsub$
take their meaning from regular expressions.
That is, equals, begins-with, ends-with, and contains respectively.
The remaining operators are more subtle.
The $\ophas$ operator means the attribute is a string of space-separated values, one of which is $\attval$.
The $\opbegin$ operator is intended for use with language identification,
e.g., as in the attribute selector
$\opatt{\texttt{lang}}{\opbegin}{\texttt{"en-GB"}}$ to mean ``English'' as
spoken in Great Britain.
Thus $\opbegin$ asserts that either the attribute has value $\attval$ or is a string of the form
$\attval\mdash\attval'$
where $\mdash$ is the dash character, and $\attval'$ is some string.
Note that if the \texttt{lang} attribute value of a node is \verb+en-GB+,
the node also matches the simple selector
$\opatt{\texttt{lang}}{\opbegin}{\texttt{"en"}}$. In addition, recall that
$\classatt$ and $\idatt$ are two special attribute names. For this reason, CSS introduces
the shorthands $\isclass{\attval}$ and $\isid{\attval}$ for,
respectively, the simple selectors $\opatt{\texttt{class}}{\ophas}{\attval}$
and $\opatt{\texttt{id}}{\opis}{\attval}$, i.e., asserting that the node has a
class or ID $\attval$. An example of a valid CSS selector is
the selector \verb+h1.fruit.vegetable+, which chooses all nodes with class
\verb+fruit+ and \verb+vegetable+, and element name \verb+h1+ (which includes
the following two elements:
\verb+<h1 class="fruit vegetable">+ and \verb+<h1 class="vegetable fruit">+).

\smallskip
\noindent
\underline{\emph{Type 2}.}
attribute selectors of the form $\hasattns{\ns}{\att}$, asserting that the
attribute is merely defined on the node. As before, we may write
$\hasatt{\att}$ to mean that the node may be matched by $\hasattns{\ns}{\att}$
for some namespace $\ns$. As an example,
$\texttt{img}\hasatt{\texttt{alt}}$ chooses all \verb+img+ elements where the
attribute \verb+alt+ is defined.

\smallskip
\noindent
\underline{\emph{Type 3}.} pseudo-class label of a node, e.g.,
the selector $\psenabled$ ensures the node is currently enabled in the
document. There are several further kinds of pseudo-classes that assert
counting constraints on the children of a selected node. Of particular interest
are selectors such as $\psnthchild{\coefa}{\coefb}$, which
assert that the node has a particular position in the sibling order.
For example, \acmeasychair{}{the selector}\ $\psnthchild{2}{1}$ means there is some $n \geq 0$ such that the
node is the $(2n + 1)$st node in the sibling order.

\smallskip
\noindent
\underline{\emph{Type 4}.} negations $\cssneg{\csscond}$ of a simple selector
$\csscond$ with the condition that negations cannot be nested or apply to multiple atoms. For example,
\verb+:not(.fruit):not(.vegetable)+ is a valid selector, whereas
\verb+:not(:not(.vegetable))+ and
\verb+:not(.fruit.vegetable)+ are \emph{not} a valid selectors.

\subsubsection{Syntax}

Fix the sets $\nspaces$,
    $\eles$,
    $\atts$, and
    $\alphabet$.
We define $\selectors$ for the set of \defn{(CSS) selectors} and
$\nodeselectors$ for the set of \defn{node selectors}.
In the following $\synalt$ will be used to separate syntax alternatives, while $\csspipe$ is an item of CSS syntax.
The set $\selectors$ is the set of formulas $\css$ defined as:
\[
    \css ::= \csssim \synalt %
           \css \cssdescendant \csssim \synalt %
           \css \csschild \csssim \synalt %
           \css \cssneighbour \csssim \synalt %
           \css \csssibling \csssim %
\]
where $\csssim \in \nodeselectors$ is a \emph{node selector} with syntax
$
    \csssim ::=
    \csstype
    \cssconds
$
with $\csstype$ having the form
\[
    \csstype ::= \isany \synalt
    \isanyns{\ns} \synalt
    \isele{\ele} \synalt
    \iselens{\ns}{\ele}
\]
where
$\ns \in \nspaces$
and
$\ele \in \eles$
and $\cssconds$ is a (possibly empty) set of conditions $\csscond$ with syntax
\[
    \csscond ::= \csscondnoneg \synalt
               \cssneg{\csssimnoneg}
\]
where $\csscondnoneg$ and $\csssimnoneg$ are conditions that do not contain
negations, i.e.:
\begingroup
\allowdisplaybreaks
\begin{eqnarray*}
    \csssimnoneg & ::= & \isany \synalt
                   \isanyns{\ns} \synalt
                   \isele{\ele} \synalt
                   \iselens{\ns}{\ele} \synalt
                   \csscondnoneg \\
    \csscondnoneg & ::= &
        \hasattns{\ns}{\att} \synalt %
        \opattns{\ns}{\att}{\attop}{\attval} \synalt
        \hasatt{\att} \synalt %
        \opatt{\att}{\attop}{\attval} \synalt
        \\ %
        & & \pslink \synalt %
        \psvisited \synalt %
        \pshover \synalt %
        \psactive \synalt %
        \psfocus \synalt %
        \\ %
        & & \psenabled \synalt %
        \psdisabled \synalt %
        \pschecked \synalt %
        \\ %
        & & \psroot \synalt %
        \psempty \synalt %
        \pstarget \synalt %
        \\ %
        & & \psnthchild{\coefa}{\coefb} \synalt %
        \psnthlastchild{\coefa}{\coefb} \synalt %
        \\ %
        & & \psnthoftype{\coefa}{\coefb} \synalt %
        \psnthlastoftype{\coefa}{\coefb} %
        \\ %
        & & \psonlychild \synalt %
        \psonlyoftype \\
     \attop & ::= & \opis \synalt \ophas \synalt \opbegin \synalt \opstrbegin
     \synalt \opstrend \synalt \opstrsub
\end{eqnarray*}
\endgroup
where $\ns \in \nspaces$,
     $\ele \in \eles$,
     $\att \in \atts$,
     $\attval \in \alphabet^\ast$, and
     $\coefa, \coefb \in \Z$.
Whenever $\cssconds$ is the empty set, we will denote the
node selector $\csstype\cssconds$ as $\csstype$ instead of $\csstype\emptyset$.

\subsubsection{Semantics}
The semantics
of a selector is defined with respect to a document tree and a node in the tree.
More precisely, the semantics of CSS3 selectors $\css$ are defined inductively
with respect to a document tree $\tree = \tup{\treedom, \treelab}$
and a node $\node \in \treedom$ as follows.
(Note: (1) $\pcls$ ranges over the set $\pclss$ of pseudo-classes,
(2) $\attval\attval'$ is the concatenation of the strings $\attval$ and
$\attval'$, and (3) $\attval\mdash\attval'$ is the concatenation of $\attval$
and $\attval'$ with a ``$\mdash$'' in between.)
\begingroup
\allowdisplaybreaks
\begin{alignat*}{2}
        \tree, \node &\models \css \cssdescendant \csssim %
        &\quad\iffdef\quad %
        &\exists \node' \treedescendant \node \ . %
                \brac{\tree, \node' \models \css} %
                \textand
                \brac{\tree, \node \models \csssim} %
        \\ %
        \tree, \node &\models \css \csschild \csssim %
        &\quad\iffdef\quad %
        &\exists \node' \treechild \node \ . %
            \brac{\tree, \node' \models \css} %
            \textand
            \brac{\tree, \node \models \csssim} %
        \\ %
        \tree, \node &\models \css \cssneighbour \csssim %
        &\quad\iffdef\quad %
        &\exists \node' \treeneighbour \node \ . %
            \brac{\tree, \node' \models \css} %
            \textand %
            \brac{\tree, \node \models \csssim} %
        \\ %
        \tree, \node &\models \css \csssibling \csssim %
        &\quad\iffdef\quad %
        &\exists \node' \treesibling \node \ . %
            \brac{\tree, \node' \models \css} %
            \textand %
            \brac{\tree, \node \models \csssim} %
        \\
        \tree, \node &\models \csstype \cssconds %
        &\quad\iffdef\quad %
        &\brac{\tree, \node \models \csstype} %
        \textand %
        \forall \csscond \in \cssconds \ . %
            \brac{\tree, \node \models \csscond} %
        \\ %
        \tree, \node &\models \isanyns{\ns} %
        &\quad\iffdef\quad %
        &\ns = \ap{\treelabns}{\node} %
        \\ %
        \tree, \node &\models \isany %
        &\quad\iffdef\quad %
        &\top %
        \\ %
        \tree, \node &\models \iselens{\ns}{\ele} %
        &\quad\iffdef\quad %
        &\ns = \ap{\treelabns}{\node} %
        \land %
        \ele = \ap{\treelabele}{\node} \\ %
        \tree, \node &\models \isele{\ele} %
        &\quad\iffdef\quad %
        &\tree, \node \models \iselens{\ns}{\ele} \quad \text{for some
        $\ns \in \nspaces$}%
        \\ %
        \tree, \node &\models \pcls %
        &\quad\iffdef\quad %
        &\pcls \in \ap{\treelabpclss}{\node} %
        \\
        \tree, \node &\models \cssneg{\csscondnoneg} %
        &\quad\iffdef\quad %
        &\neg\brac{\tree, \node \models \csscondnoneg} %
        \\
        \tree, \node &\models \hasatt{\att} %
        &\quad\iffdef\quad %
            &\tree, \node \models \hasattns{\ns}{\att} %
            \quad \text{for some $\ns \in \nspaces$}
        \\ %
        \tree, \node &\models \opatt{\att}{\attop}{\attval}
        &\quad\iffdef\quad
            &\tree, \node \models \opattns{\ns}{\att}{\attop}{\attval}
            \quad \text{for some $\ns \in \nspaces$}
        \\
        \tree, \node &\models \hasattns{\ns}{\att} %
        &\quad\iffdef\quad
            &\ap{\ap{\treelabatts}{\node}}{\ns, \att} \neq \attundef %
        \\ %
        \tree, \node &\models \attisns{\ns}{\att}{\attval} %
        &\quad\iffdef\quad %
            &\ap{\ap{\treelabatts}{\node}}{\ns,\att} = \attval %
        \\ %
        \tree, \node &\models \attbeginns{\ns}{\att}{\attval} %
        &\quad\iffdef\quad %
        &\brac{ %
            \begin{array}{c} %
                \brac{\ap{\ap{\treelabatts}{\node}}{\ns,\att} = \attval} %
                \ \text{ or } \\ %
                \exists \attval' \ . %
                \brac{ %
                    \ap{\ap{\treelabatts}{\node}}{\ns, \att} = %
                    \attval \cdash \attval' %
                } %
             \end{array} %
        } %
        \\ %
        \tree, \node &\models \attstrbeginns{\ns}{\att}{\attval} %
        &\quad\iffdef\quad %
        &\exists \attval' \in \alphabet^\ast \ .\ %
            \ap{\ap{\treelabatts}{\node}}{\ns,\att} = \attval \attval' %
        \\ %
        \tree, \node &\models \attstrendns{\ns}{\att}{\attval} %
        &\quad\iffdef\quad %
        &\exists \attval' \in \alphabet^\ast \ .\ %
            \ap{\ap{\treelabatts}{\node}}{\ns,\att} = \attval' \attval %
        \\ %
        \tree, \node &\models \attstrsubns{\ns}{\att}{\attval} %
        &\quad\iffdef\quad %
        &\exists \attval_1, \attval_2 \in \alphabet^\ast \ .\ %
            \ap{\ap{\treelabatts}{\node}}{\ns,\att} = %
            \attval_1 \attval \attval_2 %
\end{alignat*}
\endgroup
with the missing attribute selector being (noting
$\attval\cspace\attval'$
is the concatenation of $\attval$ and $\attval'$ with the space character $\cspace$ in between)
\begin{eqnarray*}
    \tree, \node \models \atthas{\att}{\attval} & \iffdef &
                \ap{\ap{\treelabatts}{\node}}{\ns, \att} = \attval %
                \ \text{ or }
                \exists \attval'  \ . %
                    \brac{ %
                        \ap{\ap{\treelabatts}{\node}}{\ns, \att} = %
                        \attval \cspace \attval' %
                    } %
                \ \text{ or } \\
            & & \exists \attval' \ . %
                    \brac{ %
                        \ap{\ap{\treelabatts}{\node}}{\ns, \att} = %
                        \attval' \cspace \attval %
                    } %
                \ \text{ or }
                \exists \attval_1, \attval_2 \ . %
                \brac{ %
                    \ap{\ap{\treelabatts}{\node}}{\ns, \att} = %
                    \attval_1 %
                    \cspace \attval \cspace %
                    \attval_2 %
                } %
\end{eqnarray*}
then, for the counting selectors
\begingroup
\allowdisplaybreaks
\begin{eqnarray*}
        \tree, \node \models \psnthchild{\coefa}{\coefb}
        & \iffdef &
        \text{there is some $\numof \in \N$ such that $\node$ is the $\coefa\numof+\coefb$th} \\
        & & \text{child} \\
        \\
        \tree, \node \models \psonlychild & \iffdef &
        \text{the parent of $\node$ has precisely one child} \\
        \tree, \node \models \psnthoftype{\coefa}{\coefb} & %
        \iffdef & \text{there is some $\numof \in \N$ such that the parent of $\node$ has} \\
               & & \text{precisely $\coefa\numof + \coefb - 1$ children with namespace $\ap{\treelabns}{\node}$} \\
               & & \text{and element name $\ap{\treelabele}{\node}$ for some $n$ that are (strictly)} \\
               & & \text{preceding siblings of $\node$} \\
        \tree, \node \models \psonlyoftype & \iffdef &
        \text{the parent of $\node$ has precisely one child with} \\
              & & \text{namespace $\ap{\treelabns}{\node}$ and element name $\ap{\treelabele}{\node}$}
\end{eqnarray*}
\endgroup
Finally, the semantics of the remaining two selectors, which are
$\psnthlastchild{\coefa}{\coefb}$ and $\psnthlastoftype{\coefa}{\coefb}$,
is exactly the same as $\psnthchild{\coefa}{\coefb}$ and respectively
$\psnthoftype{\coefa}{\coefb}$, except with the sibling ordering
reversed
(i.e. the rightmost child of a parent is treated as the first).

\begin{remark}
Readers familiar with HTML may have expected more constraints in the
semantics. For example, if a node matches $\pshover$, then its parent should
also match $\pshover$.
However, this is part of the HTML5 specification, not of CSS3.
In fact, the CSS3 selectors specification explicitly states that a node matching $\pshover$ does not imply its parent must also match $\pshover$.
\end{remark}

\subsubsection{Divergences from full CSS}

Note that we diverge from the full CSS specification in a number of places.
However, we do not lose expressivity.
\begin{compactitem}
\item
    We assume each element has a namespace.
    In particular, we do not allow elements without a namespace. There is
        no loss of generality here since we can simply assume a ``null''
        namespace is used instead.
    Moreover, we do not support default name spaces and assume namespaces are explicitly given.

\item
    We did not include
    $\pslang{\csslang}$.
    Instead, we will assume (for convenience) that all nodes are labelled with a language attribute with some fixed namespace $\ns$.
    In this case,
    $\pslang{\csslang}$
    is equivalent\footnote{
        The CSS specification defines
        $\pslang{\csslang}$
        in this way.
        A restriction of the language values to standardised language codes is only a recommendation.
    } to
    $\attbeginns{\ns}{\langatt}{\csslang}$.

\item
    We did not include $\psindeterminate$ since it is not formally part of the CSS3 specification.

\item
    We omit the selectors $\psfirstchild$ and $\pslastchild$, as well as \\ $\psfirstoftype$ and $\pslastoftype$, since they are expressible using the other operators.

\item
    We omitted $\texttt{even}$ and $\texttt{odd}$ from the nth child operators since these are easily definable as $2n$ and $2n+1$.

\item
    We do not explicitly handle document fragments.
    These may be handled in a number of ways.
    For example, by adding a phantom root element (since the root of a document fragment does not match $\psroot$) with a fresh ID $\id$ and adjusting each node selector in the CSS selector to assert $\cssneg{\isid{\id}}$.
    Similarly, lists of document fragments can be modelled by adding several subtrees to the phantom root.
\item
    A CSS selector can be suffixed with a \emph{pseudo-element}
    which is of the form $\psfirstline$, $\psfirstletter$, $\psbefore$, and $\psafter$.
    Pseudo-elements are easy to handle and only provide a distraction
    to our presentation. For this reason, we relegate them into the appendix.

\item
    We define our DOM trees to use a finite alphabet $\alphabet$.
    Currently the CSS3 selectors specification uses Unicode as its alphabet for lexing.
    Although the CSS3 specification is not explicit about the finiteness of characters appearing in potential DOMs, since Unicode is finite~\cite{Unicode} (with a maximal possible codepoint) we feel it is reasonable to assume DOMs are also defined over a finite alphabet.
\end{compactitem}

\subsection{Solving the intersection problem}
\label{sec:solving-intersection}

We now address the problem of checking the intersection of two CSS selectors.
Let us write
\[
    \sem{\css} := \{ (\tree,\node) : \tree,\node \models \css \}
\]
to denote the set of pairs of tree and node satisfying the selector $\css$.
The \defn{intersection problem of CSS selectors} is to decide if
$\sem{\css} \cap \sem{\css'} \neq \emptyset$, for two given selectors
$\css$ and $\css'$. A closely related decision problem is the
\defn{non-emptiness problem of CSS selectors}, which is to decide if
$\sem{\css} \neq \emptyset$, for a given selector $\css$. The two problems
are \emph{not} the same since CSS selectors are not closed under intersection
(i.e. the conjunction of two CSS selectors is in general not a valid CSS
selector).

\OMIT{
Our main
result is that these two problems can be efficiently reduced to satisfiability
over quantifier-free theory over integer linear arithmetic, for which there is
a highly-optimised solver (e.g. Z3 \cite{Z3}).
}

\begin{namedtheorem}{thm:emptiness}{Non-Emptiness}
    The non-emptiness problem for CSS selectors
    is efficiently reducible to satisfiability over quantifier-free theory
    over integer linear arithmetic. Moreover, the problem is NP-complete.
\end{namedtheorem}

\begin{namedtheorem}{thm:intersection}{Intersection}
    The intersection problem for CSS selectors
    is efficiently reducible to satisfiability
    over quantifier-free theory over integer linear arithmetic. Moreover,
    the problem is NP-complete.
\end{namedtheorem}
\noindent
Recall from Section \ref{sec:prelim} that satisfiability over quantifier-free
theory over integer linear arithmetic is in NP and can be solved by a
highly-optimised SMT solver (e.g. Z3 \cite{Z3}).
The NP-hardness in the above theorems suggests that 
our SMT-based approach is theoretically optimal.
In addition,
our experiments with real-world selectors (see Section \ref{sec:experiments})
suggest that our SMT-based approach can be optimised to be fast in practice (see Appendix~\ref{sec:optimised-aut-emp}), with
each problem instance solved within moments.

\paragraph{Idea behind our SMT-based approach}
We now provide the idea behind the efficient reduction to quantifier-free
theory over integer linear arithmetic. Our reduction first goes via a
new class of tree automata (called CSS automata), which --- like CSS selectors
--- are symbolic representations of sets of pairs containing a document tree
and a node in the tree. We will call such sets \defn{languages
recognised by the automata}. Given a CSS selector $\css$, we can efficiently
construct a CSS automaton $\cssaut$ that can symbolically represent
$\sem{\css}$. Unlike CSS selectors, however, we will see that languages
recognised by CSS automata enjoy closure under intersection, which will
allow us to treat the intersection problem as the non-emptiness problem.
\OMIT{
The standard
technique for solving the intersection problem for two finite-state automata
is to perform the standard product construction of the two automata
\cite{Vardi95}, which works since finite-state automata are closed under
intersection. However, this is not the case for CSS selectors. Therefore,
we introduce the class of \emph{CSS automata} that can subsume the expressivity
of CSS selectors, but at the same time are closed under intersection.
}
More precisely,
a CSS automaton navigates a tree structure in a similar manner to a CSS
selector: transitions may only move down the tree or to a sibling, while
checking a number of properties on the visited nodes.
The difficulty of taking a direct intersection of two selectors is that
the two selectors may descend to different child nodes, and then meet again
after the application of a number of sibling combinators, i.e., their paths may
diverge and combine several times. CSS automata overcome this difficulty by
always descending to the \emph{first} child, and then move from
\emph{sibling to sibling}.
Thus, the intersection of CSS automata can be done with a straightforward
automata product construction, e.g., see \cite{Vardi95}.

Next, we show that the non-emptiness of CSS automata
can be decided in NP by a polynomial-time reduction to satisfiability of
quantifier-free theory of integer linear arithmetic.
In our experiments, we used an optimised version of the reduction, which is
detailed in Appendix~\ref{sec:optimised-aut-emp}.
For non-emptiness, ideally, we would like
to show that if a CSS automaton has a non-empty language, then it accepts
a small tree (i.e. with polynomially many nodes). This is unfortunately not
the case, as the reader can see in our NP-hardness proof idea below.
Therefore, we use a different strategy. First , we prove three ``small model
lemmas''.
The first is quite straightforward and shows that, to prove non-emptiness, it
suffices to consider a witnessing automata run of length $n$ for an automaton
with $n$ transitions (each automata transition allows some nodes to be skipped).
Second, we show that
it suffices to consider attribute selector values (i.e.\ strings) of length
linear in the size of the CSS automata.
This is non-trivial and uses a construction inspired by \cite{MW05}.
Third, we show that it suffices to consider trees whose sets of namespaces
and element names are linear in the size of the CSS automaton.
Our formula $\varphi$ attempts to
guess this automata run, the attribute selector values, element names, and
namespaces.
The global ID constraint (i.e.\ all the guessed IDs are distinct) can be easily
asserted in the formula.
So far, boolean variables are sufficient because the small model lemmas
allow us to do bit-blasting.
\emph{Where, then, do the integer variables come into play?} For each position
$i$ in the guessed path, we introduce an integer variable $\numvar{\idxi}$ to
denote that the node at position $i$ in the path is the $\numvar{\idxi}$th
child. This is necessary if we want to assert counting constraints like
$\psnthchild{\alpha}{\beta}$, which would be encoded in integer linear
arithmetic as
$\exists\nvar: \numvar{\idxi} = \coefa \nvar + \coefb$.

\paragraph{Proof Idea of NP-hardness}
We now provide an intuition on how to prove NP-hardness in the above theorems.
First, observe that the intersection is computationally at least as hard as the
non-emptiness problem since we can set the second selector to be $\isany$.
To prove NP-hardness of the non-emptiness, we give a polynomial-time reduction
from the NP-complete
problem of \emph{non-universality of unions of arithmetic progressions}
\cite[Proof of Theorem 6.1]{SM73}. Examples of arithmetic progressions are
$2\N + 1$ and $5\N + 2$, which are shorthands for the sets
$\{1,3,5,\ldots\}$ and $\{2,7,12,\ldots\}$, respectively. The aforementioned
non-universality problem allows an arbitrary number of arithmetic progressions
as part of the input and we are interested in checking whether the union
equals the entire set $\N$ of natural numbers.
As an example of the reduction, checking $\N \neq 2\N+1 \cup 5\N+2$ is
equivalent to the non-emptiness of
\begin{verbatim}
    :not(root):not(:nth-child(2n+2)):not(:nth-child(5n+3))
\end{verbatim}
which can be paraphrased as checking the existence of a tree with a node that
is neither the root, nor the $2n+2$nd child, nor the $5n+3$rd child.
Observe that we add 1 to the offset of the arithmetic progressions since
the selector \verb+:nth-child+ starts counting (the number of children) from 1,
not from 0.
A full NP-hardness proof is available in Appendix~\ref{sec:np-hardness-proof}.

\subsection{Extracting the edge order $\edgeOrder$ from a CSS file}
\label{sec:edge-order-from-css}

Recall that
our original goal is to compute a CSS-graph $\CSSgraph = \inCSSgraph$ from a
given CSS file.
The sets $\propNodes$, $\selNodes$, and $\CSSedges$, and the function
$\nodeWeight$ can be computed easily as explained in Section~\ref{sec:css2graph}.
We now show how to compute $\edgeOrder$ using the algorithm for checking
intersection of two selectors.
We present an intuitive ordering, before explaining how this may be relaxed while still preserving the semantics.

An initial definition of $\edgeOrder$ is simple to obtain:
we want to order $(s, p) \edgeOrder (s', p')$ whenever
    $(s', p')$ appears later in the CSS file than $(s, p)$,
    the selectors may overlap but are not distinguished by their specificity, and
    $p$ and $p'$ assign conflicting values to a given property name.
More formally, we first compute the specificity of all the selectors in the CSS file.
This can be easily computed in the standard way \cite{CSS3sel}.
Now, the relation $\edgeOrder$ can only relate two edges $(s,p),(s',p') \in \CSSedges$ satisfying
\begin{enumerate}
\item
    \label{itm:same-spec}
    $s$ and $s'$ have the same specificity,
\item
    \label{itm:same-name}
    we have $p \neq p'$ but the property names for $p$ and $p'$ are the same
    (e.g.\ %
     $p = \mbox{\texttt{color:blue}}$
     and
     $p' = \mbox{\texttt{color:red}}$
     with property name
     \texttt{color}), and
\item
    \label{itm:intersect}
        $s$ intersects with $s'$ (i.e. $\sem{s} \cap \sem{s'} \neq \emptyset$).
\end{enumerate}
If both (\ref{itm:same-spec}) and (\ref{itm:same-name}) are satisfied, Condition (\ref{itm:intersect}) can be checked by means of SMT-solver via the reduction in Theorem~\ref{thm:intersection}.
Supposing that Condition (\ref{itm:intersect}) holds, we simply compute the indices of the edges in the file: $m := \Index((s,p))$ and $m' := \Index((s',p'))$.
Recall $\Index(e)$ was defined formally in Section~\ref{sec:css2graph}.
We put $(s,p) \edgeOrder (s',p')$ iff $m < m'$.
There are two minor technical details with the keyword \verb+!important+ and \emph{shorthand property names}; see Appendix~\ref{sec:important}.

The ordering given above unfortunately turns out to be too conservative.
In the following, we give an example to demonstrate this, and propose a refinement to the ordering.
Consider the CSS file
\begin{center}
    \begin{minted}{css}
        .a { color:red; color:rgba(255,0,0,0.5) }
        .b { color:red; color:rgba(255,0,0,0.5) }
    \end{minted}
\end{center}
In this file, both nodes matching
\texttt{.a} and \texttt{.b}
are assigned a semi-transparent red with solid red being defined as a
\emph{fallback} when the semi-transparent red is not supported.
If the edge order is calculated as above, we obtain
\begin{equation}
\label{eq:bad-order}
    (\mbox{\texttt{.a}},
     \mbox{\texttt{color:rgba(255,0,0,0.5)}})
    \edgeOrder
    (\mbox{\texttt{.b}},
     \mbox{\texttt{color:red}})
\end{equation}
which prevents the obvious rule-merging
\begin{center}
    \begin{minted}{css}
        .a, .b { color:red; color:rgba(255,0,0,0.5) }
    \end{minted}
\end{center}
The key observation is that the fact that we also have
\begin{equation}
\label{eq:good-order}
    (\mbox{\texttt{.b}},
     \mbox{\texttt{color:red}})
    \edgeOrder
    (\mbox{\texttt{.b}},
     \mbox{\texttt{color:rgba(255,0,0,0.5)}})
\end{equation}
renders any violations of (\ref{eq:bad-order}) benign:
such a violation would give precedence to right-hand declaration
\texttt{color:rgba(255,0,0,0.5)}
over
\texttt{color:red}
for nodes matching both
\texttt{.a}
and
\texttt{.b}.
However, because of (\ref{eq:good-order}) this should happen anyway and we can omit (\ref{eq:bad-order}) from $\edgeOrder$.

Formally, the ordering we need is as follows.
If Conditions (\ref{itm:same-spec}-\ref{itm:intersect}) hold, we compute
    $m := \Index((s,p))$ and
    $m' := \Index((s',p'))$
and put $(s,p) \edgeOrder (s',p')$ iff
\begin{itemize}
\item
    $m < m'$, and
\item
    $(s', p)$ does not exist or $\Index((s', p)) < m'$.
\end{itemize}
That is, we omit
$(s,p) \edgeOrder (s',p')$
if $(s', p)$ appears later in the CSS file (that is, we have %
$\Index((s', p')) < \Index((s', p))$).
Note, we are guaranteed in this latter case to include
$(s', p') \edgeOrder (s', p)$
since $(s', p')$ and $(s', p)$ can easily be seen to satisfy the conditions for introducing
$(s', p') \edgeOrder (s', p)$.

%% file: intersection.tex
\section{More Details on Solving Selector Intersection Problem}
\label{sec:intersection}

In the previous section, we have given the intuition behind the
efficient reduction from the CSS selector intersection problem to
quantifier-free theory over integer linear arithmetic, for which there is
a highly-optimised SMT-solver \cite{Z3}. In this section, we present this
reduction in full, which may be skipped on a first reading without affecting
the flow of the paper.

This section is structured as follows. We begin by defining CSS automata.
We then provide a semantic-preserving transformation from CSS selectors to CSS automata.
Next we show the closure of CSS automata under intersection.
The closure allows us to reduce the intersection problem of CSS automata to the non-emptiness problem of CSS automata.
Finally, we provide a reduction from non-emptiness
of CSS automata to satisfiability over quantifier-free integer linear
arithmetic. We will see that each such transformation/reduction runs in
polynomial-time, resulting in the complexity upper bound of NP, which is precise
due to the NP-hardness of the problem from the previous section.

\subsection{CSS Automata}

\input{css-automata-definition}

\subsection{Transforming CSS Selectors to CSS Automata}

\input{selector-to-automata}

\subsection{Closure Under Intersection}

\input{css-automata-intersection}

\subsection{Reducing Non-emptiness of CSS Automata to SMT-solving}

\input{reduction}

%% file: css-automata-definition.tex
CSS automata are a kind of finite automata which navigate the tree structure of a document tree.
Transitions of the automata will contain one of four labels:
$\arrchild$, $\arrneighbour$, $\arrsibling$, and $\arrlast$.
Intuitively, these transitions perform the following operations.
$\arrchild$ moves to the first child of the current node.
$\arrneighbour$ moves to the next sibling of the current node.
$\arrsibling$ moves an arbitrary number of siblings to the right.
Finally, $\arrlast$ reads the node matched by the automaton.
Since CSS does not have loops, we require only self loops in our automata, which are used to skip over nodes (e.g.\ %
$\isclass{\attval} \csssibling \isclass{\attval'}$
may pass over many nodes between those matching
$\isclass{\attval}$
and those matching
$\isclass{\attval'}$).
We do not allow $\arrneighbour$ to label a loop -- this is for the purposes of the \NP\ proof: it can be more usefully represented as $\arrsibling$.

An astute reader may complain that $\arrsibling$ does not need to appear on a loop since it can already pass over an arbitrary number of nodes.
However, the product construction used for intersection becomes easier if $\arrsibling$ appears only on loops.
There is no analogue of $\arrsibling$ for $\arrchild$ because we do not need it: the use of $\arrsibling$ is motivated by selectors such as
$\psnthchild{\coefa}{\coefb}$
which count the number of siblings of a node.
No CSS selector counts the number of descendants/ancestors.

\paragraph{Formal Definition of CSS Automata}
    A \defn{CSS Automaton} $\cssaut$ is a tuple
    $\tup{\astates, \atrans, \ainitstate, \afinstate}$
    where
        $\astates$ is a finite set of states,
        $\atrans \subseteq \astates \times
                           \set{\arrchild, \arrneighbour, \arrsibling, \arrlast} \times
                           \nodeselectors \times
                           \astates$ is a transition relation,
        $\ainitstate \in \astates$ is the initial state, and
        $\afinstate \in \astates$ is the final state.
        Moreover,
        \begin{compactenum}
        \item
            (only self-loops)
            there exists a partial order $\weakord$ such that
            $\tup{\astate, \arrgen, \csssim, \astate'} \in \atrans$
            implies
            $\astate' \weakord \astate$,
        \item
            ($\arrsibling$ loops and doesn't check nodes)
            for all
            $\tup{\astate, \arrsibling, \csssim, \astate'} \in \atrans$
            we have
            $\astate = \astate'$
            and
            $\csssim = \isany$,
        \item
            ($\arrneighbour$ doesn't label loops)
            for all
            $\tup{\astate, \arrgen, \csssim, \astate} \in \atrans$
            we have
            $\arrgen \neq \arrneighbour$
            and
            $\csssim = \isany$,
        \item
            ($\arrlast$ checks last node only)
            for all
            $\tup{\astate, \arrgen, \csssim, \astate'} \in \atrans$
            we have
            $\astate' = \afinstate$
            iff
            $\arrgen = \arrlast$, and
        \item
            ($\afinstate$ is a sink)
            for all
            $\tup{\astate, \arrgen, \csssim, \astate'} \in \atrans$
            we have
            $\astate \neq \afinstate$.
        \end{compactenum}

We now define the semantics of CSS automata, i.e., given an automaton
$\cssaut$, which language $\Lang(\cssaut)$ they recognise.
Intuitively, the set $\Lang(\cssaut)$ contains the set of pairs of document
tree and node, which the automaton $\cssaut$ accepts. We will now define this
more formally. Write $\astate \atran{\arrgen}{\csssim} \astate'$
to denote a transition
$\tup{\astate, \arrgen, \csssim, \astate'} \in \atrans$.
A document tree $\tree = \tup{\treedom, \treelab}$ and node $\node \in \treedom$
is \defn{accepted} by a CSS automaton $\cssaut$
if there exists a sequence
\[
    \astate_0, \node_0,
    \astate_1, \node_1,
    \ldots,
    \astate_\runlen, \node_\runlen,
    \astate_{\runlen+1}
    \in
    \brac{\astates \times \treedom}^\ast \times \set{\afinstate}
\]
such that
$\astate_0 = \ainitstate$ is the initial state,
$\node_0 = \eseq$ is the root node,
$\astate_{\runlen+1} = \afinstate$ is the final state,
$\node_\runlen = \node$ is the matched node, and
for all $\idxi$, there is some transition
$\astate_\idxi \atran{\arrgen}{\csssim} \astate_{\idxi+1}$
with $\node_\idxi$ satisfying $\csssim$ and if $\idxi \leq \runlen$,
\begin{compactenum}
\item
    if
    $\arrgen = \arrchild$
    then
        $\node_{\idxi+1} = \node_\idxi 1$ (i.e.,
        the leftmost child of $\node_\idxi$),

\item
    if
    $\arrgen = \arrneighbour$
    then there is some $\node'$ and $\treedir$ such that
    $\node_\idxi = \node' \treedir$
    and
    $\node_{\idxi+1} = \node' (\treedir + 1)$, and

\item
    if
    $\arrgen = \arrsibling$
    then there is some $\node'$, $\treedir$ and $\treedir'$ such that
    $\node_\idxi = \node' \treedir$
    and
    $\node_{\idxi+1} = \node' \treedir'$
    and
    $\treedir' > \treedir$.
\end{compactenum}
Such a sequence is called an \emph{accepting run} of length $\runlen$.
The \defn{language $\Lang(\cssaut)$ recognised by $\cssaut$} is the set of
pairs $(\tree,\node)$ accepted by $\cssaut$.

%% file: selector-to-automata.tex
The following proposition shows that CSS automata are no less expressive than
CSS selectors.
\begin{proposition} \label{prop:css-to-aut}
    For each CSS selector $\css$, we may construct in
    polynomial-time a CSS automaton $\selaut{\css}$ such that
    $\Lang(\selaut{\css}) = \sem{\css}$.
\end{proposition}

We show this proposition by giving a translation from a given CSS selector to a CSS automaton.
Before the formal definition, we consider the a simple example.
A more complex example is shown after the translation.

\subsubsection{Simple Example}

Consider the selector
\[
    \isele{\mytt{p}} \cssneighbour  \isclass{\mytt{a}}
\]
which selects a node that has a class $\mytt{a}$ and is directly a right neighbour of a node with element $\mytt{p}$.
Figure~\ref{fig:css-automata-example-easy} gives a CSS automaton representing the selector.
The automaton begins with a loop that can navigate down any branch of the tree using the $\arrchild$ and $\arrsibling$ transitions from $\selstate{1}$.
Then, since it always moves from the first child to the last, it will first see the node with the $\mytt{p}$.
When reading this node, it will move to the next child using $\arrneighbour$ before matching the node with class $\mytt{a}$, leading to the accepting state.

\begin{figure}
\begin{center}
    \psset{colsep=12ex,rowsep=6ex,arcangle=45,nodesep=2mm,loopsize=.5}
    \begin{psmatrix}
        \\
        \rnode{N1}{$\selstate{1}$} &
        \rnode{N2}{$\selstate{2}$} &
        \rnode{N3}{$\afinstate$}
        \\
        \nccircle{->}{N1}{-3.75ex}\nbput{$\arrsibling, \arrchild$}\naput{$\isany$}
        \ncline{->}{N1}{N2}\naput{$\arrneighbour$}\nbput{$\isele{\mytt{p}}$}
        \ncline{->}{N2}{N3}\naput{$\arrlast$}\nbput{$\isclass{\mytt{a}}$}
    \end{psmatrix}
    \caption{\label{fig:css-automata-example-easy}CSS Automaton for
             $\isele{\mytt{p}} \cssneighbour  \isclass{\mytt{a}}$.}
\end{center}
\end{figure}

\subsubsection{Formal Translation}

Given a CSS selector $\css$, we define $\selaut{\css}$ as follows.
We can write $\css$ uniquely in the form
\[
    \csssim_1\ \genop_1\ \csssim_2\ \genop_2\ \cdots\ \genop_{\numof-1}\ \csssim_\numof
\]
where each $\csssim_\idxi$ is a node selector, and each
$\genop_\idxi \in \set{\cssdescendant,
                       \csschild,
                       \cssneighbour,
                       \csssibling}$.
We will have a state $\selstate{\idxi}$ corresponding to each $\csssim_\idxi, \genop_\idxi$.
We define
\[
    \selaut{\css} = \tup{\astates, \eles, \atrans, \selstate{1}, \afinstate}
\]
where
$\astates =
 \setcomp{\selstate{\idxi}, \midstate{\idxi}}
         {1 \leq \idxi \leq \numof}
 \uplus
 \set{\afinstate}$
and we define the transition relation $\atrans$ to contain the following transitions.
The initial and final transitions are used to navigate from the root of the tree to the node matched by $\csssim_1$, and to read the final node matched by $\csssim_\numof$ (and the selector as a whole) respectively.
That is,
$\selstate{1} \atran{\arrchild, \arrsibling}{\isany} \selstate{1}$
and
$\selstate{\numof} \atran{\arrlast}{\csssim_\numof} \afinstate$.
We have further transitions for $1 \leq \idxi < \numof$ that are shown in Figure~\ref{fig:aut-from-sel}.
The transitions connect $\selstate{\idxi}$ to $\selstate{\idxi+1}$ depending on $\genop_\idxi$.
Figure~\ref{fig:aut-child} shows the child operator.
The automaton moves to the first child of the current node and moves to the right zero or more steps.
Figure~\ref{fig:aut-descendant} shows the descendant operator.
The automaton traverses the tree downward and rightward any number of steps.
The neighbour operator is handled in Figure~\ref{fig:aut-neighbour} by simply moving to the next sibling.
Finally, the sibling operator is shown in Figure~\ref{fig:aut-sibling}.

\begin{figure}
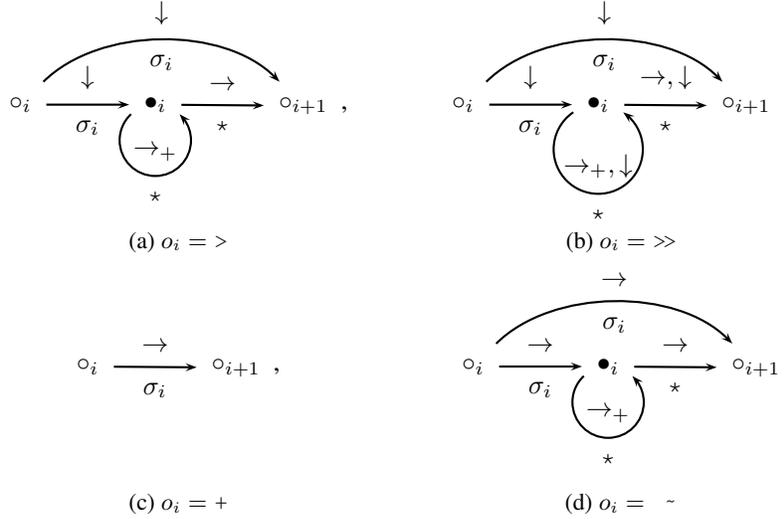

\centering
    \psset{rowsep=6ex,arcangle=45,nodesep=2mm,loopsize=.5}
    \begin{subfigure}{.4\textwidth}
    \centering
        \begin{psmatrix}
            \\
            \rnode{N1}{$\selstate{\idxi}$} &
            \rnode{N2}{$\midstate{\idxi}$} &
            \rnode{N3}{$\selstate{{\idxi+1}}$} \ ,
            \\
            \ncarc{->}{N1}{N3}\naput{$\arrchild$}\nbput{$\csssim_\idxi$}
            \ncline{->}{N1}{N2}\naput{$\arrchild$}\nbput{$\csssim_\idxi$}
            \nccircle{->}{N2}{-3ex}\nbput{$\arrsibling$}\naput{$\isany$}
            \ncline{->}{N2}{N3}\naput{$\arrneighbour$}\nbput{$\isany$}
        \end{psmatrix}
        \caption{\label{fig:aut-child}$\genop_\idxi = \csschild$}
    \end{subfigure}
    \begin{subfigure}{.4\textwidth}
    \centering
        \begin{psmatrix}
            \\
            \rnode{N1}{$\selstate{\idxi}$} &
            \rnode{N2}{$\midstate{\idxi}$} &
            \rnode{N3}{$\selstate{{\idxi+1}}$} \ ,
            \\
            \ncarc{->}{N1}{N3}\naput{$\arrchild$}\nbput{$\csssim_\idxi$}
            \ncline{->}{N1}{N2}\naput{$\arrchild$}\nbput{$\csssim_\idxi$}
            \ncline{->}{N2}{N3}\naput{$\arrneighbour,\arrchild$}\nbput{$\isany$}
            \nccircle{->}{N2}{-3.75ex}\nbput{$\arrsibling, \arrchild$}\naput{$\isany$}
        \end{psmatrix}
        \caption{\label{fig:aut-descendant}$\genop_\idxi = \cssdescendant$}
    \end{subfigure}
    \\
    \begin{subfigure}{.4\textwidth}
    \centering
        \begin{psmatrix}
            \\
            \rnode{N1}{$\selstate{\idxi}$} &
            \rnode{N2}{$\selstate{{\idxi+1}}$} \ ,
            \\
            \ncline{->}{N1}{N2}\naput{$\arrneighbour$}\nbput{$\csssim_\idxi$}
        \end{psmatrix}
        \caption{\label{fig:aut-neighbour}$\genop_\idxi = \cssneighbour$}
    \end{subfigure}
    \begin{subfigure}{.4\textwidth}
    \centering
        \begin{psmatrix}
            \\
            \rnode{N1}{$\selstate{\idxi}$} &
            \rnode{N2}{$\midstate{\idxi}$} &
            \rnode{N3}{$\selstate{{\idxi+1}}$}
            \\
            \ncarc{->}{N1}{N3}\naput{$\arrneighbour$}\nbput{$\csssim_\idxi$}
            \ncline{->}{N1}{N2}\naput{$\arrneighbour$}\nbput{$\csssim_\idxi$}
            \nccircle{->}{N2}{-3ex}\nbput{$\arrsibling$}\naput{$\isany$}
            \ncline{->}{N2}{N3}\naput{$\arrneighbour$}\nbput{$\isany$}
        \end{psmatrix}
        \caption{\label{fig:aut-sibling}$\genop_\idxi =\ \ \csssibling$}
    \end{subfigure}
    \caption{\label{fig:aut-from-sel}Converting selectors to CSS automata.}
\end{figure}

We prove the correctness of this construction in
Lemma~\ref{lem:aut-sound} (soundness) and Lemma~\ref{lem:aut-comp}
(completeness) in Appendix~\ref{sec:css-automata-proof}.

\subsubsection{Complex Example}

Figure~\ref{fig:css-automata-example-hard} gives an example of a CSS automaton representing the more complex selector
\[
    \isele{\mytt{div}} \cssdescendant \mytt{p} \csssibling \isclass{\mytt{b}}
\]
which selects a node that has class $\mytt{b}$, is a right sibling of a node with element $\mytt{p}$ and moreover is a descendent of a $\mytt{div}$ node.
This automaton again begins at $\selstate{1}$ and navigates until it finds the node with the $\mytt{div}$ element name.
The automaton can read this node in two ways.
The topmost transition covers the case where the $\mytt{p}$ node is directly below the $\mytt{div}$ node.
The lower transition allows the automaton to match the $\mytt{p}$ node and then use a loop to navigate to the descendent node that will match $\mytt{p}$.
Similarly, from $\selstate{2}$ the automaton can read a $\mytt{p}$ node and choose between immediately matching the node with class $\mytt{b}$ or navigating across several siblings (using the loop at state $\midstate{2}$) before matching $\mytt{.b}$ and accepting.

\begin{figure}
\centering
    \psset{colsep=12ex,rowsep=6ex,arcangle=45,nodesep=2mm,loopsize=.5}
    \begin{psmatrix}
        \\
        \rnode{N1}{$\selstate{1}$} &
        \rnode{N2}{$\midstate{1}$} &
        \rnode{N3}{$\selstate{2}$} &
        \rnode{N4}{$\midstate{2}$} &
        \rnode{N5}{$\selstate{3}$} &
        \rnode{N6}{$\afinstate$}
        \\
        \nccircle{->}{N1}{-3.75ex}\nbput{$\arrsibling, \arrchild$}\naput{$\isany$}
        \ncarc{->}{N1}{N3}\naput{$\arrchild$}\nbput{$\mytt{div}$}
        \ncline{->}{N1}{N2}\naput{$\arrchild$}\nbput{$\mytt{div}$}
        \ncline{->}{N2}{N3}\naput{$\arrneighbour,\arrchild$}\nbput{$\isany$}
        \nccircle{->}{N2}{-3.75ex}\nbput{$\arrsibling, \arrchild$}\naput{$\isany$}
        \ncarc{->}{N3}{N5}\naput{$\arrneighbour$}\nbput{$\mytt{p}$}
        \ncline{->}{N3}{N4}\naput{$\arrneighbour$}\nbput{$\mytt{p}$}
        \nccircle{->}{N4}{-3ex}\nbput{$\arrsibling$}\naput{$\isany$}
        \ncline{->}{N4}{N5}\naput{$\arrneighbour$}\nbput{$\isany$}
        \ncline{->}{N5}{N6}\naput{$\arrlast$}\nbput{$\mytt{.b}$}
    \end{psmatrix}
    \caption{\label{fig:css-automata-example-hard}CSS Automaton for
                    $\isele{\mytt{div}}
                     \cssdescendant
                     \mytt{p}\ \ 
                     \csssibling
                     \isclass{\mytt{b}}$}
\end{figure}

%% file: css-automata-intersection.tex
The problems of non-emptiness and intersection of CSS automata can be defined
in precisely the same way we defined them for CSS selectors.
One key property of CSS automata, which is not enjoyed by CSS selectors, is
the closure of their languages under intersection. This allows us to treat
the problem of intersection of CSS automata (i.e. the non-emptiness of
the intersection of two CSS automata languages) as the non-emptiness
problem (i.e. whether a given CSS automaton has an empty language).

\begin{proposition}
    Given two CSS automata $\cssaut_1$ and $\cssaut_2$, we may construct
    in polynomial-time an automaton $\cssaut_1 \cap \cssaut_2$ such that
    $\Lang(\cssaut_1) \cap \Lang(\cssaut_2) = \Lang(\cssaut_1 \cap
    \cssaut_2)$.
    \label{prop:aut-intersection}
\end{proposition}
The construction of the CSS automaton $\cssaut_1 \cap \cssaut_2$ is by a
variant of the standard product construction \cite{Vardi95} for finite-state
automata over finite words, which run the two given automata in parallel
synchronised by the input word. Our construction runs the two CSS automata
$\cssaut_1$ and $\cssaut_2$ in parallel synchronised by the path that they
traverse. We first proceed with the formal definition and give an example afterwards.

\subsubsection{Formal Definition of Intersection}

We first define the intersection of two node selectors.
Recall node selectors are of the form
$
    \csstype \cssconds
$
where
$\csstype \in
 \setcomp{
     \isany,
     \isanyns{\ns},
     \isele{\ele},
     \iselens{\ns}{\ele}
 }{
     \ns \in \nspaces \land \ele \in \eles
 }$.
The intersection of two node selectors
$\csstype_1 \cssconds_1$
and
$\csstype_2 \cssconds_2$
should enforce all properties defined in $\cssconds_1$ and $\cssconds_2$.
In addition, both selectors should be able to agree on the namespace and element name of the node, hence this part of the selector needs to be combined more carefully.
Thus, letting
$\cssconds = \cssconds_1 \cup \cssconds_2$.
we define
\[
    \csstype_1 \cssconds_1 \cap \csstype_2 \cssconds_2 =
    \begin{cases} %
        \csstype_2 \cssconds %
        & %
        \csstype_1 = \isany %
        \\ %
        \csstype_1 \cssconds %
        & %
        \csstype_2 = \isany %
        \\ %
        \csstype_2 \cssconds %
        & %
        \csstype_1 = \isanyns{\ns} \land %
        \brac{ %
            \begin{array}{c} %
                \csstype_2 = \isanyns{\ns} %
                \ \lor \\ %
                \csstype_2 = \iselens{\ns}{\ele} %
            \end{array}
        } %
        \\ %
        \iselens{\ns}{\ele} \cssconds %
        & %
        \csstype_1 = \isanyns{\ns} \land \csstype_2 = \isele{\ele} %
        \\ %
        \csstype_1 \cssconds %
        & %
        \csstype_1 = \iselens{\ns}{\ele} \land %
        \brac{ %
            \begin{array}{c} %
                \csstype_2 = \iselens{\ns}{\ele} %
                \ \lor \\ %
                \csstype_2 = \isele{\ele} %
                \ \lor \\ %
                \csstype_2 = \isanyns{\ns} %
            \end{array} %
        } %
        \\ %
        \csstype_2 \cssconds %
        & %
        \csstype_1 = \isele{\ele} \land %
        \brac{ %
            \csstype_2 = \iselens{\ns}{\ele} \lor %
            \csstype_2 = \isele{\ele} %
        } %
        \\ %
        \iselens{\ns}{\ele} \cssconds %
        & %
        \csstype_1 = \ele \land \csstype_2 = \isanyns{\ns} %
        \\ %
        \cssneg{\isany} %
        & %
        \text{otherwise.}
    \end{cases} %
\]

We now define the automaton $\cssaut_1 \cap \cssaut_2$.
The intersection automaton synchronises transitions that move in the same direction
(by $\arrchild$, $\arrneighbour$, $\arrsibling$)
or both agree to match the current node at the same time
(with $\arrlast$).
In addition, we observe that a $\arrsibling$ can be used by one automaton while the other uses $\arrneighbour$.
Given
\[
    \cssaut_1 = \tup{\astates_1,
                     \eles,
                     \atrans_1,
                     \ainitstate_1, \afinstate^1}
    \quad \text{and} \quad
    \cssaut_2 = \tup{\astates_2,
                     \eles,
                     \atrans_2,
                     \ainitstate_2, \afinstate^2}
\]
we define
\[
    \cssaut_1 \cap \cssaut_2 =
    \tup{\astates_1 \times \astates_2,
         \eles,
         \atrans,
         \tup{\ainitstate_1, \ainitstate_2},
         \tup{\afinstate^1, \afinstate^2}}
\]
where (letting $\arrgen$ range over
$\set{\arrneighbour, \arrsibling, \arrchild, \arrlast}$) we set
$\atrans =$
\[
    \begin{array}{c}
        \setcomp{\tup{\astate_1, \astate_2}
                 \atran{\arrgen}{\csssim_1 \cap \csssim_2}
                 \tup{\astate'_1, \astate'_2}}
                {\astate_1 \atran{\arrgen}{\csssim_1} \astate'_1
                 \land
                 \astate_2 \atran{\arrgen}{\csssim_2} \astate'_2}
        \ \cup \\
        \setcomp{\tup{\astate_1, \astate_2}
                 \atran{\arrneighbour}{\csssim_1}
                 \tup{\astate'_1, \astate_2}}
                {\astate_1 \atran{\arrneighbour}{\csssim_1} \astate'_1
                 \land
                 \astate_2 \atran{\arrsibling}{\isany} \astate_2}
        \ \cup \\
        \setcomp{\tup{\astate_1, \astate_2}
                 \atran{\arrneighbour}{\csssim_2}
                 \tup{\astate_1, \astate'_2}}
                {\astate_1 \atran{\arrsibling}{\isany} \astate_1
                 \land
                 \astate_2 \atran{\arrneighbour}{\csssim_2} \astate'_2} \ .
    \end{array}
\]
We verify that this transition relation satisfies the appropriate conditions:
\begin{compactenum}
\item
    (only self-loops)
    for a contradiction, a loop in $\atrans$ that is not a self-loop can be projected to a loop of $\cssaut_1$ or $\cssaut_2$ that is also not a self-loop,
    e.g.\ if there exists
    $\tup{\astate_1, \astate_2}
     \atran{\arrgen}{\csssim}
     \tup{\astate'_1, \astate'_2}
     \atran{\arrgen'}{\csssim'}
     \tup{\astate_1, \astate_2}$
    that is not a self-loop, then either
    $\astate_1 \neq \astate'_1$
    or
    $\astate_2 \neq \astate'_2$,
    and thus we have a loop
    from $\astate_1$ to $\astate'_1$ to $\astate_1$ in $\cssaut_1$ or similarly for $\cssaut_2$,

\item
    ($\arrsibling$ loops and doesn't check nodes)
    $\arrsibling$ transitions are built from $\arrsibling$ transitions in $\cssaut_1$ and $\cssaut_2$, thus a violation in the intersection implies a violation in one of the underlying automata,

\item
    ($\arrneighbour$ doesn't label loops)
    $\arrneighbour$ transitions are built from at least one $\arrneighbour$ transition in $\cssaut_1$ or $\cssaut_2$, thus a violation in the intersection implies a violation in one of the underlying automata,

\item
    ($\arrlast$ checks last node only)
    similarly, a violation of this constraint in the intersection implies a violation in one of the underlying automata,

\item
    ($\afinstate$ is a sink)
    again, a violation of this constraint in the intersection implies a violation in the underlying automata.
\end{compactenum}

\subsubsection{Example of Intersection}

Recall the automaton in Figure~\ref{fig:css-automata-example-easy}
(equivalent to
$\isele{\mytt{p}} \cssneighbour  \isclass{\mytt{a}}$)
and the automaton in Figure~\ref{fig:css-automata-example-hard}
(equivalent to
$\isele{\mytt{div}} \cssdescendant \mytt{p} \csssibling \isclass{\mytt{b}}$).
The intersection of the two automata is given in Figure~\ref{fig:automata-intersection-example}.
Each state is a tuple
$\tup{\astate_1, \astate_2}$
where
    the first component $\astate_1$ represents the state of the automaton equivalent to
    $\isele{\mytt{p}} \cssneighbour  \isclass{\mytt{a}}$
    and the second component $\astate_2$ the automaton equivalent to
    $\isele{\mytt{div}} \cssdescendant \mytt{p} \csssibling \isclass{\mytt{b}}$.

In this example, accepting runs of the automaton will use only the top row of states.
The lower states are reached when the two automata move out of sync and can no longer reach agreement on the final node matched.
Usually this is by the first automaton matching a node labelled $\mytt{p}$, after which it must immediately accept the neighbouring node.
This leaves the second automaton unable find a match.
Hence, the first automaton needs to stay in state $\selstate{1}$ until the second has reached a near-final state.
Note, the two automata need not match the same node with element name $\mytt{p}$.

\begin{figure}
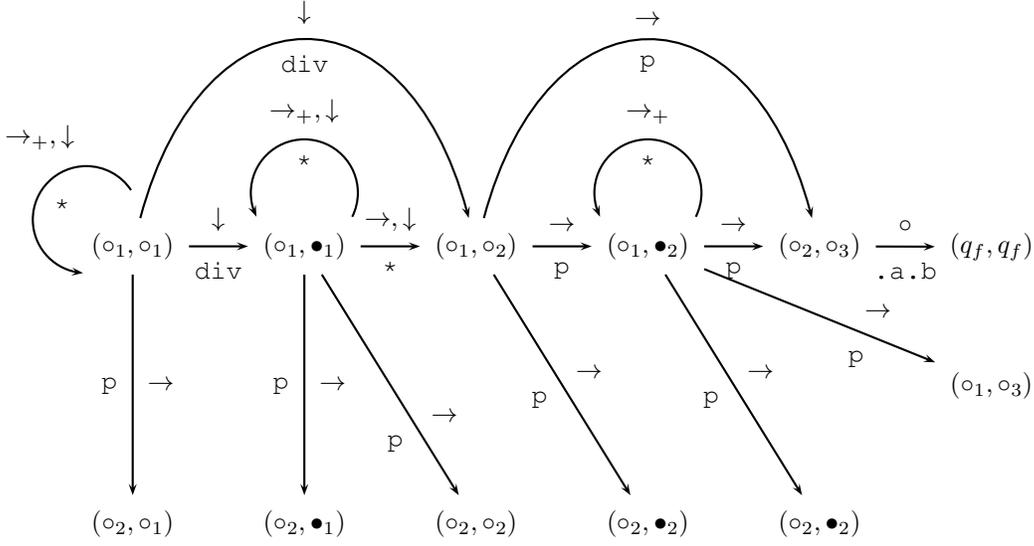

\centering
    \qquad
    \psset{colsep=7.5ex,rowsep=9ex,arcangle=45,nodesep=2mm,loopsize=.5}
    \begin{psmatrix}
        \\
        \\
        \rnode{N1}{$\tup{\selstate{1}, \selstate{1}}$} &
        \rnode{N2}{$\tup{\selstate{1}, \midstate{1}}$} &
        \rnode{N3}{$\tup{\selstate{1}, \selstate{2}}$} &
        \rnode{N4}{$\tup{\selstate{1}, \midstate{2}}$} &
        \rnode{N5}{$\tup{\selstate{2}, \selstate{3}}$} &
        \rnode{N6}{$\tup{\afinstate, \afinstate}$}
        \\
        & & & & &
        \rnode{N9}{$\tup{\selstate{1}, \selstate{3}}$}
        \\
        \rnode{N10}{$\tup{\selstate{2}, \selstate{1}}$} &
        \rnode{N11}{$\tup{\selstate{2}, \midstate{1}}$} &
        \rnode{N12}{$\tup{\selstate{2}, \selstate{2}}$} &
        \rnode{N7}{$\tup{\selstate{2}, \midstate{2}}$} &
        \rnode{N8}{$\tup{\selstate{2}, \midstate{2}}$} &
        \\
        \nccircle[angleA=60]{->}{N1}{4.5ex}\nbput[npos=.4]{$\arrsibling, \arrchild$}\naput{$\isany$}
        \ncarc[ncurv=1.5,arcangle=75]{->}{N1}{N3}\naput{$\arrchild$}\nbput{$\mytt{div}$}
        \ncline{->}{N1}{N2}\naput{$\arrchild$}\nbput{$\isele{\mytt{div}}$}
        \ncline{->}{N2}{N3}\naput{$\arrneighbour,\arrchild$}\nbput{$\isany$}
        \nccircle{->}{N2}{4.5ex}\nbput{$\arrsibling, \arrchild$}\naput{$\isany$}
        \ncarc[ncurv=1.5,arcangle=75]{->}{N3}{N5}\naput{$\arrneighbour$}\nbput{$\isele{\mytt{p}}$}
        \ncline{->}{N3}{N4}\naput{$\arrneighbour$}\nbput{$\isele{\mytt{p}}$}
        \nccircle{->}{N4}{4.5ex}\nbput{$\arrsibling$}\naput{$\isany$}
        \ncline{->}{N4}{N5}\naput{$\arrneighbour$}\nbput{$\isele{\mytt{p}}$}
        \ncline{->}{N5}{N6}\naput{$\arrlast$}\nbput{$\isclass{\mytt{a}}\isclass{\mytt{b}}$}
        \ncline{->}{N1}{N10}\naput{$\arrneighbour$}\nbput{$\isele{\mytt{p}}$}
        \ncline{->}{N2}{N11}\naput{$\arrneighbour$}\nbput{$\isele{\mytt{p}}$}
        \ncline{->}{N2}{N12}\naput[npos=.7]{$\arrneighbour$}\nbput[npos=.7]{$\isele{\mytt{p}}$}
        \ncline{->}{N3}{N7}\naput{$\arrneighbour$}\nbput{$\isele{\mytt{p}}$}
        \ncline{->}{N4}{N8}\naput{$\arrneighbour$}\nbput{$\isele{\mytt{p}}$}
        \ncline{->}{N4}{N9}\naput[npos=.7]{$\arrneighbour$}\nbput[npos=.7]{$\isele{\mytt{p}}$}
    \end{psmatrix}
    \caption{\label{fig:automata-intersection-example}The intersection of the automaton in Figure~\ref{fig:css-automata-example-easy} (in the first component) and the automaton in Figure~\ref{fig:css-automata-example-hard} (in the second component).}
\end{figure}

%% file: reduction.tex
We will now provide a polynomial-time reduction from the non-emptiness of a
CSS automaton to satisfiability of quantifier-free theory over integer linear
arithmetic.
That is, given a CSS automaton $\cssaut$, our algorithm constructs a
quantifier-free
formula $\presof{\cssaut}$ over integer linear arithmetic such that $\cssaut$
recognises a non-empty language iff $\presof{\cssaut}$ is satisfiable.
The encoding is quite involved and requires three small model properties discussed earlier.
Once we have these properties we can construct the required formula of the quantifier-free theory over linear arithmetic.
We begin by discussing each of these properties in turn, and then provide the reduction.
The reduction is presented in a number of stages.
We show how to handle attribute selectors separately before handling the encoding of CSS automata.
The encoding of CSS automata is further broken down:
we first describe the variables used in the encoding, then we describe how to handle node selectors, finally we put it all together to encode runs of a CSS automaton.

For the remainder of the section, we fix a CSS automaton
$\cssaut = \tup{\astates, \atrans, \ainitstate, \afinstate}$ and show
how to construct the formula $\presof{\cssaut}$ in polynomial-time.

\subsubsection{Bounded Run Length}

The first property required is that the length of runs can be bounded.
That is, if the language of $\cssaut$ is not empty, there is an accepting run over some tree whose length is smaller than the derived bound.
We will construct a formula that will capture all runs of length up to this bound.
Thanks to the bound we know that if an accepting run exists, the formula will encode at least one, and hence be satisfiable.

\begin{namedproposition}{prop:boundedlen}{Bounded Runs}
    Given a CSS Automaton
    $\cssaut = \tup{\astates, \atrans, \ainitstate, \afinstate}$,
    if $\ap{\Lang}{\cssaut} \neq \emptyset$, there exists
    $\aaccepts{\tree}{\node}{\cssaut}$
    with an accepting run of length $\sizeof{\atrans}$.
\end{namedproposition}

This proposition is straightforward to obtain.
We exploit that any loop in the automaton is a self loop that only needs to be taken at most once.
For loops labelled $\arrchild$, a CSS formula cannot track the depth in the tree, so repeated uses of the loop will only introduce redundant nodes.
For loops labelled $\arrsibling$, selectors such as
$\psnthchild{\coefa}{\coefb}$
may enforce the existence of a number of intermediate nodes.
But, since $\arrsibling$ can cross several nodes, such loops also only needs to be taken once.
Hence, each transition only needs to appear once in an accepting run.
That is, if there is an accepting run of a CSS automaton with $\numof$ transitions, there is also an accepting run of length at most $\numof$.

\subsubsection{Bounding Namespaces and Elements}

It will also be necessary for us to argue that the number of namespaces and elements can be bounded linearly in the size of the automaton.
This is because our formula will need keep track of the number of nodes of each type appearing in the tree.
This is required for encoding, e.g., the pseudo-classes of the form
$\psnthoftype{\coefa}{\coefb}$.
By bounding the number of types, our formula can use a bounded number of variables to store this information.

We state the property below.
The proof is straightforward and appears in Appendix~\ref{sec:bounded-types-proof}.
Intuitively, since only a finite number of nodes can be directly inspected by a CSS automaton, all others can be relabelled to a dummy type unless their type matches one of the inspected nodes.

\begin{namedproposition}{prop:boundedtypes}{Bounded Types}
    Given a CSS Automaton
    $\cssaut = \tup{\astates, \atrans, \ainitstate, \afinstate}$
    if there exists
    $\aaccepts{\tree}{\node}{\cssaut}$
    with
    $\tree = \tup{\treedom,\treelab}$,
    then there exists some
    $\aaccepts{\tree'}{\node}{\cssaut}$
    where
    $\tree' = \tup{\treedom, \treelab'}$.
    Moreover, let $\finof{\eles}$ be the set of element names and $\finof{\nspaces}$ be the set of namespaces appearing in the image of $\treelab'$.
    Both the size of $\finof{\eles}$ and the size of $\finof{\nspaces}$ are bounded linearly in the size of $\cssaut$.
\end{namedproposition}

\subsubsection{Bounding Attribute Values}

We will need to encode the satisfiability of conjunctions of attribute selectors.
This is another potential source of unboundedness because the values are strings of arbitrary length.
We show that, in fact, if the language of the automaton is not empty, there there is a solution whose attribute values are strings of a length less than a bound polynomial in the size of the automaton.

The proof of the following lemma is highly non-trivial and uses techniques inspired by results in Linear Temporal Logic and automata theory.
To preserve the flow of the article, we present the proof in Appendix~\ref{sec:poly-string-solution}.

\begin{namedproposition}{prop:boundedatts}{Bounded Attributes}
    Given a CSS Automaton
    $\cssaut = \tup{\astates, \atrans, \ainitstate, \afinstate}$
    if there exists
    $\aaccepts{\tree}{\node}{\cssaut}$
    with
    $\tree = \tup{\treedom,\treelab}$,
    then there exists some bound $\attvalbound$ polynomial in the size of $\cssaut$ and some
    $\aaccepts{\tree'}{\node}{\cssaut}$
    where the length of all attribute values in $\tree'$ is bound by $\attvalbound$.
\end{namedproposition}

\matt{Right term for Presburger?}
Given such a bound on the length of values, we can use quantifier-free Presburger formulas to ``guess'' these witnessing strings by using a variable for each character position in the string.
Then, the letter in each position is encoded by a number.
This process is discussed in the next section.

\subsubsection{Encoding Attribute Selectors}

Before discussing the full encoding, we first show how our formula can encode attribute selectors.
Once we have this encoding, we can invoke it as a sub-routine of our main encoding whenever we have to handle attribute selectors.
It is useful for readability reasons to present this in its own section.

The first key observation is that we can assume each positive attribute selector that does not specify a namespace applies to a unique, fresh, namespace.
Thus, these selectors do not interact with any other positive attribute selectors and we can handle them easily.
Note, these fresh namespaces do not appear in
$\finof{\nspaces}$.

We present our encoding which works by identifying combinations of attribute selectors that must apply to the same attribute value.
That is, we discover how many attribute values are needed, and collect together all selectors that apply to each selector.
To that end, let $\attop$ range over the set of operators
$\set{\opis, \ophas, \opbegin, \opstrbegin, \opstrend, \opstrsub}$
and let $\csstype\cssconds$ be a node selector.
For each $\ns$ and $\att$, let
$\cssconds^\ns_\att$
be the set of conditions in $\cssconds$ of the form $\csscond$ or
$\cssneg{\csscond}$
where $\csscond$ is of the form
$\hasattns{\ns}{\att}$
or
$\opattns{\ns}{\att}{\attop}{\attval}$.
Recall we are encoding runs of a CSS automaton of length at most $\numof$.
For a given position $\idxi$ in the run, we define
$\attspres{\csstype\cssconds}{\idxi}$
to be the conjunction of the following constraints, where the encoding for
$\attspresns{\ns}{\att}{\cssconds}{\idxi}$
is presented below.
\newcommand\setneg[2]{\ap{\text{Neg}}{#1, #2}}
Since a constraint of the form
$\cssneg{\opatt{\att}{\attop}{\attval}}$
applies to all attributes $\att$ regardless of their namespace, we define for convenience
$\setneg{\ns}{\att} =
 \setcomp{\cssneg{\opattns{\ns}{\att}{\attop}{\attval}}}
                 {\cssneg{\opatt{\att}{\attop}{\attval}} \in \cssconds}$.

\begin{compactitem}
\item
    For each $\ns$ and $\att$ with $\cssconds^\ns_\att$ non-empty and containing at least one selector of the form
    $\hasattns{\ns}{\att}$
    or
    $\opattns{\ns}{\att}{\attop}{\attval}$,
    we enforce
    \[
        \attspresns{\ns}{\att}{
                \cssconds^\ns_\att
                \cup
                \setneg{\ns}{\att}
        }{\idxi}
    \]
    if
    $\cssneg{\hasatt{\att}} \notin \cssconds$ and
    $\cssneg{\hasattns{\ns}{\att}} \notin \cssconds$,
    else we assert false.

\item
    For each
    $\hasatt{\att} \in \cssconds$,
    let $\ns$ be fresh namespace.
    We assert
    \[
        \attspresns{\ns}{\att}{
                \set{\hasattns{\ns}{\att}}
                \cup
                \setneg{\ns}{\att}
        }{\idxi}
    \]
    and for each
    $\opatt{\att}{\attop}{\attval} \in \cssconds$
    we assert
    \[
        \attspresns{\ns}{\att}{
                \set{\opattns{\ns}{\att}{\attop}{\attval}}
                \cup
                \setneg{\ns}{\att}
        }{\idxi}
    \]
    whenever, in both cases,
    $\cssneg{\hasatt{\att}} \notin \cssconds$.
    If
    $\cssneg{\hasatt{\att}} \in \cssconds$
    in both cases we assert false.
\end{compactitem}

It remains to encode
$
    \attspresns{\ns}{\att}{\consset}{\idxi}
$
for some set of attribute selectors $\consset$ all applying to $\ns$ and $\att$.

We can obtain a polynomially-sized global bound $(\attvalbound - 1)$ on the length of any satisfying value of an attribute $\qatt{\ns}{\att}$ at some position $\idxi$ of the run from \refproposition{prop:boundedatts}\footnote{%
    Of course, we could obtain individual bounds for each $\ns$ and $\att$ if we wanted to streamline the encoding.
}.
Finally, we increment the bound by one to allow space for a trailing null character.

Once we have a bound on the length of a satisfying value, we can introduce variables
$\wordpos{\ns}{\att}{\idxi}{1},
 \ldots,
 \wordpos{\ns}{\att}{\idxi}{\attvalbound}$
for each character position of the satisfying value, and encode the constraints almost directly.
That is, letting $\csscond$ range over positive attribute selectors, we define\footnote{
    Note, we allow negation in this formula.  This is for convenience only as
    the formulas we negate can easily be transformed into existential Presburger.
}
\[
    \begin{array}{rcl} %
        \attspresns{\ns}{\att}{\consset}{\idxi} %
        &=& %
        \bigwedge\limits_{\csscond \in \consset} %
            \attspres{\csscond}{\wordposvec} %
        \ \land \\ %
        & & %
        \bigwedge\limits_{\cssneg{\csscond} \in \consset} %
            \neg \attspres{\csscond}{\wordposvec} %
        \ \land %
        \attspresnulls{\wordposvec} \ .%
    \end{array} %
\]
where
$\wordposvec
 =
 \wordpos{\ns}{\att}{\idxi}{1},
 \ldots,
 \wordpos{\ns}{\att}{\idxi}{\attvalbound}$
will be existentially quantified later in the encoding and whose values will range\footnote{
    Strictly speaking, Presburger variables range over natural numbers.
    It is straightforward to range over a finite number of values.
    That is, we can assume, w.l.o.g. that
    $\alphabet \uplus {\nullch} \subseteq \N$
    and the quantification is suitably restricted.
}
over $\alphabet \uplus \set{\nullch}$
where $\nullch$ is a null character used to pad the suffix of each word.
We define $\attspres{\csscond}{\wordposvec}$ for several $\csscond$, the
rest can be defined in the same way
(see Appendix~\ref{sec:attspres-missing}).
Letting
$\attval = \cha_1\ldots\cha_\numofalt$,
\[
    \begin{array}{rcl} %
        \attspres{\hasattns{\ns}{\att}}{\wordposvec} %
        &=& %
        \top %
        \\ %
        \attspres{\attisns{\ns}{\att}{\attval}}{\wordposvec} %
        &=& %
        \bigwedge\limits_{1 \leq \idxi \leq \numofalt} %
            \wordpos{\ns}{\att}{\idxi}{\idxj} = \cha_\idxj %
        \land %
        \wordpos{\ns}{\att}{\idxi}{\numofalt+1} = \nullch %
        \\ %
        \attspres{\attstrbeginns{\ns}{\att}{\attval}}{\wordposvec} %
        &=& %
        \bigwedge\limits_{1 \leq \idxj \leq \numofalt} %
            \wordpos{\ns}{\att}{\idxi}{\idxj} = \cha_\idxj %
        \\ %
        \attspres{\attstrsubns{\ns}{\att}{\attval}}{\wordposvec} %
        &=& %
        \bigvee\limits_{0 \leq \idxj \leq \attvalbound - \numofalt - 1} %
            \bigwedge\limits_{1 \leq \idxj' \leq \numofalt} %
            \wordpos{\ns}{\att}{\idxi}{\idxj + \idxj'} = \cha_\idxj %
    \end{array} %
\]
Finally, we enforce correct use of the null character
\[
    \attspresnulls{\wordposvec}
    =
    \bigvee\limits_{1 \leq \idxj \leq \attvalbound}
    \bigwedge\limits_{\idxj \leq \idxj' \leq \attvalbound}
        \wordpos{\ns}{\att}{\idxi}{\idxj'} = \nullch \ .
\]

\subsubsection{Encoding Non-Emptiness}

We are now ready to give the main encoding of the emptiness of a CSS automaton using the quantifier-free theory over integer linear arithmetic.
\matt{
    This is a weird thing to say because we use $\exists$ all the time...
}
This encoding makes use of a number of variables, which we explain intuitively below.
After describing the variables, we give the encoding in two parts:
first we explain how a single node selector can be translated into existential Presuburger arithmetic.
Once we have this translation, we give the final step of encoding a complete run of an automaton.

\paragraph{Variables Used in the Encoding}
Our encoding makes use of the following variables for $0 \leq \idxi \leq \numof$, representing the node at the $\idxi$th step of the run.
We use the overline notation to indicate variables.
\begin{compactitem}
\item
    $\astatevar{\idxi}$, taking any value in $\astates$, indicating the state of the automaton when reading the $\idxi$th node in the run,
\item
    $\nsvar{\idxi}$,
    taking any value in
    $\finof{\nspaces}$
    indicating the element tag (with namespace) of the $\idxi$th node read in the run,
\item
    $\elevar{\idxi}$,
    taking any value in
    $\finof{\eles}$
    indicating the element tag (with namespace) of the $\idxi$th node read in the run,
\item
    $\pclsvar{\idxi}{\pcls}$,
    for each pseudo-class
    $\pcls \in \pclss \setminus \set{\psroot}$
    indicating that the $\idxi$th node has the pseudo-class $\pcls$,
\item
    $\numvar{\idxi}$, taking a natural number indicating that the $\idxi$th node is the $\numvar{\idxi}$th child of its parent, and
\item
    $\numtypevar{\idxi}{\qele{\ns}{\ele}}$,
    for all
    $\ns \in \finof{\nspaces}$
    and
    $\ele \in \finof{\eles}$,
    taking a natural number variable indicating that there are
    $\numtypevar{\idxi}{\qele{\ns}{\ele}}$
    nodes of type $\qele{\ns}{\ele}$ strictly preceding the current node in the sibling order, and
\item
    $\lnumvar{\idxi}$, taking a natural number indicating that the current node is the $\lnumvar{\idxi}$th to last child of its parent, and
\item
    $\lnumtypevar{\idxi}{\qele{\ns}{\ele}}$,
    for all
    $\ns \in \finof{\nspaces}$
    and
    $\ele \in \finof{\eles}$,
    taking a natural number variable indicating that there are
    $\lnumtypevar{\idxi}{\qele{\ns}{\ele}}$
    nodes of type $\qele{\ns}{\ele}$ strictly following the current node in the sibling order, and
\item
    $\wordpos{\ns}{\att}{\idxi}{\idxj}$ as used in the previous section for encoding the character at the $\idxj$th character position of the attribute value for $\qatt{\ns}{\att}$ at position $\idxi$ in the run\footnote{
        Recall $\ns$ is not necessarily in
        $\finof{\nspaces}$
        as it may be some fresh value.
    }.
\end{compactitem}
Note, we do not need a variable for $\psroot$ since it necessarily holds uniquely at the $0$th position of the run.

\paragraph{Encoding Node Selectors}
We define the encoding of node selectors below using the variables defined in the previous section.
Note, this translation is not correct in isolation: global constraints such as ``no ID appears twice in the tree'' will be enforced later.
The encoding works by translating each part of the selector directly.
For example, the constraint $\isele{\ele}$ simply checks that
$\elevar{\idxi} = \ele$.
Even in the more complex cases of selectors such as
$\psnthchild{\coefa}{\coefb}$
we are able to use a rather direct translation of the semantics:
$\exists \nvar . \xvar = \coefa \nvar + \coefb$.
For the case of
$\psnthoftype{\coefa}{\coefb}$
we have to consider all possible namespaces $\ns$ and element names $\ele$ that the node could take, and use the
$\numtypevar{\idxi}{\qele{\ns}{\ele}}$
variables to do the required counting.

In our presentation we allow ourselves to negate existentially quantified formulas of the form
$
    \exists \nvar . \xvar = \coefa \nvar + \coefb
$
where $\xvar$ is a variable, and $\coefa$ and $\coefb$ are constants.
Although this is not strictly allowed in existential Presburger arithmetic, it is not difficult to encode correctly.
For completeness, we provide the encoding of such negated formulas in Appendix~\ref{sec:negating-pos-formulas}.

In the following, let
$\noatts{\cssconds}$
be $\cssconds$ \emph{less} all selectors of the form
$\hasattns{\ns}{\att}$,
$\opattns{\ns}{\att}{\attop}{\attval}$,
$\hasatt{\att}$,
or
$\opatt{\att}{\attop}{\attval}$,
or
$\cssneg{\hasattns{\ns}{\att}}$,
$\cssneg{\opattns{\ns}{\att}{\attop}{\attval}}$,
$\cssneg{\hasatt{\att}}$,
or
$\cssneg{\opatt{\att}{\attop}{\attval}}$.

\begin{definition}[$\csssimpres{\csssim}{\idxi}$]
    Given a node selector $\csstype\cssconds$, we define
    \[
        \csssimpres{\csstype\cssconds}{\idxi} %
        = %
        \brac{ %
            \begin{array}{c} %
                \csssimpres{\csstype}{\idxi} %
                \land \\ %
                \brac{ %
                    \bigwedge\limits_{\csscond \in \noatts{\cssconds}} %
                        \csssimpres{\csscond}{\idxi} %
                } %
                \land \\ %
                \attspres{\csstype\cssconds}{\idxi} %
            \end{array} %
        } %
    \]
    where we define
    $\csssimpres{\csscond}{\idxi}$
    as follows:
    \allowdisplaybreaks
    \[
        \begin{array}{rcl} %
            \csssimpres{\isany}{\idxi} %
            &=& %
            \ptrue %
            \\ %
            \csssimpres{\isanyns{\ns}}{\idxi} %
            &=& %
            \brac{\nsvar{\idxi} = \ns} %
            \\ %
            \csssimpres{\isele{\ele}}{\idxi} %
            &=& %
            \brac{\elevar{\idxi} = \ele} %
            \\ %
            \csssimpres{\iselens{\ns}{\ele}}{\idxi} %
            &=& %
            \brac{ %
                \nsvar{\idxi} = \ns %
                \land %
                \elevar{\idxi} = \ele %
            } %
            \\ %
            \csssimpres{\cssneg{\csssimnoneg}}{\idxi} %
            &=& %
            \neg \csssimpres{\csssimnoneg}{\idxi} %
            \\ %
            \csssimpres{\psroot}{\idxi} %
            &=& %
            \begin{cases} %
                \ptrue  & \idxi = 0 %
                \\ %
                \pfalse & \text{otherwise} %
            \end{cases} %
            \\ %
            \forall \pcls \in \pclss \setminus \set{\psroot} \ .\ 
            \csssimpres{\pcls}{\idxi} %
            &=& %
            \pclsvar{\idxi}{\pcls} %
        \end{array} %
    \]
    and, finally, for the remaining selectors, we have
    \[
        \begin{array}{rcl}
            \csssimpres{\psnthchild{\coefa}{\coefb}}{0}
            &=&
            \pfalse
            \\
            \csssimpres{\psnthlastchild{\coefa}{\coefb}}{0}
            &=&
            \pfalse
            \\
            \csssimpres{\psnthoftype{\coefa}{\coefb}}{0}
            &=&
            \pfalse
            \\
            \csssimpres{\psnthlastoftype{\coefa}{\coefb}}{0}
            &=&
            \pfalse
            \\
            \csssimpres{\psonlychild}{0}
            &=&
            \pfalse
            \\
            \csssimpres{\psonlyoftype}{0}
            &=&
            \pfalse
        \end{array}
    \]
    and when $\idxi > 0$
    \[
        \begin{array}{rcl}
            \csssimpres{\psnthchild{\coefa}{\coefb}}{\idxi} %
            &=& %
            \exists \nvar . %
                \numvar{\idxi} = \coefa \nvar + \coefb %
            \\
            \csssimpres{\psnthlastchild{\coefa}{\coefb}}{\idxi} %
            &=& %
            \exists \nvar . %
                \lnumvar{\idxi} = \coefa \nvar + \coefb %
            \\
            \csssimpres{\psnthoftype{\coefa}{\coefb}}{\idxi} %
            &=&  %
            \bigvee\limits_{\substack{ %
                \ns \in \finof{\nspaces} %
                \\ %
                \ele \in \finof{\eles} %
            }}\brac{ %
                \begin{array}{c} %
                    \nsvar{\idxi} = \ns \land %
                    \elevar{\idxi} = \ele\ \land %
                    \\ %
                    \exists \nvar . %
                        \numtypevar{\idxi}{\qele{\ns}{\ele}} + 1 %
                        = %
                        \coefa \nvar + \coefb %
                \end{array} %
            } %
            \\
            \csssimpres{\psnthlastoftype{\coefa}{\coefb}}{\idxi} %
            &=& %
            \bigvee\limits_{\substack{ %
                \ns \in \finof{\nspaces} %
                \\ %
                \ele \in \finof{\eles} %
            }}\brac{ %
                \begin{array}{c} %
                    \nsvar{\idxi} = \ns \land %
                    \elevar{\idxi} = \ele\ \land %
                    \\ %
                    \exists \nvar . %
                        \lnumtypevar{\idxi}{\qele{\ns}{\ele}} + 1 %
                        = %
                        \coefa \nvar + \coefb %
                \end{array} %
            } %
            \\
            \csssimpres{\psonlychild}{\idxi} %
            &=& %
            \numvar{\idxi} = 1 %
            \land %
            \lnumvar{\idxi} = 1 %
            \\
            \csssimpres{\psonlyoftype}{\idxi} %
            &=& %
            \bigvee\limits_{\substack{ %
                \ns \in \finof{\nspaces} %
                \\ %
                \ele \in \finof{\eles} %
            }}\brac{ %
                \begin{array}{c} %
                    \nsvar{\idxi} = \ns \land %
                    \elevar{\idxi} = \ele\ \land %
                    \\ %
                    \numtypevar{\idxi}{\qele{\ns}{\ele}} = 0 %
                    \land %
                    \lnumtypevar{\idxi}{\qele{\ns}{\ele}} = 0 %
                \end{array} %
            } %
        \end{array} %
    \]
\end{definition}

We are now ready to move on to complete the encoding.

\paragraph{Encoding Runs of CSS Automata}

Finally, now that we are able to encode attribute and node selectors, we can make use of these to encode accepting runs of a CSS automaton.
Since we know that, if there is an accepting run, then there is a run of length at most $\numof$ where $\numof$ is the number of transitions in $\atrans$, we encode the possibility of an accepting run using the variables discussed above for all $0 \leq \idxi \leq \numof$.
The shape of the translation is given below and elaborated on afterwards.

\begin{definition}{$\presof{\cssaut}$}
    Given a CSS automaton $\cssaut$ we define
    \[
        \presof{\cssaut} = \brac{ %
            \brac{ %
                \begin{array}{c} %
                    \astatevar{0} = \ainitstate \\ %
                    \land \\ %
                    \astatevar{\numof} = \afinstate %
                \end{array} %
            } %
            \land %
            \bigwedge\limits_{0 \leq \idxi < \numof} %
            \brac{ %
                \begin{array}{c} %
                    \tranpres{\idxi} \\ %
                    \lor \\ %
                    \astatevar{\idxi} = \afinstate %
                \end{array} %
            } %
            \land %
            \consistent %
        } %
    \]
    where $\tranpres{\idxi}$ and $\consistent$ are defined below.
\end{definition}

Intuitively, the first two conjuncts asserts that a final state is reached from an initial state.
Next, we use $\tranpres{\idxi}$ to encode a single step of the transition relation, or allows the run to finish early.
Finally $\consistent$ asserts consistency constraints.

We define as a disjunction over all possible (single-step) transitions
$
    \tranpres{\idxi} = \bigvee\limits_{\atrant \in \atrans} \tranprest{\idxi}{\atrant}
$
where $\tranprest{\idxi}{\atrant}$ is defined below by cases.
There are four cases depending on whether the transition is labelled
$\arrchild$, $\arrneighbour$, $\arrsibling$, or $\arrlast$.
In most cases, we simply assert that the state changes as required by the transition, and that the variables
$\numvar{\idxi}$
and
$\numtypevar{\idxi}{\qele{\ns}{\ele}}$
are updated consistently with the number of nodes read by the transition.
Although the encodings look complex, they are essentially simple bookkeeping.

To ease presentation, we write
$\qele{\nsvar{\idxi}}{\elevar{\idxi}} = \qele{\ns}{\ele}$
as shorthand for
$\brac{
    \nsvar{\idxi} = \ns \land \elevar{\idxi} = \ele
}$
and
$\qele{\nsvar{\idxi}}{\elevar{\idxi}} \neq \qele{\ns}{\ele}$
as shorthand for
$\brac{
    \nsvar{\idxi} \neq \ns \lor \elevar{\idxi} \neq \ele
}$.

\begin{compactenum}
\item
    When
    $\atrant = \astate \atran{\arrchild}{\csssim} \astate'$
    we define $\tranprest{\idxi}{\atrant}$ to be
    \[
        \begin{array}{c} %
            \brac{\astatevar{\idxi} = \astate} \land %
            \brac{\astatevar{\idxi+1} = \astate'} \land %
            \neg{\pclsvar{\idxi}{\psempty}} \land %
            \csssimpres{\csssim}{\idxi} \ \land %
            \\ %
            \brac{\numvar{\idxi+1} = 1} \land %
            \bigwedge\limits_{\substack{ %
                \ns \in \finof{\nspaces} %
                \\ %
                \ele \in \finof{\eles} %
            }}\brac{ %
                \numtypevar{\idxi+1}{\qele{\ns}{\ele}} = 0 %
            } %
            \ . %
        \end{array} %
    \]

\item
    When
    $\atrant = \astate \atran{\arrneighbour}{\csssim} \astate'$
    we define $\tranprest{\idxi}{\atrant}$ to be false when $\idxi = 0$ (since the root has no siblings) and otherwise
    \[
        \begin{array}{c} %
            \brac{\astatevar{\idxi} = \astate} \land %
            \brac{\astatevar{\idxi+1} = \astate'} \land %
            \csssimpres{\csssim}{\idxi} \ \land %
            \\ %
            \brac{\numvar{\idxi+1} = \numvar{\idxi} + 1} \land %
            \brac{\lnumvar{\idxi+1} = \lnumvar{\idxi} - 1} \ \land %
            \\ %
            \bigwedge\limits_{\substack{ %
                \ns \in \finof{\nspaces} %
                \\ %
                \ele \in \finof{\eles} %
            }} \brac{ %
                \begin{array}{c} %
                    \brac{ %
                        \brac{ %
                            \qele{\nsvar{\idxi}}{\elevar{\idxi}} %
                            = %
                            \qele{\ns}{\ele} %
                        } %
                        \Rightarrow %
                        \brac{ %
                            \numtypevar{\idxi+1}{\qele{\ns}{\ele}} = %
                            \numtypevar{\idxi}{\qele{\ns}{\ele}} + 1 %
                        } %
                    } \ \land %
                    \\ %
                    \brac{ %
                        \brac{ %
                            \qele{\nsvar{\idxi}}{\elevar{\idxi}} %
                            \neq %
                            \qele{\ns}{\ele} %
                        } %
                        \Rightarrow %
                        \brac{ %
                            \numtypevar{\idxi+1}{\qele{\ns}{\ele}} = %
                            \numtypevar{\idxi}{\qele{\ns}{\ele}} %
                        } %
                    } \ \land %
                    \\ %
                    \brac{ %
                        \brac{ %
                            \qele{\nsvar{\idxi+1}}{\elevar{\idxi+1}} %
                            = %
                            \qele{\ns}{\ele} %
                        } %
                        \Rightarrow %
                        \brac{ %
                            \lnumtypevar{\idxi+1}{\qele{\ns}{\ele}} = %
                            \lnumtypevar{\idxi}{\qele{\ns}{\ele}} - 1 %
                        } %
                    }\ \land %
                    \\ %
                    \brac{ %
                        \brac{ %
                            \qele{\nsvar{\idxi+1}}{\elevar{\idxi+1}} %
                            \neq %
                            \qele{\ns}{\ele} %
                        } %
                        \Rightarrow %
                        \brac{ %
                            \lnumtypevar{\idxi+1}{\qele{\ns}{\ele}} = %
                            \lnumtypevar{\idxi}{\qele{\ns}{\ele}} %
                        } %
                    } %
                \end{array} %
            } %
            \ . %
        \end{array}
    \]

\item
    When
    $\atrant = \astate \atran{\arrsibling}{\isany} \astate$
    we define $\tranprest{\idxi}{\atrant}$ to be false when $\idxi = 0$ and otherwise
    \[
        \begin{array}{c} %
            \brac{\astatevar{\idxi} = \astate} \land %
            \brac{\astatevar{\idxi+1} = \astate} \ \land %
            \\ %
            \exists \shiftvar . \brac{ %
                \brac{\numvar{\idxi+1} = \numvar{\idxi} + \shiftvar} \land %
                \brac{\lnumvar{\idxi+1} = \lnumvar{\idxi} - \shiftvar} %
            } \ \land %
            \\ %
            \bigwedge\limits_{\substack{ %
                \ns \in \finof{\nspaces} %
                \\ %
                \ele \in \finof{\eles} %
            }} %
            \exists \shiftvartype{\qele{\ns}{\ele}} .  \brac{ %
                \begin{array}{c} %
                    \brac{ %
                        \begin{array}{l} %
                            \brac{ %
                                \qele{\nsvar{\idxi}}{\elevar{\idxi}} %
                                = %
                                \qele{\ns}{\ele} %
                            } %
                            \ \Rightarrow \\ %
                            \quad %
                            \brac{ %
                                \numtypevar{\idxi+1}{\qele{\ns}{\ele}} = %
                                \numtypevar{\idxi}{\qele{\ns}{\ele}} + %
                                \shiftvartype{\qele{\ns}{\ele}} + %
                                1 %
                            } %
                        \end{array} %
                    } \ \land %
                    \\ %
                    \brac{ %
                        \begin{array}{l} %
                            \brac{ %
                                \qele{\nsvar{\idxi}}{\elevar{\idxi}} %
                                \neq %
                                \qele{\ns}{\ele} %
                            } %
                            \ \Rightarrow \\ %
                            \quad %
                            \brac{ %
                                \numtypevar{\idxi+1}{\qele{\ns}{\ele}} = %
                                \numtypevar{\idxi}{\qele{\ns}{\ele}} + %
                                \shiftvartype{\qele{\ns}{\ele}} %
                            } %
                        \end{array} %
                    } \ \land %
                    \\ %
                    \brac{ %
                        \begin{array}{l} %
                            \brac{ %
                                \qele{\nsvar{\idxi+1}}{\elevar{\idxi+1}} %
                                = %
                                \qele{\ns}{\ele} %
                            } %
                            \ \Rightarrow \\ %
                            \quad %
                            \brac{ %
                                \lnumtypevar{\idxi+1}{\qele{\ns}{\ele}} = %
                                \lnumtypevar{\idxi}{\qele{\ns}{\ele}} - %
                                \shiftvartype{\qele{\ns}{\ele}} - %
                                1 %
                            } %
                        \end{array} %
                    }\ \land %
                    \\ %
                    \brac{ %
                        \begin{array}{l} %
                            \brac{ %
                                \qele{\nsvar{\idxi+1}}{\elevar{\idxi+1}} %
                                \neq %
                                \qele{\ns}{\ele} %
                            } %
                            \ \Rightarrow \\ %
                            \quad %
                            \brac{ %
                                \lnumtypevar{\idxi+1}{\qele{\ns}{\ele}} = %
                                \lnumtypevar{\idxi}{\qele{\ns}{\ele}} - %
                                \shiftvartype{\qele{\ns}{\ele}} %
                            } %
                        \end{array} %
                    } %
                \end{array} %
            } %
            \ . %
        \end{array} %
    \]

\item
    When
    $\atrant = \astate \atran{\arrlast}{\csssim} \astate'$
    we define $\tranprest{\idxi}{\atrant}$ to be
    \[
        \brac{\astatevar{\idxi} = \astate} \land
        \brac{\astatevar{\idxi+1} = \astate'}  \land
        \csssimpres{\csssim}{\idxi} \ .
    \]
\end{compactenum}

To ensure that the run is over a consistent tree, we assert the consistency constraint
\[
    \consistent = \consistentnums
                  \land
                  \consistentids
                  \land
                  \consistentpseudo
\]
where each conjunct is defined below.
\begin{compactitem}
\item
    The clause $\consistentnums$ asserts that the values of
    $\numvar{\idxi}$, $\lnumvar{\idxi}$, $\numtypevar{\idxi}{\qele{\ns}{\ele}}$, and $\lnumtypevar{\idxi}{\qele{\ns}{\ele}}$
    are consistent.
    That is
    \[
        \bigwedge\limits_{1 \leq \idxi \leq \numof} %
        \brac{ %
            \numvar{\idxi} = %
            1 + %
            \sum\limits_{\qele{\ns}{\ele} \in \eles} \numtypevar{\idxi}{\qele{\ns}{\ele}} %
        } %
        \land %
        \brac{ %
            \lnumvar{\idxi} = %
            1 + %
            \sum\limits_{\qele{\ns}{\ele} \in \eles} \lnumtypevar{\idxi}{\qele{\ns}{\ele}} %
        } \ . %
    \]

\item
    The clause $\consistentids$ asserts that ID values are unique.
    It is the conjunction of the following clauses.
    For each $\ns$ for which we have created variables of the form
    $\wordpos{\ns}{\idatt}{\idxi}{\idxj}$
    we assert
    \[
        \bigwedge\limits_{1 \leq \idxi \neq \idxi' \leq \numof}
            \bigvee\limits_{1 \leq \idxj \leq \attvalbound}
                \wordpos{\ns}{\idatt}{\idxi}{\idxj}
                \neq
                \wordpos{\ns}{\idatt}{\idxi'}{\idxj} \ .
    \]

\item
    Finally, $\consistentpseudo$ asserts the remaining consistency constraints on the pseudo-classes.
    We define $\consistentpseudo =$
    \[
        \bigwedge\limits_{0 \leq \idxi \leq \numof} %
        \brac{ %
            \begin{array}{c} %
                \neg\brac{\pclsvar{\idxi}{\pslink} \land \pclsvar{\idxi}{\psvisited}} %
                \ \land \\ %
                \bigwedge\limits_{0 \leq \idxj \neq \idxi \leq \numof} %
                \brac{ %
                    \begin{array}{c} %
                        \neg\brac{ %
                            \pclsvar{\idxi}{\pstarget} \land %
                            \pclsvar{\idxj}{\pstarget} %
                        } %
                    \end{array} %
                } %
                \ \land \\ %
                \neg\brac{\pclsvar{\idxi}{\psenabled} \land \pclsvar{\idxi}{\psdisabled}} %
            \end{array} %
        } %
        \ . %
    \]
    These conditions assert the mutual exclusivity of $\pslink$ and $\psvisited$, that at most one node in the document can be the target node, that nodes are not both enabled and disabled.
\end{compactitem}

\subsubsection{Correctness of the Encoding}

We have now completed the definition of the reduction from the emptiness problem of CSS automata to the satisfiability of existential Presburger arithmetic.
What remains is to show that this reduction is faithful:
that is, the CSS automaton has an empty language if and only if the formula is satisfiable.
The proof is quite routine, and presented in Lemma~\ref{lem:aut-enc-complete} and Lemma~\ref{lem:aut-enc-sound} in Appendix~\ref{sec:aut-emptiness-proof}.

\begin{namedlemma}{lem:aut-enc-correct}{Correctness of $\presof{\cssaut}$}
    For a CSS automaton $\cssaut$, we have
    \[
        \ap{\Lang}{\cssaut} \neq \emptyset
        \iff
        \presof{\cssaut}
        \text{ is satisfiable.}
    \]
\end{namedlemma}

We are thus able to decide the emptiness problem, and therefore the emptiness of intersection problem, for CSS automata by reducing the problem to satisfiability of existential Presburger arithmetic and using a fast solver such as Z3~\cite{Z3} to resolve the satisfiability.

%% file: graph2maxsat.tex
\section{Rule-Merging to Max-SAT}
\label{sec:graph2maxsat}

In this section,
we provide a reduction from the rule-merging problem to partial weighted
MaxSAT. The input will be a valid covering $\covering = \incovering{1}{m}$ of
a CSS graph $\CSSgraph$.
We aim to find a rule $\bucket = \inbucket$ and a position $j$
that minimises the weight of $\Trim{\inSeq{\covering}{j}{\bucket}}$.

There is a fairly straightforward encoding of the rule-merging problem into Max-SAT using for each node
$w \in \selNodes \cup \propNodes$
a boolean variable $\vble{w}$ which is true iff the node is included in the new rule.
Unfortunately, early experiments showed that such an encoding does not perform well in practice, causing prohibitively long run times even on small examples.
The failure of the naive encoding may be due to the search space that includes
a large number of possible pairs $(\selsBucket,\propsBucket)$ that turn out
to be invalid rules (e.g. include edges not in the CSS-graph
$\CSSgraph$).
Hence, we will use a different Max-SAT encoding that, by means of syntax,
further restricts the search space of valid rule-merging opportunities.

The crux of our new encoding
is to explicitly say in the Max-SAT formula $\varphi$ that the rule $\bucket$
in a merging opportunity $(\bucket, j)$ is a ``valid'' sub-biclique of one
of the \emph{maximal} bicliques $\biclique = \inbiclique$
 --- maximal with respect to subset-of relations of
$\selsBucket$ and $\propsSet$ components of bicliques --- in the CSS-graph
$\CSSgraph$.
By insisting $\bucket$ is contained within a maximal biclique of the CSS-graph, we automatically ensure that $\bucket$ does not contain edges that are not in $\CSSgraph$.

The formula $\varphi$ will try to guess a maximal
biclique $\biclique$ and which nodes to omit from $\biclique$.
Since the number of maximal bicliques in a bipartite graph
is exponential (in the number of nodes) in the worst case, one concern with this idea is that the
constraint $\varphi$
might become prohibitively large. As we shall see, this turns out not to be
the case in practice. Intuitively, based on our experience, the number of
rules in a real-world CSS file is \emph{at most} a few thousand. Second, the number of maximal bicliques in a CSS-graph
that corresponds to a real-world CSS file also is typically of a similar size to the number of rules, and,
furthermore, can be enumerated using the algorithm from \cite{K10} (which runs 
in time polynomial in the size of the input and the size of the output).
To be more precise, the benchmarks in our experiments had between $31$ and $2907$ rules, and the mean number of rules was $730$.
The mean ratio of the number of maximal bicliques to the number of rules was 
1.25, and the maximum was 2.05.
As we shall see in Section~\ref{sec:experiments}, Z3 may solve the constraints via this encoding quite efficiently.

In the rest of the section, we will describe our encoding in detail. For
convenience, our encoding also allows bounded integer variables.  There are
standard ways to encode these in binary as booleans~(e.g.\ see~\cite{P15}) by bit-blasting
(using a logarithmic number of boolean variables).

\subsection{Orderable Bicliques}
Our description above of the crux of our encoding (by restricting to containment
in a maximal biclique) is a simplification. This is because \emph{not all}
sub-bicliques of a maximal biclique correspond to a valid rule $\bucket$ in
a merging opportunity $(\bucket,j)$ with respect to the covering
$\covering$. (A biclique $(\selsSet',\propsSet')$ is a
\defn{sub-biclique} of a biclique $\inbiclique$ if $\selsSet' \subseteq
\selsSet$ and $\propsSet' \subseteq \propsSet$.)
    To ensure that our constraint $\varphi$ chooses only valid rules,
it needs to ensure that the sub-biclique that is chosen is ``orderable''.
More precisely, a biclique $\biclique = \inbiclique$ is \defn{orderable at
position $j$} if it can be turned into a rule $\bucket = \inbucket$ (i.e.
turning the
set $\propsSet$ of declarations into a sequence $\propsBucket$ by assigning
an order) that can be inserted
at position $j$ in $\covering$ without violating the validity of the resulting
covering with respect to the order $\edgeOrder$ (from the CSS-graph
$\CSSgraph$). If there are $m$ rules in $\covering$, there are $m+1$ positions
(call these positions $0,\ldots,m$) where $\propsBucket$ may be inserted into
$\covering$.
We show below that the position $j$ is crucial to whether $\bucket$ is orderable.

Unorderable bicliques rarely arise in practice (in our benchmarks, the mean
percentage of maximal bicliques that were unorderable at some position was
$\meanbadbicliques\%$), but they have to be accounted for if our analysis is to find the
optimal rule-merging opportunity.
A biclique $\biclique = \inbiclique$ is unorderable when they have the same
property name (or two related
property names, e.g., shorthands) occuring multiple times with different values
in $\propsSet$. One reason having a CSS rule with the same property name
occuring multiple times with different values is to provide
``fallback options'' especially because old browsers may not support certain
values in some property names, e.g., the rule
\begin{center}
\begin{minted}{css}
    .c { color:#ccc; color:rgba(0, 0, 0, 0.5); }
\end{minted}
\end{center}
says that if \texttt{rgba()} is not supported (e.g. in IE8 or older browsers),
then use \texttt{\#ccc}. Using this idea, we can construct the simple
 example of an unorderable biclique in the CSS file in Figure
\ref{fig:CSSorder}.
The ordering constraints we can derive from this file include
\[
    (\texttt{.a},\texttt{color:blue}) \edgeOrder
    (\texttt{.a},\texttt{color:green})
\]
and
\[
    (\texttt{.b},\texttt{color:green}) \edgeOrder
    (\texttt{.b},\texttt{color:blue}) \ .
\]

\begin{figure}
    \begin{mdframed}
    \begin{minipage}{1.0\textwidth}
    \begin{minted}{css}
    .a { color:blue; color:green }
    .b { color:green; color:blue }
    \end{minted}
    \end{minipage}
    \end{mdframed}
    \caption{A CSS file with an unorderable sub-biclique.\label{fig:CSSorder}}
\end{figure}
The biclique $\biclique$
\[
    (\{\texttt{.a},\texttt{.b}\},\{\texttt{color:blue},\texttt{color:green}\})
\]
is easily seen to be orderable at position 0 and 1.
This is because the final rule in Figure~\ref{fig:CSSorder} will ensure the ordering
$(\texttt{.b},\texttt{color:green}) \edgeOrder (\texttt{.b},\texttt{color:blue})$
is satisfied, and since $\biclique$ will appear before this final rule, only
$(\texttt{.a},\texttt{color:blue}) \edgeOrder (\texttt{.a},\texttt{color:green})$
needs to be maintained by $\biclique$ (in fact, at position $0$, neither of the orderings need to be respected by $\biclique$).
However, at position 2, which is at the end of the file in Figure~\ref{fig:CSSorder}, both orderings will have to be respected by $\biclique$.
Unfortunately, one of these orderings will be violated regardless of how one may assign
an ordering to \texttt{blue} and \texttt{green}.
\OMIT{
Notice that $(\texttt{.a},\texttt{color:green}) \edgeOrder
(\texttt{.b},\texttt{color:green})$ because there is potentially an HTML element
in the DOM that has both classes \texttt{.a} and \texttt{b}.
}
This contrived example was made only for illustration, however, our technique should still be able to handle even contrived examples.

We mention that both checking orderability and ordering a given biclique can be done efficiently.
\begin{proposition}
\label{prop:poly-order}
    Given a biclique $\biclique$, a covering $\covering$ (with $m$ rules) of a
    CSS-graph $\CSSgraph$, and a number $j \in \{0,\ldots,m\}$, checking
    whether $\biclique$ is orderable at position $j$ in $\covering$ can be
    done in polynomial time.
    Moreover, if $\biclique$ is orderable an ordering can be calculated in polynomial time.
\end{proposition}
The proof of the proposition is easy and is relegated into Appendix~\ref{sec:orderable-bicliques-appendix}.

\subsubsection*{Maximal Orderable Bicliques}
Our Max-SAT encoding $\varphi$ needs to ensure that we only pick a pair
$(\biclique,j)$ such that $\biclique$ is an orderable biclique at position $j$
in the given
covering $\covering$, i.e., $\biclique$ corresponds to a rule that can be
inserted at position $j$ in $\covering$.
Although the check of orderability can be \emph{declaratively} expressed as a
constraint in $\varphi$,
we found that this results in Max-SAT formulas that are rather difficult to
solve by existing Max-SAT solvers.
For this reason, we propose to express the check of orderability in a different
way. Intuitively, for each $j \in \{0,\ldots,m\}$, we enumerate all orderable
bicliques $\biclique = \inbiclique$ that are also maximal, i.e., it is
\emph{not} a (strict) sub-biclique of a different orderable biclique.
%
Since ``orderability is inherited by sub-bicliques'' (as the following
lemma, whose proof is immediate from the definition, states), the constraint
$\varphi$ needs to simply choose a sub-biclique
of a maximal orderable biclique that appears in our enumeration.
\begin{lemma}
    Every sub-biclique $\biclique' = (\selsSet,\propsSet)$ of an orderable
    biclique $\biclique = \inbiclique$ is orderable.
    \label{lm:downward}
\end{lemma}

\OMIT{
An orderable biclique
$(\selsBucket, \propsSet)$
is \defn{maximal} if there does not exist an orderable biclique
$(\selsBucket', \propsSet')$
with
$\selsBucket \times \propsSet \subset \selsBucket' \times \propsSet'$.
Orderable bicliques may become unorderable as $j$ increases, but not vice-versa.
We will often write $\mbiclique$ to denote a maximal biclique.
}

The above enumeration of maximal orderable bicliques can be described as a
pair
$(\{\mbiclique_i\}_{i=1}^{\nummbicliques}, \forbidden)$
where
\begin{itemize}
\item
    $\{\mbiclique_i\}_{i=1}^{\nummbicliques}$
    is an enumeration of all bicliques that are orderable and maximal at some position $j$, and
\item
    $\forbidden$ \emph{forbids} certain bicliques at each position.
    I.e. it is a function from
    $[1, m]$
    to the set of bicliques in
    $\{\mbiclique_i\}_{i=1}^{\nummbicliques}$
    that are unorderable at position $j$.
\end{itemize}
Observe that the set of orderable bicliques at position $j$ in
$\covering$ is a subset of the set of orderable bicliques at position $j+1$
in $\covering$. This may be formally expressed as: $\forbidden(j) \subseteq
\forbidden(j+1)$ for all $j \in [1,m)$.

In the majority ($54\%$) of examples that we have
from real-world CSS, the function $\forbidden$ maps all values of $[1,m]$ to
$\emptyset$, i.e., all maximal bicliques are orderable at all positions.
The mean percentage of maximal bicliques that were unorderable at some position was $\meanbadbicliques\%$, with a maximum of $5.84\%$.

In our description of the Max-SAT encoding below, we assume that the pair
$(\{\mbiclique_i\}_{i=1}^{\nummbicliques}, \forbidden)$ has been computed
for the input $\covering$.

\subsection{The Max-SAT Encoding}

We present the full reduction of the rule-merging problem to Max-SAT.  In
particular, the constraints we produce are
\[
    (\hardcons, \softcons)
\]
where $\hardcons$ and $\softcons$ are, respectively, hard and soft constraints.
First, we describe the variables used in our encoding.
Note, our encoding will rely on the assumption that covering $\covering$ has already been trimmed.
Recall, the notion of trimming is defined in Section~\ref{sec:css2graph} and is the process of removing redundant nodes from a covering.
A node is redundant in a rule if all of its incident edges also appear later in the covering.

\subsubsection{Representing the rule-merging opportunity}
We need to represent a merging opportunity
$(\bucket, j)$.
We use a bounded integer variable $\inpos$ (with range $0 \leq \inpos \leq m$)
to encode $j$.

For $\bucket$ we select a biclique in
$\{\mbiclique_i\}_{i=1}^{\nummbicliques}$
and allow some nodes to be removed (i.e.\ to produce sub-bicliques of the $\mbiclique_i$).
We use a bounded integer variable
$\bicliquevble$ (with range $[1,\nummbicliques]$)
to select $\mbiclique_i$.
Next, we need to choose a sub-biclique of $\mbiclique_i$, which can be
achieved by choosing nodes $\mbiclique_i$ to be removed.
To minimise the number of variables used, we number the
nodes contained in each biclique in some (arbitrary) way, i.e., for each
$i \in [1,\nummbicliques]$ and biclique $\mbiclique_i = (\selsBucket,
\propsSet)$,
we define a bijection
$\selpropord{i} : \selsBucket \cup \propsSet \to [1,|\selsBucket \cup
\propsSet|]$.
Let $\numexcludes$ be the maximum number of nodes in a biclique $\mbiclique_i$
in the enumeration $\{\mbiclique_i\}_{i=1}^{\nummbicliques}$, i.e.,
the maximal integer $k$ such that $\ap{\selpropord{i}}{w} = k$ for some
$w \in \selNodes \cup \propNodes$ and $1 \leq i \leq \nummbicliques$.

We introduce boolean variables
$\excludevble{1}, \ldots, \excludevble{\numexcludes}$.
Once the maximal orderable biclique $\mbiclique_i$ is picked, for a node $w$ with $\ap{\selpropord{i}}{w} = k$,
the variable $\excludevble{k}$ is used to indicate
that $w$ is to be excluded from selected $\mbiclique_i$ (i.e. $\excludevble{k}$ is
true iff $w$ is to be excluded from $\mbiclique_i$).
More precisely, for an edge $e = (s,p)$ in $\mbiclique_i$,
we define a predicate $\hasedge{e}$ to be
\[
    \hasedge{(s, p)} =
    \bigvee\limits_{1 \leq i \leq \nummbicliques}
        \bicliquevble = i
        \land
        \neg \excludevble{\ap{\selpropord{i}}{s}}
        \land
        \neg \excludevble{\ap{\selpropord{i}}{p}} \ .
\]
Note, when
$\mbiclique_i = (\selsBucket, \propsSet)$
and
$w \notin \selsBucket \cup \propsSet$
we let
$\excludevble{\ap{\selpropord{i}}{w}}$
denote the formula ``true''.

\subsubsection{Hard Constraints}

We define
\[
    \hardcons = \set{\fmlavalid, \fmlaordering}
\]
where
$\fmlavalid$ and $\fmlaordering$ are described below.

We need to ensure the rule-merging opportunity is valid, i.e., it has not been forbidden at the chosen position and inserting it into the covering does not violate the edge order $\edgeOrder$.
For the former, we define $\forbiddenfst$, which is used to discover which of the $\mbiclique_i$ first become unorderable at position $j$.
That is $\mbiclique_i \in \ap{\forbiddenfst}{j}$ if $j$ is the smallest integer such that $\mbiclique_i \in \ap{\forbidden}{j}$.
\[
    \fmlavalid =
    \bigwedge\limits_{1 \leq j \leq m}
    \brac{
        \brac{\inpos >= j}
        \Rightarrow
        \bigwedge\limits_{\mbiclique_i \in \ap{\forbiddenfst}{j}}
            \brac{\bicliquevble \neq  i}
    } \ .
\]

We also need to ensure the edge ordering is respected by the rule-merging
opportunity, for which we define $\fmlaordering$.
If
$e_1 \edgeOrder e_2$,
and $e_1 = (s_1,p_1)$ is in the rule $\bucket$ in the guessed merging
opportunity $(\bucket,j)$, then we need to assert
$e_1 \edgeOrder e_2$ is respected.
It is respected either if $e_2 = (s_2,p_2)$ appears after the position where $j$ is to be inserted in $\covering$, or $e_2$ is also contained in the new rule $\bucket$ (in which case the ordering can still be respected since $\bucket$ is orderable).
That is,
\[
    \fmlaordering =
    \brac{
        \bigwedge\limits_{(s_1, p_1) \edgeOrder (s_2, p_2)}
        \hasedge{(s_1, p_1)}\ \Rightarrow
        \brac{
            \inpos < \Index((s_2, p_2))
            \lor
            \hasedge{(s_2, p_2)}
        }
    } \ .
\]
This is because only the last occurrence of an edge in a covering matters
(recall the definition of index of an edge in Section \ref{sec:css2graph}).
Also note that
our use of bicliques ensures that we do not introduce pairs $(s,p)$ that
are not edges in $\CSSedges$.

\subsubsection{Soft Constraints}

The soft constraints will calculate the weight of
$\Trim{\inSeq{\covering}{j}{\bucket}}$.
Since we want to minimise this weight, the optimal solution to the Max-SAT
problem will give an optimal rule-merging opportunity.
Our soft constraints are
\[
    \softcons = \softconsbucket \cup \softconssels \cup \softconsprops
\]
where $\softconsbucket$ counts the weight of the $\bucket$, and $\softconssels$ and $\softconsprops$ counts the weight of the remainder (not including $\bucket$) of the stylesheet by counting the remaining selectors and properties respectively after trimming.
These are defined below.

To count the weight of $\bucket$, we count the weight of the non-excluded nodes
of $\mbiclique_i = (\selsBucket_i, \propsSet_i)$.
We have
\[
    \softconsbucket =
    \setcomp{
        (\softfmlabucket{i}{w}, \nodeWeight(w))
    }{
        1 \leq i \leq \nummbicliques
        \land
        w \in \selsBucket_i \cup \propsSet_i
    }
\]
where
\[
    \softfmlabucket{i}{w} =
    \brac{\brac{\bicliquevble = i} \Rightarrow \excludevble{\ap{\selpropord{i}}{w}}} \ .
\]
Note that $\softfmlabucket{i}{w}$ is true iff, whenever $\mbiclique_i$ is
picked, the node $w$ is omitted from $\mbiclique_i$ in the guessed $\bucket$.
Furthermore, the cost of \emph{not} omitting $w$ from $\mbiclique_i$ is
$\nodeWeight(w)$.

Next, we count the remaining weight of
$\inSeq{\covering}{j}{\bucket}$ (i.e. excluding the weight of $\bucket$).
Assume that $\covering$ has already been trimmed, i.e., $\Trim{\covering}
= \covering$.
The intuition is that a node $v$ can be removed from $\bucket_i$ in $\covering$ if
all edges $e$ incident to $v$ appear in a later rule in $\covering$, i.e.,
$\Index(e) > i$.
In particular, let
$\bucket_i = (\selsBucket_i, \propsBucket_i)$
in
$\incovering{1}{m}$.
To count the weight of the untrimmed selectors we use the clauses
\[
    \softconssels =
    \setcomp{
        (\fmlanode{i}{s}, \nodeWeight(s))
    }{
        1 \leq i \leq m
        \land
        s \in \selsBucket_i
    }
\]
where
\[
    \fmlanode{i}{s} =
    \brac{
        i \leq \inpos \land
        \bigwedge\limits_{\substack{
            \Index((s, p)) = i \\
            p \in \propsBucket
        }}
            \hasedge{(s, p)}
    } \ .
\]
The idea is that $\fmlanode{i}{s}$ will be satisfied whenever $s$ can be removed from $\bucket_i$.
We assume $\covering$ has already been trimmed, so $s$ can only become removable
because of the application of rule-merging.
This explains the first conjunct which asserts that nodes can only be removed from rules appearing before $\inpos$.
Next, $s$ can be removed after rule-merging if none of its incident edges $(s, p)$ are the final occurrence of $(s, p)$ in
$\inSeq{\covering}{j}{\bucket}$.
The crucial edges in this check are those such that $\Index((s, p)) = i$, which means $\bucket_i$ contains the final occurrence of $(s, p)$ in $\covering$.
For these edges, if
$\hasedge{(s, p)}$
holds, then the new rule contains $(s, p)$ and $\bucket_i$ will no longer
contain the final occurrence of $(s, p)$ after rule-merging.

Similarly, to count the weight of the properties that cannot be removed we have
\[
    \softconsprops =
    \setcomp{
        (\fmlanode{i}{p}, \nodeWeight(p))
    }{
        1 \leq i \leq m
        \land
        p \in \propsBucket_i
    }
\]
where
\[
    \fmlanode{i}{p} =
    \brac{
        i \leq \inpos \land
        \bigwedge\limits_{\substack{
            \Index((s, p)) = i \\
            s \in \selsBucket
        }}
            \hasedge{(s, p)}
    } \ .
\]

\subsection{Generated Rule-Merging Opportunity}
The merging opportunity
$(\bucket, j)$
is built from an optimal satisfying assignment to
$(\hardcons, \softcons)$.
First, $j$ is the value assigned to $\inpos$.
Then
$\bucket = (\selsBucket, \propsBucket)$
where, letting $i_\mbiclique$ be the value of $\bicliquevble$, and letting
$\mbiclique_{i_\mbiclique} = (\selsBucket', \propsSet')$,
\begin{itemize}
\item
    $s \in \selsBucket$ iff
        $s \in \selsBucket'$ and
        $\excludevble{\ap{\selpropord{i_\mbiclique}}{s}}$
        is assigned the value false, and

\item
    $\propsBucket = \{p_i\}_{i=1}^m$, where $\{p_i\}_{i=1}^m$ is obtained by
        assigning an ordering to $\propsSet'$ such that $(\bucket,j)$
        is a valid merging opportunity.
\end{itemize}
That the ordering of $\propsSet'$ above exists is guaranteed by the fact that
$\mbiclique_{i_\mbiclique}$
is orderable at $j$, and Lemma \ref{lm:downward}.
We can compute the ordering in polynomial time via Proposition~\ref{prop:poly-order}.
We argue the following proposition in Appendix~\ref{sec:encoding-correct}.

\begin{proposition}
    \label{prop:encoding-correct}
    The merging opportunity $(\bucket, j)$ generated from the maximal solution to
    $(\hardcons, \softcons)$
    is the optimal merging opportunity of $\covering$.
\end{proposition}

%% file: experiments.tex
\section{Experimental Results}
\label{sec:experiments}

We implemented a tool \satcss (in Python 2.7) for CSS minification via
rule-merging which can be found in our supplementary material~\cite{HHL18}.  The source code is also available at the following URL.
\begin{center}
    \url{https://github.com/matthewhague/sat-css-tool}
\end{center}
It constructs the edge order following Section~\ref{sec:edgeOrder} and discovers
a merging opportunity following Section~\ref{sec:graph2maxsat}.
We use Z3~\cite{Z3} as the back end SMT and Max-SAT solver.
As an additional contribution, our tool can also generate instances in the extended DIMACS format \cite{maxsat_benchmarks} allowing us to use any Max-SAT solvers.
This output may also provide a source of industrially inspired examples for Max-SAT competition benchmarks.

\begin{wrapfigure}{r}{0.2\textwidth}
  \begin{center}
    \includegraphics[width=0.2\textwidth]{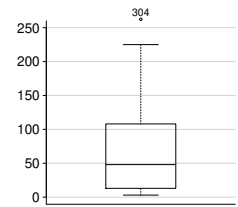}
  \end{center}
  \caption{Box plot of the file sizes in kilobytes}\label{fig:filesizes}
\end{wrapfigure}

Our benchmarks comprise \numbenchmarks\ CSS files coming from three different sources
(see Appendix~\ref{sec:benchmark-details} for a complete listing):
\begin{itemize}
\item
    We collected CSS files from each of the top 20 websites on a global ranking list~\cite{website-ranking}.
\item
    CSS files were taken from a further 12 websites selected from the top 20-65 websites from the listing above.
\item
    Finally, CSS files were taken from 11 smaller websites such as DBLP.
    These were examples used during the development of the tool.
\end{itemize}
This selection gives us a wide range of different types of websites, including large scale industrial websites and those developed by smaller teams.
Note, several websites contained more than one CSS file and in this case we took a selection of CSS files from the site.
Hence, we collected more examples than the number of websites used.
Figure~\ref{fig:filesizes} gives the file-size distribution of the CSS files we collected.

\OMIT{
As well as running our tool on the examples as downloaded from the websites above, we also ran our tool in conjunction with six popular CSS minifiers~\cite{minifier-list}.
This allows us to compare the savings of our refactoring technique to savings obtained by traditional minification.
Moreover, these experiments show that our techniques can be used to great effect \emph{as a part of} a minification suite.
That is, we identify savings that are different to those identified by popular minifiers.
}

In the following, we describe the optimisations implemented during the development of \satcss.
We then describe the particulars of the experimental setup and provide the results in Figure~\ref{fig:box-plots}.

\subsection{Optimisations}

We give an overview of the optimisations we used when implementing \satcss.
In section \ref{sec:eval_op} we provide a detailed evaluation of the proposed
optimisations.
\begin{itemize}
\item
    We introduce variables $\excludevble{k}$ only for nodes appearing in the edge order, rather than for all nodes in a biclique.
    This is enough to be able to define sub-bicliques of any $\mbiclique_i$ that
        can appear at any position in the covering, but means that the smallest
        rule-merging opportunity cannot always be constructed.
    However, this reduces the search space and leads to an improvement to run times.
    For example, given a biclique
    \begin{verbatim}
        .a { color: red; background: blue }
    \end{verbatim}
    \vspace{-2ex}
    where the property \texttt{color: red} appears in some pair of the edge ordering, but the property \texttt{background: blue} does not, we introduce a variable $\excludevble{k}$ which is true whenever the declaration $\texttt{color: red}$ is excluded from the biclique, but do not introduce a similar variable for \texttt{background: blue}.
    Since \texttt{background: blue} does not appear in the edge order, its presence can never cause a violation of the edge ordering.
    This is not the case for \texttt{color: red}, hence we still need the allow the possibility of removing it to satisfy the edge order.
\item
    For a rule-merging application to reduce the size of the file, it must remove at least two nodes from the covering.
    Hence, for each
    $\mbiclique_i$
    let
    $j^i_l$
    be the index of the \emph{second} rule in $\covering$ containing some edge in
    $\mbiclique_i$
    (not necessarily the same edge), and
    $j^i_h$
    be the index of the \emph{last} rule containing some edge in
    $\mbiclique_i$
    Without loss of generality we assert
    \[
        \bicliquevble = i \Rightarrow \brac{
            \inpos \geq j^i_l
            \land
            \inpos \leq j^i_h
        } \ .
    \]

\item
    We performed the following optimisation when calculating
    $({\{\mbiclique_i\}}_{i=1}^{\nummbicliques}, \forbidden)$.
    The majority ($\pctzerobadbicliques\%$) of benchmarks  all maximal bicliques are orderable at all positions, and the mean percentage of maximal bicliques that were unorderable at some position was $\meanbadbicliques\%$, with a maximum of $\maxpctbadbicliques\%$.
    However, there are one or two of the largest examples of our experiments where this enumeration took a minute or so.
    Since the number of maximal bicliques that are unorderable at some position is so small, we decided in our implementation to simply remove all such bicliques from the analysis.
    In this case
    $\{\mbiclique_i\}_{i=1}^{\nummbicliques}$
    is an enumeration of all maximal bicliques that are orderable at all positions $j$, and $\forbidden$ maps all values of $[1,m]$ to $\emptyset$.
    This means that we may not be able find the optimal rule-merging opportunity
        as not all bicliques are available, but since the amount of time
        required to find a rule-merging opportunity is reduced, we are able to
        find more merging opportunities to apply.

\item
    Finally, we allowed a multi-threaded partitioned search.
    This works as follows.
    The search space is divided across $n$ threads, and each thread partitions its search space into $m$ partitions.
    During iteration $j$, thread $k$ allows only those $\mbiclique_i$ where
    $i = k*m + (j \mod m)$.
    If the fastest thread finds a merging opportunity in $t$ seconds, we wait up to $0.1t$ seconds for further instances.
    We take the best merging opportunity of those that have completed.
    A thread reports ``no merging opportunities found'' only if none of its
        partitions contain a merging opportunity.

    For the experiments we implemented a simple heuristic to determine the number of threads and partitions to use.
    We describe this heuristic here, but first note that better results could likely be obtained with systematic parameter tuning rather than our ad-hoc settings.
    The heuristic used was to count the number of nodes in the CSS file;
    that is, the total number of selectors and property declarations (this total includes repetitions of the same node -- we do not identify repetitions of the same node).
    The tool creates enough partitions to give up to 750 nodes per partition.
    If only two partitions are needed, only one thread is used.
    Otherwise \satcss first creates new threads -- up to the total number of CPUs on the machine.
    Once this limit is reached, each thread is partitioned further until the following holds:
    $\mathit{number\ of\ threads} * \mathit{partitions\ per\ threads} * 750 \leq \mathit{number\ of\ nodes}$.

\end{itemize}
For the edge ordering:
\begin{itemize}
\item
    We do not support attribute selectors in full, but perform simple satisfiability checks on the most commonly occurring types of constraints.
    These cover all constraints in our benchmarks.
    However, we do support shorthand properties.

\item
    Instead of doing a full Existential Presburger encoding, we do a backwards emptiness check of the CSS automata.
    This backwards search collects smaller Existential Presburger constraints describing the relationship between siblings in the tree, and checks they are satisfiable before moving to a parent node.
    Global constraints such as ID constraints are also collected and checked at the root of the tree.
    This algorithm is described in Appendix~\ref{sec:optimised-aut-emp}.
\end{itemize}

\subsection{Results}

The experiments were run on a Dell Latitude e6320 laptop with 4Gb of RAM and four 2.7GHz Intel i7-2620M cores.
The Python interpreter used was PyPy 5.6~\cite{pypy} and the backend version of Z3 was 4.5.
Each experiment was run up to a timeout of 30 minutes.
In the case where CSS files used media queries (which our techniques do not yet
support), we used stripmq~\cite{stripmq} with default arguments to remove the
media queries from CSS files.
    Also, we removed whitespaces, comments, and invalid CSS from all the
CSS files before they were processed by the minifiers and our tool.

We used a timeout of 30 minutes because minification only needs to be applied once before deployment of the website.
We note that the tool finds many merging opportunities during this period and a
minified stylesheet can be produced after each application of rule-merging.
This means the user will be able to trade time against the amount of size reduction made.
Moreover, applications of rule-merging tend to show a ``long-tail'' where the
first applications found provide larger savings and the later applications show
diminishing returns (see Figure~\ref{fig:long-tail}).
The returns are diminishing because our MaxSAT encoding always searches for the rule merging opportunity with the largest saving.

\subsection{Main Results}

Table~\ref{tab:percentile-tables} and Figure~\ref{fig:box-plots} summarise the results.
We compared our tool with six popular minifiers in current usage~\cite{minifier-list}.
There are two groups of results:
the first set show the results when either our tool or one of the minifiers is used alone on each benchmark;
the second show the results when the CSS files are run first through a minifier and then through our tool.
This second batch of results shows a significant improvement over running single tools in isolation.
Thus, our tool complements and improves existing minification algorithms and our
techniques may be incorporated into a suite of minification techniques.

Table~\ref{tab:percentile-tables} shows
the savings, after whitespaces and comments are removed,
obtained by \satcss and the six minifiers when they are used alone and together.
The table presents the savings in seven percentile ranks that include the minimal (0th),
the median (50th), and the maximal values (100th).
The upper half of the table shows the savings obtained in bytes,
while the lower half shows the savings as percentages of the original file sizes.
The first seven columns in the table show the savings
when either our tool or one of the minifiers is used alone;
the rest of the columns show the results when the CSS files
are processed first by a minifier and then by our tool.
Figure~\ref{fig:box-plots} shows the same data in visual form.
It can be seen that \satcss tends to achieve greater savings
than each of the six minifiers on our benchmarks.
Furthermore, even greater savings can be obtained
when \satcss is used in conjunction with any of the six minification tools.
More precisely, when run individually, \satcss achieves savings
with a third quartile of \satcssthird\% and a median value of \satcssmedian\%,
while the six minifiers achieve savings with third quantiles and medians
up to \miniuptothird\% and \miniuptomedian\%, respectively.
When we run \satcss after running any one of these minifiers,
the third quartile of the savings can be increased to \afterminithird\%
and the median to \afterminimedian\%.
The additional gains obtained by \satcss on top of the six minifiers
(as a percentage of the original file size)
have a third quartile of \gainsthird\% and a median value of \gainsmedian\%.
Moreover, the ratios of the percentage of savings made by \satcss to
the percentage of savings made by the six minifiers have third quartiles of
at least \ratiothird\% and medians of at least \ratiomedian\%. These figures clearly indicate a
substantial proportion of extra space savings made by \satcss.
We comment in the next section on how our work may be integrated with existing tools.
\begin{wrapfigure}{r}{.4\linewidth}
    \centering
    \includegraphics[width=\linewidth]{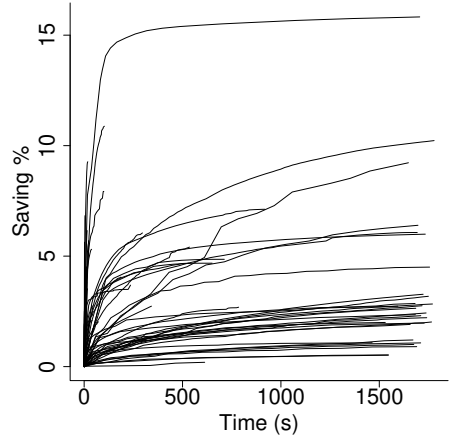}
    \caption{\label{fig:long-tail} Savings against time for \satcss on each benchmark}
\end{wrapfigure}

In Figure~\ref{fig:long-tail} we plot for each benchmark, the savings made as time progresses.
Each line represents one benchmark.
Our algorithm repeatedly searches for and applies merging opportunities.
Once one opportunity has been found, the search begins again for another.
Hence, the longer \satcss is run, the more opportunities can be found, and the more space saved.
Since we search for the optimal merging opportunities first, we can observe a long-tail, with the majority of the savings being made early in the run.
Hence, SatCSS could be stopped early while still obtaining a majority of the benefits.
We further note that in \numcompleted\ cases, \satcss terminated once no more merging opportunities could be found.
In the remaining \numtimedout\ cases, the timeout was reached.
If the timeout were extended, further savings may be found.

\input{experiments-graph}

Finally, we remark on the validation of our tool.
First, our model is built upon formal principles using techniques that are proven to maintain the CSS semantics.
Moreover, our tool verifies that the output CSS is semantically equivalent to the input CSS.
Thus we are confident that our techniques are also correct in practice.
A reliable method for truly comparing whether the rendering of webpages using the original CSS and the minified CSS is identical is a matter for future work, for which the recent Visual Logic of Panchekha \emph{et al.}~\cite{PGETK18} may prove useful.
In lieu of a systematic validation, for each of our experiment inputs, we have visually verified that the rendering remains unchanged by the minification.
Such a visual inspection can, of course, only be considered a sanity-check.

\subsection{Evaluations of Optimisations}
\label{sec:eval_op}

For each of our optimisations, we ran \satcss on the unminified \numbenchmarks\ benchmarks with the optimisation disabled.
The results are shown in Figure~\ref{fig:opt-emp} and Figure~\ref{fig:box-plots-opt}.
In Figure~\ref{fig:box-plots-opt}, the ``satcss'' column shows the performance of \satcss implemented as described above.
As described above, of the $\numbenchmarks$ benchmarks, \satcss completed within the timeout (no more merging-opportunities could be found) in \numcompleted\ cases, and was stopped early due to the timeout in the remaining \numtimedout.
The comparison with the effect of disabling optimisations is shown in Table~\ref{tbl:timeouts}.

Figure~\ref{fig:opt-emp} shows the time taken by \satcss to construct the edge order for each benchmark using the optimised emptiness of intersection algorithm presented in Appendix~\ref{sec:optimised-aut-emp} and the non-optimised encoding presented in Section~\ref{sec:intersection}.
The timeout was kept at 30 minutes.
The optimised algorithm completed the edge order construction on all benchmarks, while the timeout was reached in \unoptimisedtimeouts\ cases by the unoptimised algorithm.
A clear advantage can be seen for the optimised approach.
\begin{wrapfigure}{r}{.35\linewidth}
    \centering
    \includegraphics[width=\linewidth]{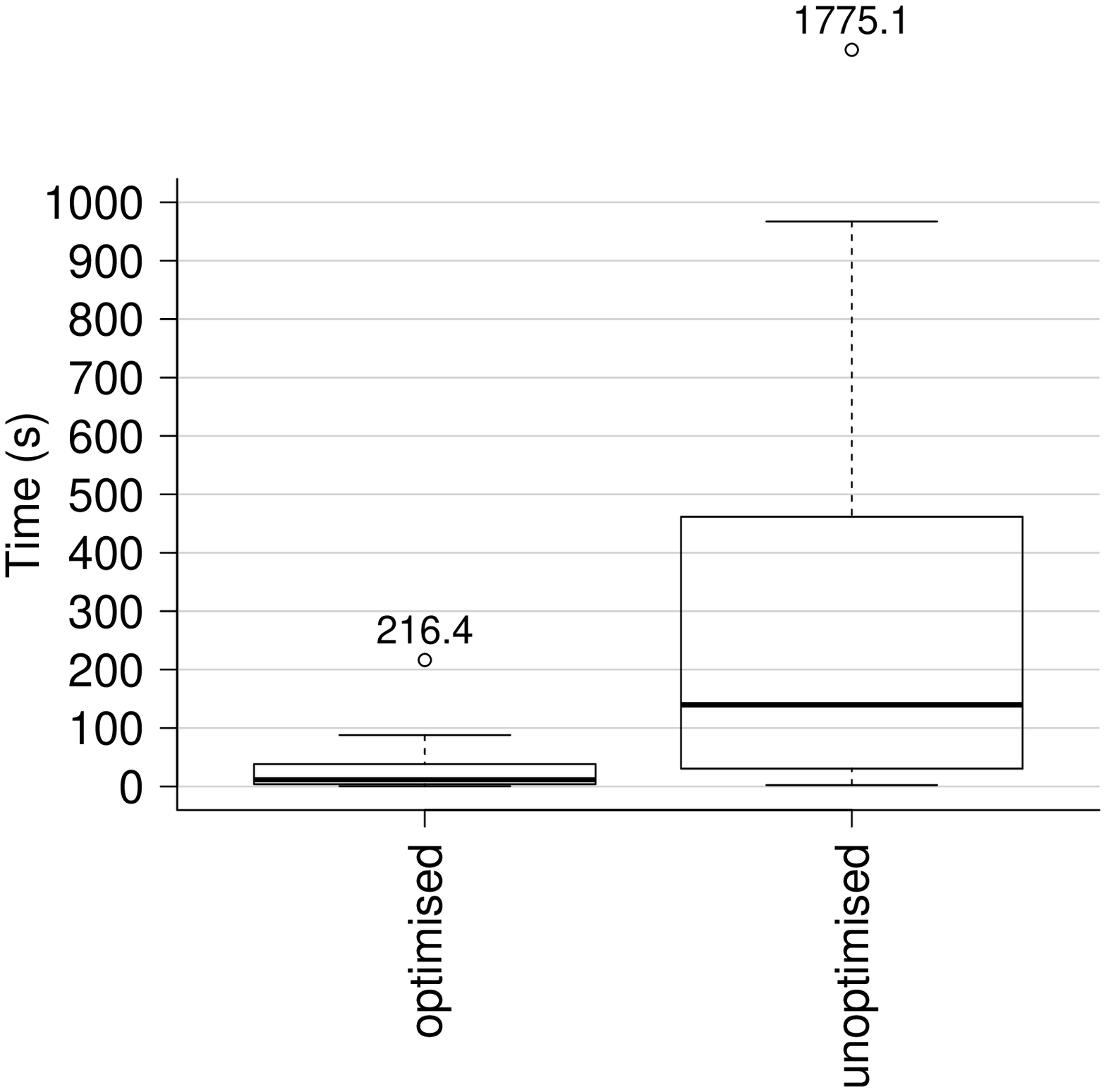}
	\caption{\label{fig:opt-emp} Box plots of the time taken to construct the CSS edge order with the optimised and unoptimised emptiness of intersection tests}
\end{wrapfigure}

\begin{figure}
    \centering
    \includegraphics[width=0.7\linewidth]{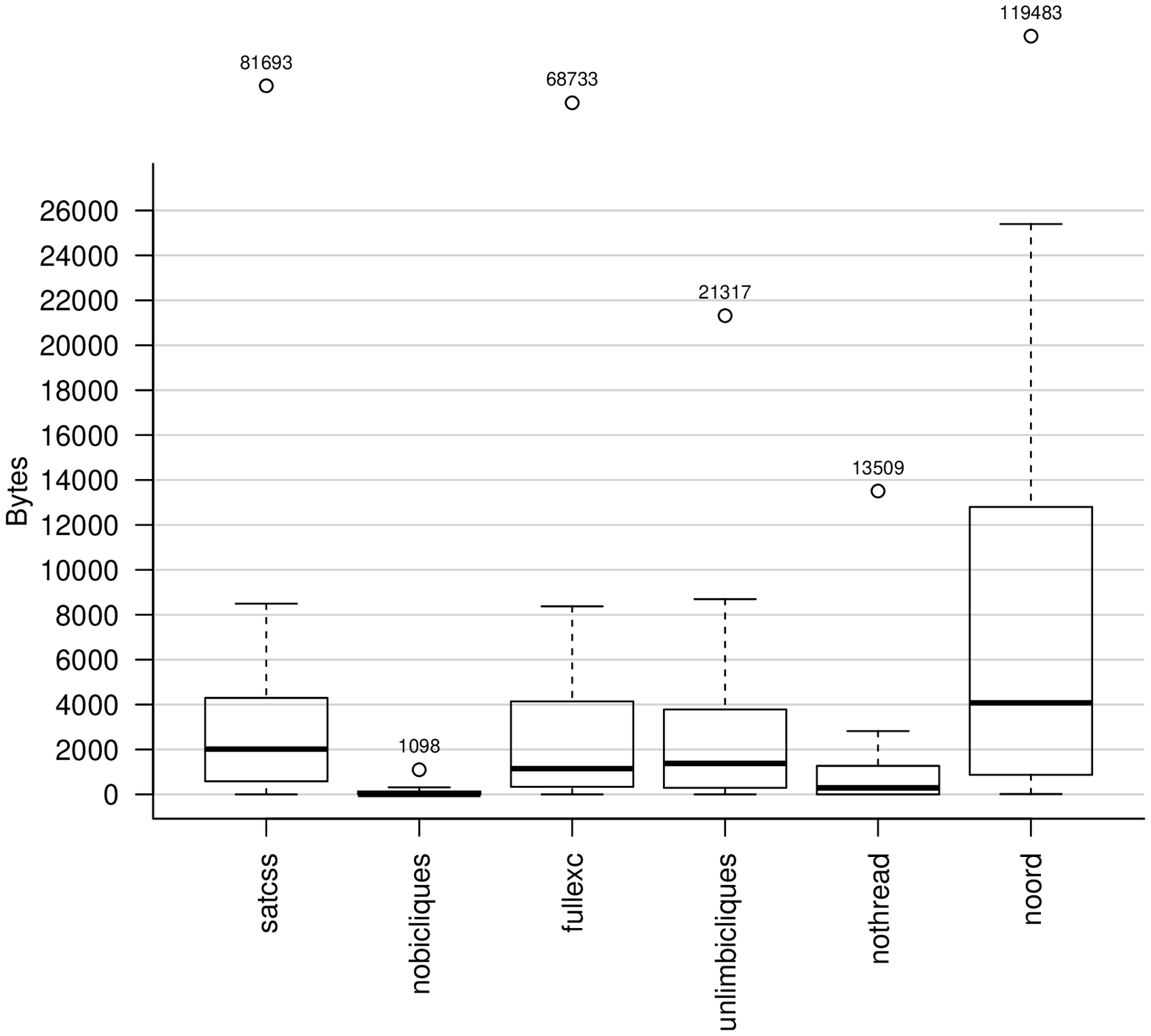}
	\hspace{.2em}
    \includegraphics[width=0.7\linewidth]{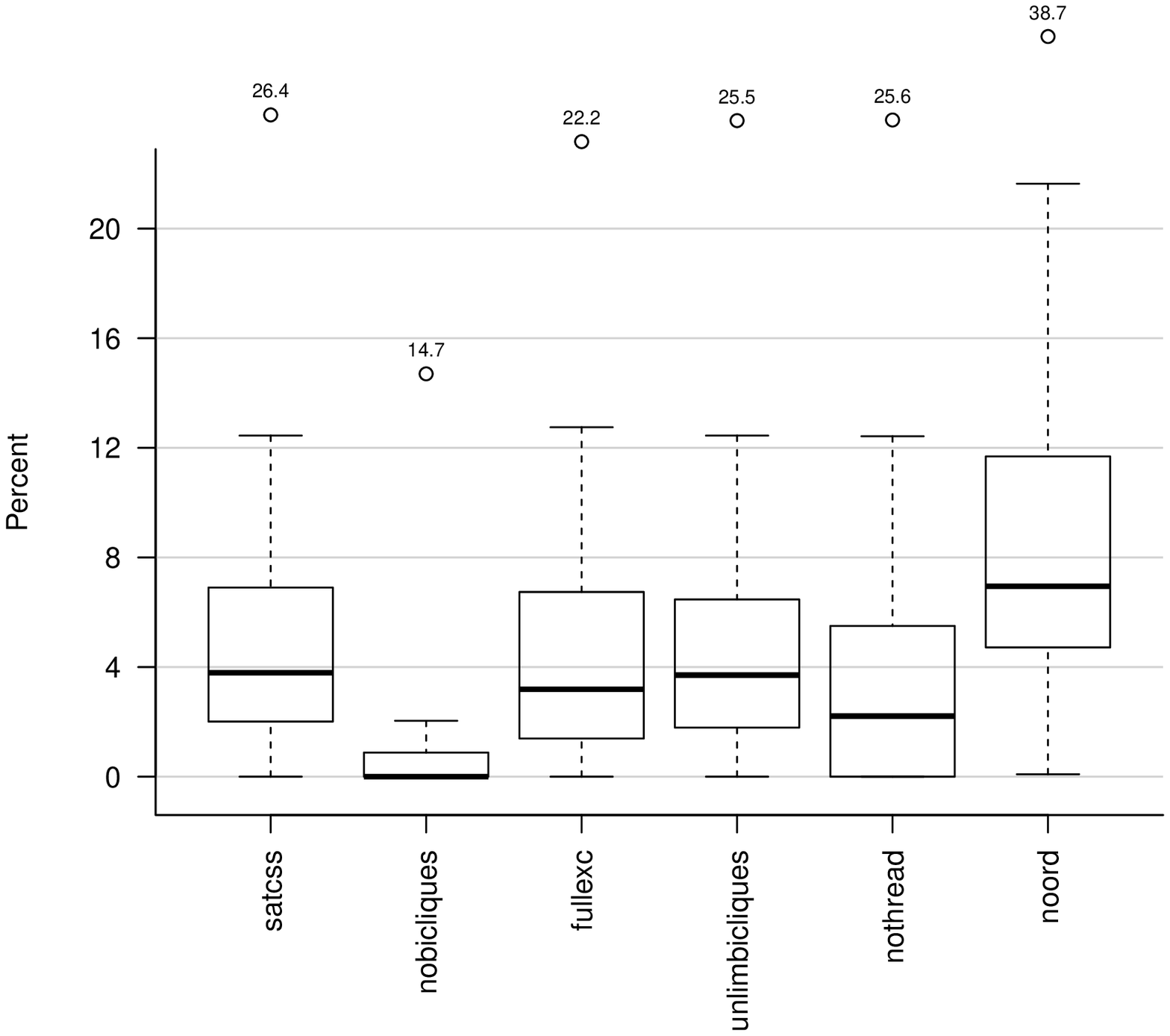}
	\caption{\label{fig:box-plots-opt} The savings in bytes (above) and in percentages (below) with certain optimisations disabled}
\end{figure}

Next we compared the straightforward encoding of the rule-merging problem into Max-SAT discussed at the start of Section~\ref{sec:graph2maxsat} with the biclique encoding which was the main topic of Section~\ref{sec:graph2maxsat}.
\satcss using the straightforward encoding appears as ``nobicliques'' in Figure~\ref{fig:box-plots-opt}.
With the straightforward encoding, the tool completed within the timeout in \nobicliquesnumcompleted\ cases, and reached the timeout in the remaining \nobicliquesnumtimedout.
It can been seen that there is a clear benefit to the biclique encoding.

We also consider the optimisation that introduces variables $\excludevble{k}$ only for nodes appearing in the edge order, rather than for all nodes in a biclique.
Disabling this optimisation appears as ``fullexc'' in Figure~\ref{fig:box-plots-opt}.
With this optimisation disabled, the tool completed within the timeout in \fullexcnumcompleted\ cases, and reached the timeout in the remaining \fullexcnumtimedout.
We can see a modest gain from this simple optimisation.

We then study the effect of removing all unorderable bicliques.
The ``unlimbicliques'' column in Figure~\ref{fig:box-plots-opt} shows the effect of allowing
$({\{\mbiclique_i\}}_{i=1}^{\nummbicliques}, \forbidden)$
to be calculated in full.
That is, unorderable bicliques are split into orderable sub-bicliques.
With full biclique enumeration, the tool completed within the timeout in \unlimbicliquesnumcompleted\ cases, and reached the timeout in the remaining \unlimbicliquesnumtimedout.
This had only a small effect on performance, which is expected.
The purpose of this limiting biclique generation in \satcss is to prevent the rare examples of unorderable bicliques from having a large effect on performance in some cases.

The next optimisation in Figure~\ref{fig:box-plots-opt}, ``nothread'', shows the performance of \satcss with multi-threading and partitioning disabled.
That is, the search space is not split up across several Max-SAT encodings.
Without multi-threading or partitioning, the tool completed within the timeout in \nothreadnumcompleted\ cases, and reached the timeout in the remaining \nothreadnumtimedout.
There is a noticeable degradation in performance when this optimisation is disabled.
The final column in Figure~\ref{fig:box-plots-opt} studies the following hypothetical scenario.
The reader may wonder whether CSS developers may be able to improve performance by specifying invariants of the documents to which the CSS will be applied.
For example, the developer may specify that a node with class \texttt{a} will never also have class \texttt{b}.
In this case pairs such as
$(\texttt{.a}, \texttt{.b})$
can be removed from the edge-order.
Thus, document-specific knowledge may improve performance by reducing the number of ordering constraints that need to be maintained.
The column ``noord'' shows the performance of \satcss when the edge-ordering is empty, i.e., there is no edge ordering.
Without the edge ordering, the tool completed within the timeout in \noordnumcompleted\ cases, and reached the timeout in the remaining \noordnumtimedout.
This represents a best-case scenario if document invariants were considered.
Thus, if we were to extend \satcss to also take a set of invariants, we could see a large improvement in the savings found.
However, we also note that these savings do not dwarf the performance of \satcss without invariants.
\begin{wraptable}{r}{.35\linewidth}
    \centering
    \includegraphics[width=\linewidth]{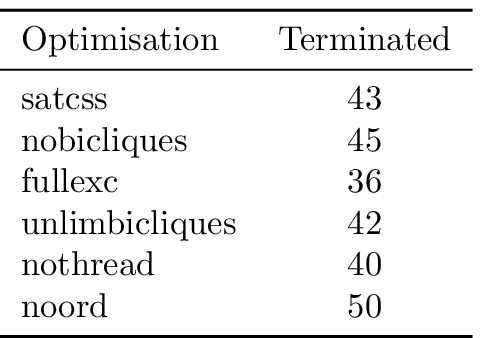}
    \caption{\label{tbl:timeouts} The number of examples terminating with no more discovered merging opportunities in standard mode and with certain optimisations disabled}
\end{wraptable}

%% file: experiments-graph.tex
\begin{table}
    \centering
    \includegraphics[width=1\linewidth]{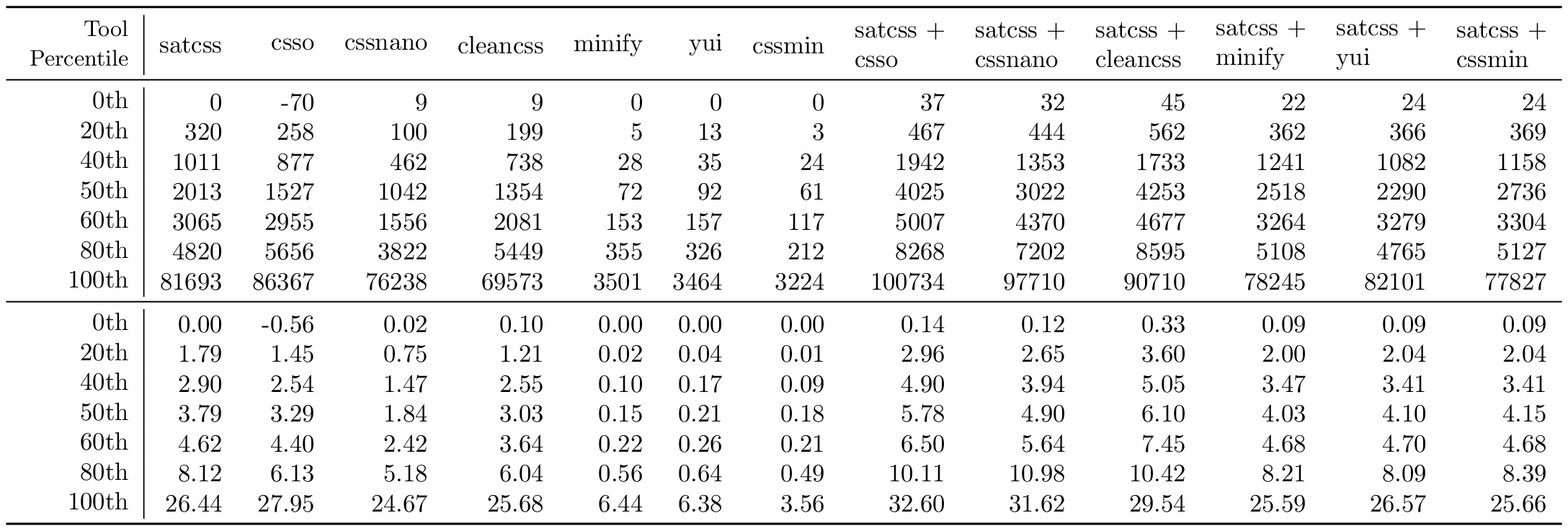}
  \caption{\label{tab:percentile-tables} Percentile ranks of the savings in bytes (above) and in percentages (below)}
\end{table}

\begin{figure}
    \centering
    \includegraphics[width=0.7\linewidth]{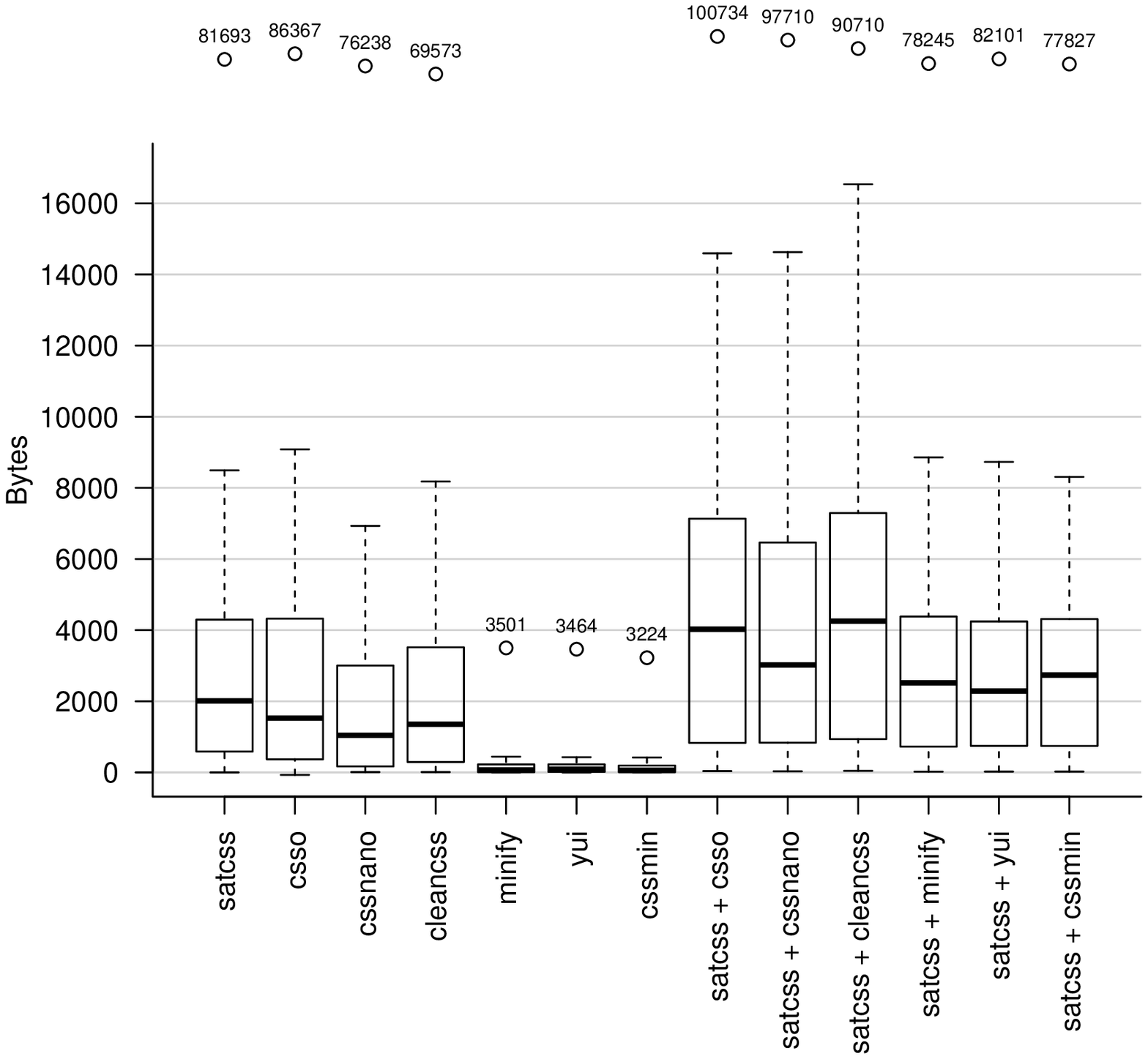}
	\hspace{.2em}
    \includegraphics[width=0.7\linewidth]{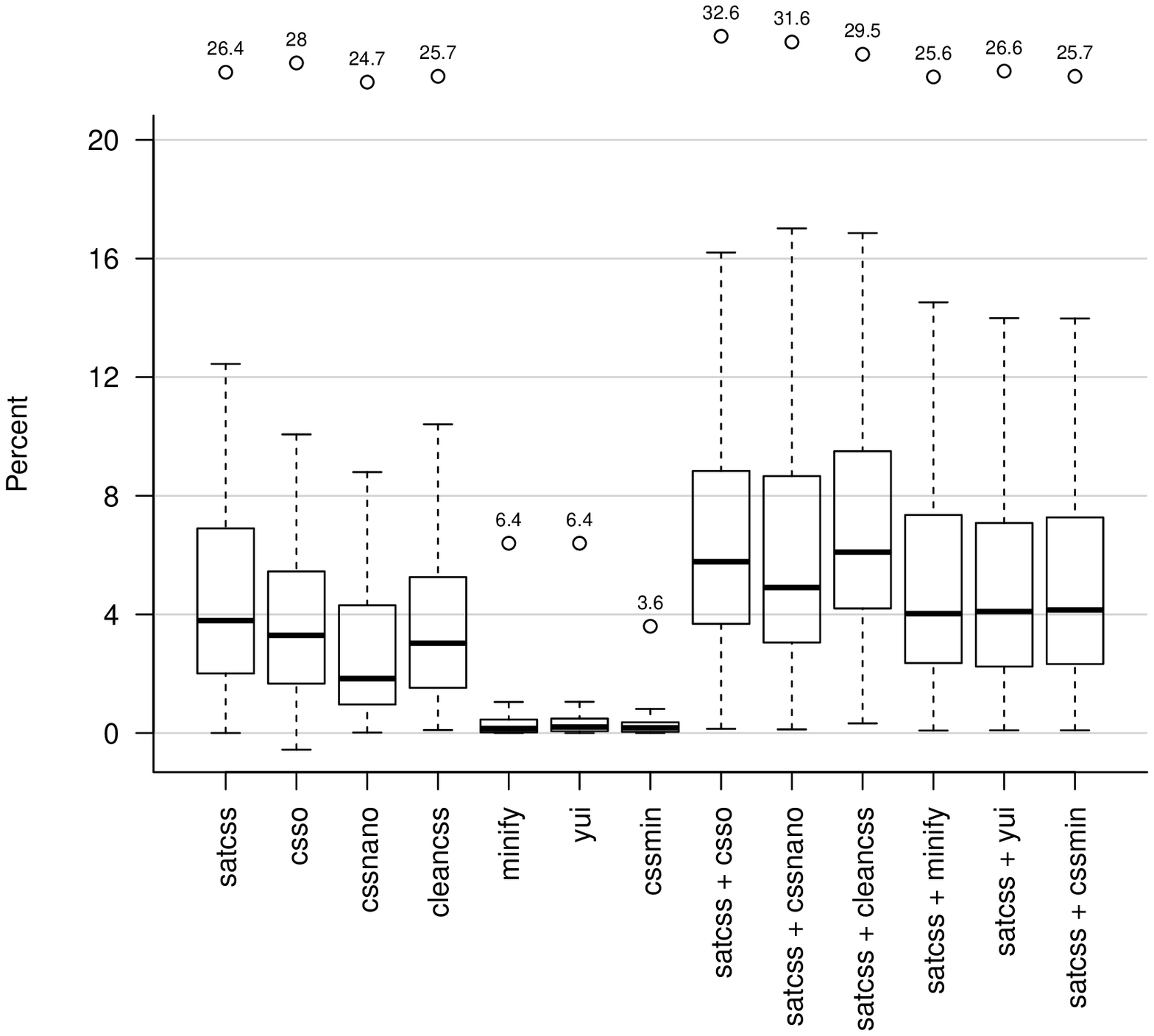}
	\caption{\label{fig:box-plots} Box plots of the savings in bytes (above) and in percentages (below)}
\end{figure}

%% file: related.tex
\section{Related Work}
\label{sec:related}

CSS minification started to receive attention in the web programming 
community around the year 2000. 
To the best of our knowledge, the first major tools that could perform CSS
minification were Yahoo!~YUI Compressor \cite{yui} and Microsoft Ajax Minifier,
both of which were developed around 2004--2006. This is followed by the 
development of many other CSS minifiers including (in no particular order) 
\texttt{cssmin} \cite{cssmin}, \texttt{clean-css} \cite{cleancss}, \texttt{csso}
\cite{csso}, \texttt{cssnano} \cite{cssnano}, and \texttt{minify} \cite{minify}.
Such minifiers mostly apply syntactic transformations including removing 
whitespace characters, comments, and replacing strings by their abbreviations
(e.g. \verb+#f60+ by \verb+#ff6600+). More and more advanced optimisations are
also being developed. For example, \texttt{cssnano} provides a limited support
of our rule-merging transformations, wherein only \emph{adjacent} rules may be 
merged.
The lack of techniques for handling the order dependencies of
CSS rules \cite{Souders-book} was most likely one main reason why a more general
notion of rule-merging transformations (e.g. that can merge rules 
that are far away in the file) is not supported by CSS minifiers.

In our experiments, we ran \satcss after running the existing minifiers described above.
It is likely that the order of execution is important: the rewrites applied by existing minifiers will put the CSS into a more normalised format, which will improve the possibility of selectors sharing the same declarations.
Moreover, these minifiers implement ad-hoc techniques, such as the limited
rule-merging transformations of \texttt{cssnano} described above.
After an application of our tool it is possible that some of these ad-hoc techniques may become applicable, leading to further savings.
Thus, we posit that our techniques could be combined with the techniques of existing minifiers in a combined minification loop, which is run until a fixed point or timeout is reached.

Although the importance of CSS minification is understood in industry,
the problem received little attention in academia until very recently. 
Below we will mention the small number of existing work on formalisation of CSS 
selectors and CSS minification, and other relevant work.

The lack of theories for reasoning about CSS selectors was first mentioned
in the paper \cite{Geneves12}, wherein the authors developed a tree logic
for algorithmically reasoning about CSS selectors, i.e., by developing
algorithms for deciding satisfiability of formulas in the logic. This
formalisation does \emph{not} capture the whole class of CSS3 selectors; as 
remarked in their follow-up paper \cite{Geneves-style}, the logic captures 
70\% of the selectors in their benchmarks from real-world CSS files. 
In particular, they do not fully support attribute selectors
(e.g. \verb+[l*="bob"]+). 
Our paper provides a \emph{full} formalisation of CSS3 selectors.
In addition, their tree logic can express properties that are \emph{not}
expressible in CSS. Upon a closer inspection, their logic is at least as 
expressive as $\mu$-calculus, which was a well-known logic in database theory 
for formalising 
query languages over XML documents, e.g., see \cite{Libkin-survey,Neven-survey} 
for two wonderful surveys. As such, the complexity of satisfiability for
their tree logic is EXPTIME-hard. Our formalisation captures \emph{no more}
than the expressive power of CSS3 selectors,
which helps us obtain the much lower complexity NP and enables the use of 
highly-optimised SMT-solvers.
There is a plethora of other work on logics
and automata on unranked trees (with and without data), e.g., see
\cite{Marx05,BFK03,GK02,CM07,CM09,CLM10,MR05,Hidders03,Wenfei-PODS05,NS06,GF05,LS10,GL06,Diego-thesis,Counting-free,DLT12,GLSG15}
and the surveys \cite{Libkin-survey,Neven-survey,data-aut-survey}.
However, none of these formalisms can capture certain
aspects of CSS selectors (e.g.\ string constraints on attribute values), even
though they are much more powerful than CSS selectors in other aspects
(e.g.\ navigational). 

There are a handful of research results on CSS minification that 
appeared in recent years. Loosely speaking, these optimisations can be 
categorised into two: (a) \emph{document-independent}, and (b) 
\emph{document-dependent}. Document-independent optimisations are program
transformations that are performed completely independently of the web documents
(XML, HTML, etc.). On the other hand, document-dependent optimisations
are performed with respect to a (possibly infinite) set of documents.
Existing CSS minifiers \emph{only} perform document-independent
optimisations
since they are meant to preserve the semantics 
of the CSS file \emph{regardless} of the DOMs to which the CSS file
is applied. Our work in this paper falls within this category too.
Such optimisations are often the most sensible option \emph{in 
practice}
including (1) the case of \emph{generic} stylesheets as part of web templates 
(e.g.~WordPress), and (2) the case when the DOMs are generated by a program.
Case (2) requires some explanation. 
A typical case of DOMs being generated by programs occurs in HTML5 web pages.
An HTML5 application comes with  a finite set of HTML 
documents, JavaScript code, and CSS files. The presence of JavaScript means that
potentially \emph{infinitely} many possible DOM-trees could be generated and 
displayed by 
the browser. Therefore, a CSS minification should \emph{not}
affect the rendering of \emph{any} such tree by the browser. Although
a document-dependent optimisation (that take these infinitely many trees
into account) seems appropriate, this is far from realistic given the
long-standing difficulty of performing sound static analysis for JavaScript
especially in the presence of DOM-trees
\cite{SSDT13,SDCST12,AM14,JMT09,JMM11,HLO15}. 
This would make an interesting
long-term research direction with many more advances on static analysis
for JavaScript. 
\mhchanged{
    However, the problem is further compounded by the multitude of frameworks deployed on the server-side for HTML creation (e.g. PHP, Java Server Pages, etc.), for which individual tools will need to be developed.
}
Until then, a practical minifier for HTML5 applications will 
have to make do with document-independent optimisations for CSS files.

The authors of \cite{Cilla} developed a document-dependent dynamic analysis
technique for detecting and removing unused CSS selectors in a CSS file that
is part of an HTML5 application. A similar tool, called UnCSS \cite{UnCSS},
was also later developed. This is done by instrumenting the HTML5
application and removing CSS selectors that have not been used by the end
of the instrumentation. The drawback of this technique is that it cannot
test all possible behaviours of an HTML5 application and may
may accidentally delete selectors that can in reality be used by the
application. It was noted in \cite{HLO15} that such tools may accidentally
delete selectors, wherein the HTML5 application has event listeners that 
require user interactions. The same paper \cite{HLO15} develops a static 
analysis technique for overapproximating the set of generated DOM-trees by 
using tree rewriting for abstracting the dynamics of an HTML5 application. 
The technique, however, covers only a very small subset of JavaScript, and
is difficult to extend without first overcoming the hard problem of
static analysis of JavaScript.


The authors of \cite{Geneves-style} applied their earlier techniques 
\cite{Geneves12} to develop a document-independent CSS minification
technique that removes ``redundant'' property declarations, and 
merges two rules with semantically equivalent selectors. The optimisations that 
they considered are orthogonal to and can be used in conjunction with the 
optimisation that we consider in this paper. 
More precisely, they developed an algorithm for 
checking selector subsumption (given two selectors $S_1$ and $S_2$, whether the 
selector $S_1$ is subsumed by the selector $S_2$, written $S_1 \subseteq S_2$).
A redundant property declaration $p$ in a rule $R_1$ with a selector $S_1$ can,
then, 
be detected by finding a rule $R_2$ that also contains the declaration $p$ and 
has a selector $S_2$ with a higher specificity than $S_1$ and that 
$S_1 \subseteq S_2$. As another example, whenever we can show that the selectors
$S_1$ and $S_2$ of two rules $R_1$ and $R_2$ to be semantically equivalent
(i.e. $S_1 \subseteq S_2$ and $S_2 \subseteq S_1$), we may merge $R_1$ with
$R_2$ under certain conditions. The authors of \cite{Geneves-style} 
provided sufficient conditions for performing this merge by
relating the specificities of $S_1$ and $S_2$ with the specificities of
other related selectors in the file (but not accounting for the order of 
appearances of these rules in the file). In general, a CSS rule might have
multiple selectors (a.k.a. selector group), each with a different
specificity, and it is not clear from the presentation of the paper 
\cite{Geneves-style} how their optimisations extend to the general case.

\OMIT{
Their logic is akin to $\mu$-calculus,
which was a well-known logic in database theory for formalising query languages 
over XML documents, e.g., see \cite{Libkin-survey,Nevin-survey} for two
wonderful surveys. 
}

The authors of \cite{Mesbah-refactoring} developed a
document-dependent\footnote{More precisely, dependent on a given finite set
of HTML documents} CSS minification method with an advanced type of rule 
merging as one of their optimisations.
%
\mhchanged{
    This is an ambitious work utilising a number of techniques from areas such as data mining and constraint satisfaction.
    Although their work differs from ours because of its document-dependence,
    the use of rule-merging is closely related to our own, hence we will describe in detail some key differences with our approach.
}
The techniques presented in this paper can be
viewed as a substantial generalisation of their rule merging optimisation.
Loosely speaking, in terms of our graph-theoretic
framework, their technique first enumerates all maximal bicliques with at least
two selectors. This is done with the 
help of an association rule mining algorithm (from data mining) with a 
set of property declarations viewed as an \emph{itemset}. 
Second, for each such maximal biclique $B$, a value $n$ is computed that 
reflects how much saving will be obtained if $B$ could somehow be inserted into
the file and every occurrence of each property declaration in $B$ is erased
from the rest of the CSS file. Note that $n$ is independent of where $B$
is inserted into the CSS file. Third, for each such maximal biclique $B$ (ranked
according to their values in a non-increasing order), a solver for the
(finite-domain) constraint satisfaction problem is invoked
to check whether $B$ can be placed in the file (with every occurrence of 
each property declaration in $B$ is erased from the rest of the CSS file) while
preserving the order dependency. If this check fails, the solver will also
be invoked to check if one can insert sub-bicliques of $B$ (with a maximal set 
$S$ of selectors with $|S| \geq 2$) in the file. Possible positions in the file 
to place each selector of $B$ are encoded as variables constrained by
the edge order dependency that is \emph{relativised} to the provided
HTML documents. 
To test whether two edges should be ordered in this relativised edge order, the 
selectors are not subject to a full intersection test, but instead a 
\emph{relativised intersection} test that checks
whether there is some node in the given finite set of html documents that is 
matched by both selectors.
%
Their techniques do not work when the HTML documents are not given, which 
we handle in this paper.
Another major difference to our paper is that their algorithm sequentially
goes through every maximal biclique $B$ (ranked according to their values) and 
checks 
if it can be inserted into
the file, which is computationally too prohibitive especially when the
(unrelativised) edge order $\edgeOrder$ is used. 
Our algorithm, instead, fully relegates the search of an appropriate $B$ and
the position in the file to place it to a 
highly-optimised Max-SAT solver, which scales to real-world CSS files.
In addition, their type of rule merging is also more restricted than ours for
two other reasons. First, the new rule inserted into the file has
to contain a maximal set of selectors. This \mhchanged{prohibits many
rule-merging opportunities} and in fact
does not subsume the 
merging adjacent rule optimisation of \texttt{cssnano} \cite{cssnano} in
general. For example, consider the CSS file
\begin{center}
\begin{minted}{css}
.class1 { color:blue }
.class2 { color:blue }
.class3 { color:red }
.class4 { color:blue }
\end{minted}
\end{center}
Notice that we cannot group together the first, second, and fourth rules 
since this would change
the colour of a node associated with the classes \verb+class2+ and 
\verb+class3+, or with the classes \verb+.class3+ and \verb+class4+.
On the other hand, the first two rules can be merged resulting in the new file
\begin{center}
\begin{minted}{css}
.class1, .class2 { color:blue } 
.class3 { color:red }
.class4 { color:blue }
\end{minted}
\end{center}
However, this is not permitted by their merging rule since 
\verb+.class1,.class2{color:blue}+ does not contain a maximal set of 
selectors. Second, given a maximal biclique $B$, their merging operation 
erases \emph{every} occurrence of the declarations of $B$ everywhere else in 
the file. This further rules out certain rule-merging opportunities.
For example, consider the CSS file
\begin{center}
\begin{minted}{css}
.class1 { color:blue; font-size: large }
.class2 { color:blue; font-size: large }
.class4 { font-size: large }
.class3 { color:red }
.class4 { color:blue }
\end{minted}
\end{center}
and observe the following maximal biclique in the file.
\begin{center}
\begin{minted}{css}
.class1, .class2, .class4 { color:blue; font-size: large }
\end{minted}
\end{center}
Unfortunately, this is not a valid opportunity using their merging rule
since this CSS file is not equivalent to
\begin{center}
\begin{minted}{css}
.class1, .class2, .class4 { color:blue; font-size: large }
.class3 { color:red }
\end{minted}
\end{center}
nor to the following file.
\begin{center}
\begin{minted}{css}
.class3 { color:red }
.class1, .class2, .class4 { color:blue; font-size: large }
\end{minted}
\end{center}
In this paper, we permit duplicate declarations, and would insert this 
maximal biclique just before the fourth rule in the file (and perform trim)
resulting in the following equivalent file.
\begin{center}
\begin{minted}{css}
.class1, .class2, .class4 { color:blue; font-size: large }
.class3 { color:red }
.class4 { color: blue }
\end{minted}
\end{center}
Finally, each maximal biclique in the enumeration of \cite{Mesbah-refactoring} 
does \emph{not}
allow two property declarations with the same property name. As we 
explained in Section \ref{sec:graph2maxsat}, CSS rules satisfying this
property are rather common since they support fallback options. Handling
such bicliques (which we do in this paper) requires extra technicalities,
e.g., the notion of orderable bicliques, and adding an order to the 
declarations in a biclique.

We also mention that there have been works \cite{MB10,PT16,HB15}
on solving web page layout using constraint solvers. These works are 
orthogonal to this paper.
For example, \cite{PT16} provides a mechanised formalisation of the semantics 
of CSS for web page layout (in quantifier-free linear arithmetic), which allows 
them to use an SMT-solver to automatically reason about layout. 
Our work provides a full formalisation of CSS selectors, which is not
especially relevant for layout. Conversely,
the layout semantics of various property declarations is not relevant in
our CSS minification problem.

\alchanged{
Finally, we also mention the potential application of rule-merging to CSS
refactoring. This was already argued in \cite{Mesbah-refactoring}, wherein the
metric of minimal file size is equated with minimal redundancies. More research
is required to discover other classes of CSS transformations and metrics that 
are more meaningful in the context of improving the design of stylesheets.
Constraint-based refactoring has also been studied in the broader context of 
programming languages, e.g., see \cite{Stein18,Tip11}. It would be interesting
to study how refactoring for CSS can be cast into the framework of
constraint-based refactoring as previously studied (e.g. in \cite{Stein18}).
}

%% file: conclusion.tex
\section{Conclusion and Future Work}
\label{sec:conclusion}

We have presented a new CSS minification technique via merging 
similar rules. Our techniques can handle stylesheets composed of CSS rules which contain a set of CSS Level 3 selectors and list of property declarations.
This technique exploits the fact that new rules may be introduced that render other parts of the document redundant.
After removing the redundant parts, the overall file size may be reduced.
Such a solution has required the development of a complete formalisation of CSS selectors and their intersection problem as well as a formalisation of the dependency ordering present in a stylesheet.
This intersection problem was solved by means of an efficient encoding to
quantifier-free integer linear arithmetic, for which there are highly-optimised
SMT solvers.
Moreover, we have formalised our CSS rule-merging problem and presented a
solution to this problem using an efficient encoding into MaxSAT formulas.
These techniques have been implemented in our tool \satcss which we have comprehensively compared with state-of-the-art minification tools.
Our results show clear benefits of our approach.

Both our formalisation and our tool strictly follow the W3C specifications.
In practice, web developers may not always follow these guidelines, and implement convenient abuses that do not trouble existing web browsers.
One particular example is the use of ID values that are not necessarily unique.
In this example case, it would be possible to treat ID values similarly to classes, and relax our analysis appropriately.
In general, one may wish to adapt our constraints to handle other common abuses.
However, this is beyond the scope of the current work.

CSS preprocessors such as Less~\cite{less} and Sass~\cite{sass} --- which 
extend the CSS language with useful features such as variables and partial 
rules --- are commonly used in web development. Since Less and Sass code is 
compiled into CSS before deployment, our techniques are still applicable.

\OMIT{
We may consider using our techniques alongside the additional information provided by the higher-level features of Less and Sass to produce an optimised compilation process.
}

There are many technologies involved in website development and deployment.
These technologies provide a variety of options for further research, some of which we briefly discuss here.

First, we may expand the scope of the CSS files we target.
For example, we may expand our definition of CSS selectors to include features
proposed in the CSS Selectors Level 4 working draft~\cite{CSS4} (i.e. still
not stable).
These features include extensions of the negation operator to allow arbitrary selectors to be negated.
It would be interesting to systematically investigate the impact of these 
features on the complexity of the intersection problem. 
We believe that such a systematic study would be informative in
determining the future standards of CSS Selectors.

Another related technology is that of \emph{media queries} that allow portions 
of a CSS file to only be applied if the host device has certain properties, 
such as a minimum screen size. 
This would involve 
defining semantics of media queries (not part of selectors), and extending our
rule-merging problem to include media queries and rules 
grouped under media queries.

Second, we could also consider additional techniques for stylesheet optimisation.
Currently we take a greedy approach where we search for the ``best'' merging
opportunity at each iteration.
Techniques such as simulated annealing allow a proportion of non-greedy steps to
be applied (i.e.\ choose a merging opportunity that does not provide the largest reduction in file size).
This allows the optimisation process to explore a larger search space, potentially leading to improved final results.
Another approach might be to search for multiple simultaneous rule-merging
opportunities.

Finally, our current optimisation metric is the raw file size.
We could also attempt to provide an encoding that seeks to find the best file size reduction after gzip compression.
[Gzip is now supported by many web hosts and most modern browsers (though not 
including old IE browsers).]
One simple technique that could help bring down the compressed file size is to sort the selectors and declarations in each rule after the minification process is done \cite{advanced-css-book}.

%% file: graph2maxsat-appendix.tex
\section{Additional Material for Max-SAT Encoding}
\label{sec:graph2maxsat-appendix}

Recall, given a valid covering
$\covering = \incovering{1}{m}$
of a CSS graph $\CSSgraph$, we aim to find a rule
$\bucket = \inbucket$
and a position $j$ that minimises the weight of $\Trim{\inSeq{\covering}{j}{\bucket}}$.

In this section we give material omitted from Section~\ref{sec:graph2maxsat}.
We begin with a definition of orderability that will be useful for the remainder of the section.
Then we will discuss how to produce the pair
$(\{\mbiclique_i\}_{i=1}^{\nummbicliques}, \forbidden)$.
Finally we will show that our encoding is correct.

\subsection{Orderable Bicliques}
\label{sec:orderable-bicliques-appendix}

To insert a biclique
$\biclique = (\selsBucket, \propsSet)$
into the covering, we need to make sure the order of its edges respects the edge order.
We can only order the edges by ordering the properties in the biclique.
More precisely, if we insert the biclique at position $j$, we need all edges in $\biclique$ that do not appear later in the file
(i.e.\ in $\incovering{j+1}{m}$)
to respect the edge order.
This is because it is only the last occurrence of an edge that influences the semantics of the stylesheet.
Thus, let
\[
    \edgeslast{\biclique}{j} = \setcomp{e \in \biclique}{\Index(e) \leq j} \ .
\]
The edge ordering implies a required ordering of
$\edgeslast{\biclique}{j}$,
which implies an ordering on the properties in $\propsSet$.
This ordering is defined as follows.
For all
$p_1, p_2 \in \propsSet$
we have
\[
    p_1 \propord{\biclique}{j} p_2
    \iff
    \exists (s_1, p_1), (s_2, p_2) \in \edgeslast{\biclique}{j}\ .\ %
        (s_1, p_1) \edgeOrder^\ast (s_2, p_2) \ .
\]
That is, we require $p_1$ to appear before $p_2$ if there are two edges
$(s_1, p_1)$
and
$(s_2, p_2)$
in $\biclique$ that must be ordered according to the transitive closure of $\edgeOrder$.
A biclique is orderable iff its properties can be ordered in such a way to respect
$\propord{\biclique}{j}$.

\begin{definition}[Orderable Bicliques]
    The biclique $\biclique$ is orderable at $j$ if
    $\propord{\biclique}{j}$
    is acyclic.
    That is, there does not exist a sequence
    $(s_1, p_1), \ldots, (s_\numof, p_\numof)$
    such that
    $(s_\idxi, p_\idxi) \propord{\mbiclique}{j} (s_{\idxi+1}, p_{\idxi+1})$
    for all $1 \leq \idxi < \numof$
    and
    $(s_1, p_1) = (s_\numof, p_\numof)$.
\end{definition}
This can be easily checked in polynomial time.
Moreover, if a biclique is orderable at a given position, a suitable ordering can be found by computing
$\propord{\biclique}{j}$,
also in polynomial time.
Thus, this proves Proposition~\ref{prop:poly-order}.

\subsection{Enumerating Maximal Rules}

Fix a covering
$\covering = \incovering{1}{m}$
of a CSS graph
$\CSSgraph = \inCSSgraph$.
We show how to efficiently create the pair
$(\{\mbiclique_i\}_{i=1}^{\nummbicliques}, \forbidden)$.
In fact, we will create the pair
$(\{\mbiclique_i\}_{i=1}^{\nummbicliques}, \forbiddenfst)$
where $\ap{\forbiddenfst}{j}$ is set of $\mbiclique_i$ for which $j$ is the smallest position such that $\mbiclique_i$ is unorderable at position $j$.
Computing $\forbidden$ from this is straightforward (though unnecessary since our encoding uses $\forbiddenfst$ directly).

Our algorithm begins with an enumeration of the maximal bicliques that are orderable at position $0$, and iterates up to position $m$, extending the enumeration at each step and recording in $\forbiddenfst$ which maximal bicliques become unorderable at each position.
In fact, we will construct a pair
$(\mbicliqueset, \forbiddenfst)$
where $\mbicliqueset$ is a set of bicliques rather than a sequence.
To construct a sequence we apply any ordering to the elements of $\mbicliqueset$.

\subsubsection{Initialisation}

At position 0, all maximal bicliques
$(\selsBucket, \propsSet)$
with
$\selsBucket \times \propsSet \subseteq \CSSedges$
are orderable.
This is because all edges appearing in
$(\selsBucket, \propsSet)$
also appear in
$\incovering{1}{m}$,
hence, the semantics of
$\inSeq{\covering}{0}{\bucket}$
will be decided by edges already in $\covering$, and thus the ordering of the properties does not matter.
Hence, we begin with a set of all maximal bicliques
$\mbicliqueset_0$.
That is, all bicliques
$(\selsBucket, \propsSet)$
such that
$\selsBucket \times \propsSet\subseteq \CSSedges$
where there is no
$(\selsBucket', \propsSet')$
with
$\selsBucket' \times \propsSet' \subseteq \CSSedges$
and
$\selsBucket \times \propsSet
 \subset
 \selsBucket' \times \propsSet'$.
Algorithms for generating such an enumeration are known.
For example, we use the algorithm from Kayaaslan~\cite{K10}.

For an initial value $\forbiddenfst_0$ of $\forbiddenfst$, we can simply take the empty function $\emptyset$.

\subsubsection{Iteration}

Assume we have generated
$(\mbicliqueset_j, \forbiddenfst_j)$.
We show how to generate the extension
$(\mbicliqueset_{j+1}, \forbiddenfst_{j+1})$.

The idea is to find all elements of
$\mbicliqueset_j$
that are not orderable at position $j+1$.
Let $\unorderablebicliques$ be this set.
We first define
\[
    \forbiddenfst_{j+1} =
    \forbiddenfst_j \cup \set{(j+1, \unorderablebicliques)} \ .
\]
Then, for each biclique
$\mbiclique \in \unorderablebicliques$,
we search for smaller bicliques contained within $\mbiclique$ that are maximal and orderable at position $j+1$.
This results in the extension
$\{\mbiclique_i\}_{i=1}^{\nummbicliques_{j+1}}$
giving us
$(\{\mbiclique_i\}_{i=1}^{\nummbicliques_{j+1}}, \forbiddenfst_{j+1})$.

Thus, the problem reduces to finding bicliques contained within some $\mbiclique$ that are maximal and orderable at position $j+1$.
We describe a simple algorithm for this in the next section.

\subsubsection{Algorithm for $(\mbicliqueset, \forbiddenfst)$}

We write
$\orderable{\mbiclique}{j}$
to assert that a biclique $\mbiclique$ is orderable at position $j$.
For now, assume we have the subroutine
$\orderablesub{\mbiclique}{j}$
which returns a set
$\mbicliqueset'$
of all orderable maximal bicliques at position $j$ contained within $\mbiclique$.
We use the following algorithm to generate
$(\mbicliqueset, \forbiddenfst)$.

\begin{algorithmic}
    \State $\mbicliqueset := \mbicliqueset_0$
    \State $\forbiddenfst := \emptyset$

    \For {$j := 1$ to $m$}
        \State $\unorderablebicliques := \emptyset$
        \ForAll {$\mbiclique \in \mbicliqueset$}
            \If {$\neg \orderable{\mbiclique}{j}$}
                \State {$\unorderablebicliques := \unorderablebicliques \cup \set{\mbiclique}$}
                \State {$\mbicliqueset := \mbicliqueset
                                          \cup
                                          \orderablesub{\mbiclique}{j}$}
            \EndIf
        \EndFor
        \State $\forbiddenfst := \forbiddenfst \cup \set{(j+1, \unorderablebicliques)}$
    \EndFor
    \State \Return $(\mbicliqueset, \forbiddenfst)$
\end{algorithmic}

Note, we can improve the algorithm by restricting the nested for all loop over elements
$\mbiclique \in \mbicliqueset$
to only those $\mbiclique$ that are not orderable at $m$.
This is because an ordering at position $m$ is also an ordering at position
$j <= m$.
Hence, these bicliques will never be unorderable and do not need to be checked repeatedly.

\subsubsection{Generating Orderable Sub-Bicliques}

We now give an algorithm for implementing the subroutine
$\orderablesub{\mbiclique}{j}$.
Naively we can simply generate all sub-bicliques $\mbiclique'$ of $\mbiclique$ and check
$\orderable{\mbiclique'}{j}$.
However, to avoid the potentially high cost of such an iteration, we first determine which selectors and properties contribute to
$\propord{\mbiclique}{j}$.
Removing nodes outside of these sets will not affect the orderability, hence we do not need to try removing them.
Then we first attempt only removing one node from this set, computing all sub-bicliques that have one fewer element and are orderable.
Then, for all nodes for which this fails, we attempt to remove two nodes, and so on.
Note, if removing node $w$ renders $\mbiclique$ orderable, we do not need to test any bicliques obtained by removing $w$ and some other node $w'$, since this will not result in a maximal biclique.

Hence, we define the sets of candidate selectors and properties that may be removed to restore orderability.
These are all selectors and nodes that contribute to
$\propord{\mbiclique}{j}$.
That is
\[
    \candnodes
    =
    \setcomp{
        s_1, s_2, p_1, p_2
    }{
        \exists (s_1, p_1), (s_2, p_2) \in \edgeslast{\mbiclique}{j}\ .\ %
            (s_1, p_1) \edgeOrder^\ast (s_2, p_2)
    } \ .
\]
We define
$\orderablesub{\mbiclique}{j} =
 \orderablesubcand{\mbiclique}{j}{\candnodes}$
where
$\orderablesubcand{\mbiclique}{j}{\candnodes}$
generates a set
$\orderablebicliques$
of orderable sub-bicliques and is defined below.
When
$\mbiclique = (\selsBucket, \propsSet)$
we will abuse notation and write
$\mbiclique \setminus \set{w}$
for
$(\selsBucket \setminus \set{w}, \propsSet)$
when $w$ is a selector, and
$(\selsBucket, \propsSet \setminus \set{w})$
when $w$ is a property.
When defining the algorithm, we will collect all orderable bicliques in a set $\orderablebicliques$.
We will further collect in $\candnodes'$ the set of all nodes which fail to create an orderable biclique when removed by themselves.
We define
$\orderablesubcand{\mbiclique}{j}{\candnodes}$
recursively, where the recursive call attempts the removal of an increasing number of nodes.
It is
\begin{algorithmic}
    \State $\orderablebicliques := \emptyset$

    \State $\candnodes' := \emptyset$

    \ForAll {$w \in \candnodes$}
        \State $\mbiclique' := \mbiclique \setminus \set{w}$
        \If {$\orderable{\mbiclique'}{j}$}
            \State $\orderablebicliques := \orderablebicliques
                                           \cup
                                           \set{\mbiclique'}$
        \Else
            \State $\candnodes' := \candnodes' \cup \set{w}$
        \EndIf
    \EndFor

    \ForAll {$w \in \candnodes'$}
        \State $\orderablebicliques :=
                \orderablebicliques
                \cup
                \orderablesubcand{\mbiclique \setminus \set{w}}
                                 {j}
                                 {\candnodes' \setminus \set{w}}$
    \EndFor

    \State \Return $\orderablebicliques$
\end{algorithmic}

\subsection{Correctness of the Encoding}
\label{sec:encoding-correct}

We argue Proposition~\ref{prop:encoding-correct} which claims that the encoding
$(\hardcons, \softcons)$
is correct.
To prove this we need to establish three facts.
\begin{enumerate}
\item
    If $(\bucket, j)$ is a valid merging opportunity of $\covering$ and $\bucket = (\biclique, \propOrder)$ , then $(\biclique, j)$ is a solution to the hard constraints.

\item
    If $(\bucket, j)$ is generated from a solution to the hard constraints, it
        is a valid merging opportunity.

\item
    The weight of a solution generating $(\bucket, j)$ is the size of
    $\Trim{\inSeq{\covering}{j}{\bucket}}$.
\end{enumerate}

We argue these properties below.
\begin{enumerate}
\item
    Take a valid merging opportunity
    $(\bucket, j)$
    and let
    $\bucket = (\biclique, \propOrder)$.
    We construct a solution to the hard constraints.
    First, we assign $\inpos = j$.
    Next, since the merging opportunity is valid, we know $\biclique$ contains only edges in $\CSSedges$.
    That is, it is contained within a maximal biclique.
    Furthermore, since
    $\inSeq{\covering}{j}{\bucket}$
    is valid, we know that $\biclique$ is orderable at position $j$.
    Thus, it is a sub biclique of some
    $\mbiclique_i$
    in
    $(\{\mbiclique_i\}_{i=1}^{\nummbicliques}, \forbiddenfst)$,
    and, moreover, it is not the case that
    $\mbiclique_i \in \ap{\forbiddenfst}{j'}$
    for some $j' \leq j$.
    Thus, we assign
    $\bicliquevble = i$
    and we know that
    \[
        \bigwedge\limits_{1 \leq j \leq m + 1}
        \brac{
            \brac{\inpos >= j}
            \Rightarrow
            \bigwedge\limits_{\mbiclique_i \in \ap{\forbiddenfst}{j}}
                \brac{\bicliquevble \neq  i}
        }
    \]
    is satisfied.

    Additionally, for all $w$ appearing in $\biclique$ but not in $\mbiclique_i$, we set
    $\excludevble{\ap{\selpropord{i}}{w}}$ to false, otherwise we set it to true.
    Thus
    $\hasedge{(s, p)}$
    holds only if $(s, p)$ is an edge in $\biclique$.

    Next, we argue
    \[
        \brac{
            \bigwedge\limits_{(s_1, p_1) \edgeOrder (s_2, p_2)}
            \begin{array}{l}
                \hasedge{(s_1, p_1)}\ \Rightarrow \\
                \qquad
                \brac{
                    \inpos \leq \Index((s_2, p_2))
                    \lor
                    \hasedge{(s_2, p_2)}
                }
                \end{array}
        }
    \]
    is satisfied.
    This follows from
    $\inSeq{\covering}{j}{\bucket}$
    being valid.
    To see this, take some
    $(s_1, p_1) \edgeOrder (s_2, p_2)$.
    If $(s_1, p_1)$ does not appear in $\biclique$, then there is nothing to prove.
    If it does, we know $(s_2, p_2)$ must appear later in the file.
    There are two cases.
    If
    $j \leq \Index((s_2, p_2))$
    then the clause is satisfied.
    Otherwise we must have
    $(s_2, p_2)$
    in $\biclique$ or edge order would be violated.
    Thus the clause also holds in this case.

\item
    We need to prove that if the hard constraints are satisfied, then then
    generated merging opportunity
    $(\bucket, j)$
    is valid.
    Let
    $\bucket = (\biclique, \propOrder)$.
    For
    $\inSeq{\covering}{j}{\bucket}$
    to be valid, we first have to show that $\biclique$ introduces no new edges to the stylesheet.
    This is immediate since $\biclique$ is a sub biclique of some
    $\mbiclique_i$
    in
    $(\{\mbiclique_i\}_{i=1}^{\nummbicliques}, \forbiddenfst)$,
    which can only contain edges in $\CSSedges$.

    Next, we need to argue that we can create the ordering $\propOrder$ for the properties in $\biclique$.
    First note that $\mbiclique_i$ is orderable at position $j$.
    In particular, for any
    $(s_1, p_1) \edgeOrder^\ast (s_2, p_2)$
    with
    $(s_1, p_1)$
    and
    $(s_2, p_2)$
    appearing in $\mbiclique_i$, we have
    $p_1 \propord{\mbiclique_i}{j} p_2$.
    Since all edges in $\biclique$ also appear in $\mbiclique_i$, the existence of an ordering is immediate.

    Finally, we need to argue that
    $\inSeq{\covering}{j}{\bucket}$
    respects the edge order.
    Suppose
    $(s_1, p_2) \edgeOrder (s_2, p_2)$.
    To violate this ordering, we need to introduce a copy of $(s_1, p_1)$ after the last copy of $(s_2, p_2)$.
    Thus, we must have $(s_1, p_1)$ in $\biclique$.
    However, from
    \[
        \brac{
            \bigwedge\limits_{(s_1, p_1) \edgeOrder (s_2, p_2)}
            \begin{array}{l}
                \hasedge{(s_1, p_1)}\ \Rightarrow \\
                \qquad
                \brac{
                    \inpos \leq \Index((s_2, p_2))
                    \lor
                    \hasedge{(s_2, p_2)}
                }
                \end{array}
        }
    \]
    we are left with two cases.
    In the first
    $j \leq \Index((s_2, p_2))$
    and the edge order is maintained.
    In the second, we also have
    $(s_2, p_2)$
    in
    $\biclique$.
    However, the edge order is maintained because $\biclique$ is orderable.
    Thus we are done.

\item
    Finally, we argue that the weight of a satisfying assignment accurately reflects the size of
    $\Trim{\inSeq{\covering}{j}{\bucket}}$.
    This is fairly straightforward.
    The size of
    $\Trim{\inSeq{\covering}{j}{\bucket}}$.
    comprises two parts:
    the size of $\bucket$, and the size of $\covering$ after the trim operation.
    It is immediate to see that the size of $\bucket$ is equal to the size of all of its nodes.
    In particular, this is the size of all nodes of
    $\mbiclique_i$
    that appear in $\bucket$.
    That is, have not been excluded.
    Thus the clause with weight $\nodeWeight(w)$
    \[
        \brac{\bicliquevble = i} \Rightarrow \excludevble{\ap{\selpropord{i}}{w}} \ .
    \]
    for each $w$ appearing in $\mbiclique_i$ accurately computes the size of $\bucket$.

    For the size of $\covering$ after the trim operation, we first use the
    assumption that $\covering$ has already been trimmed before applying the
    merging opportunity. Thus, any further nodes removed in
    $\Trim{\inSeq{\covering}{j}{\bucket}}$
    from a rule
    $\bucket_{i'}$
    must be removed because some edge $e$ in $\bucket$ also appears in $\bucket_{i'}$
    and, moreover, it was the case
    $i' = \Index(e)$ and $i' \leq j$.
    In particular, we can only remove a node $w$ from $\bucket_{i'}$ if all edges $e$ incident to $w$ with
    $i' = \Index(e)$
    have $e$ appearing in $\bucket$
    (else there will still be some edge preventing $w$ from being trimmed after
    applying the merging opportunity).
    Thus, for each selector node $s$, we know it is not removed if the clause with weigth
    $\nodeWeight(s)$
    \[
        i \leq \inpos \land
        \bigwedge\limits_{\substack{
            \Index((s, p)) = i \\
            p \in \propsBucket
        }}
            \hasedge{(s, p)}
    \]
    is not satisfied.
    Similarly for property nodes $p$.
    Thus, these clauses accurately count the size of the covering after trimming.
\end{enumerate}

%% file: edgeOrder-appendix.tex
\def\refsecedgeOrder{\ref{sec:edgeOrder}}

\section{Additional Material for Section~\protect\refsecedgeOrder}
\label{sec:edge-order-appendix}

\input{css3-namespace}

\input{hardness-proof}

\input{important}

%% file: css3-namespace.tex
\subsection{Handling Pseudo-Elements}

CSS selectors can also finish with a \emph{pseudo-element}.
For example
$\css\psbefore$.
These match nodes that are not formally part of a document tree.
In the case of $\css\psbefore$ the selector matches a phantom node appearing before the node matched by $\css$.
These can be used to insert content into the tree for stylistic purposes.
For example
\begin{center}
\begin{minted}{css}
    .a::before { content:">" }
\end{minted}
\end{center}
places a ``>'' symbol before the rendering of any node with class \texttt{a}.

We divide CSS selectors into five different types depending on the
pseudo-element appearing at the end of the selector.
We are interested here in the nodes matched by a selector.
The pseudo-elements
$\psfirstline$, $\psfirstletter$, $\psbefore$, and $\psafter$
essentially match nodes inserted into the DOM tree.
The CCS3 specification outlines how these nodes should be created.
For our purposes we only need to know that
    the five syntactic cases in the above grammar can never match the same
    inserted node, and
    the selectors $\psfirstletter$ and $\psfirstline$ require that the node matched by $\css$ is not empty.

Since we are interested here in the non-emptiness and non-emptiness-of-intersection problems,
we will omit pseudo-elements in the remainder of this article, under the assumptions that
\begin{compactitem}
\item
    selectors of the form $\css\psfirstline$ or $\css\psfirstletter$ are replaced by a selector
    $\css\cssneg{\psempty}$, and
\item
    selectors of the form $\css$, $\css\psbefore$, or $\css\psafter$ are replaced by $\css$, and
\item
    we \emph{never} take the intersection of two selectors $\css$ and $\css'$ such that it's not the case that either
    \begin{compactitem}
    \item
        $\css$ and $\css'$ were derived from selectors containing no pseudo-elements, or
    \item
        $\css$ and $\css'$ were derived from selectors ending with the same pseudo-element.
    \end{compactitem}
\end{compactitem}
In this way, we can test non-emptiness of a selector by testing its replacement.
For non-emptiness-of-intersection, we know if two selectors end with different pseudo-elements (or one does not contain a pseudo-element, and one does), their intersection is necessarily empty.
Thus, to check non-emptiness-of-intersection, we immediately return ``empty'' for any two selectors ending with different pseudo-elements.
To check two selectors ending with the same pseudo-element, the problem reduces to testing the intersection of their replacements.

%% file: hardness-proof.tex
\subsection{NP-hardness of Theorem \ref{thm:emptiness}}
\label{sec:np-hardness-proof}

\begin{lemma} \label{lem:emptiness-hard}
    Given a CSS selector $\css$, deciding
    $
        \exists \tree, \node \ .\ \tree, \node \models \css
    $
    is NP-hard.
\end{lemma}
\begin{proof}
    We give a polynomial-time reduction from the NP-complete problem of
    non-universality of unions of arithmetic progressions~\cite[Proof of
    Theorem 6.1]{SM73}. To define this, we first fix some notation.
    Given a pair
    $\tup{\coefa, \coefb} \in \N \times \N$,
    we define $\sem{\tup{\coefa,\coefb}}$ to be the set of natural numbers of the
    form
    $\coefa \numof + \coefb$ for $\numof \in \N$. That is,
    $\sem{\tup{\coefa,\coefb}}$ represents an arithmetic progression, where
    $\coefa$ represents the \defn{period} and $\coefb$ represents the
    \defn{offset}.
    Let
    $E \subseteq \N \times \N$
    be a finite subset of pairs $\tup{\coefa, \coefb}$.
    We define
    $\sem{E} = \bigcup_{\tup{\coefa,\coefb}\in E} \sem{\tup{\coefa,\coefb}}$.
    The NP-complete problem is: given $E$ (where numbers may be represented
    in unary or in binary representation), is $\sem{E} \neq \N$?
    Observe that this problem is equivalent to checking whether
    $\sem{E+1} \neq \N_{>0}$
    where $E+1$ is defined by adding 1 to the offset
    $\coefb$ of each arithmetic progression $\tup{\coefa,\coefb}$ in $E$.
    By complementation, this last problem is equivalent to checking whether
    $\N_{>0} \setminus \sem{E+1} \neq \emptyset$.
    Since $\N_{>0} \setminus \sem{E+1} =
    \bigcap_{\tup{\coefa,\coefb} \in E} \overline{\sem{\coefa,\coefb+1}}$,
    the problem can be seen to be equivalent to testing the non-emptiness of
    \[
        \isany\setcomp{\cssneg{\psnthchild{\coefa}{(\coefb+1)}}}
                      {\tup{\coefa, \coefb} \in E} \ .
    \]
    Thus, non-emptiness is NP-hard.
\end{proof}

\OMIT{
\begin{lemma}
    Given two CSS selectors $\css_1$ and $\css_2$, deciding whether
    $
        \exists \tree, \node \ .\ %
            \tree, \node \models \css_1 \land \tree, \node \models \css_2
    $
    is NP-hard.
\end{lemma}
\begin{proof}
    Non-emptiness reduces to non-emptiness-of-intersection.
    That is, $\css$ is not empty iff the intersection of $\css$ and $\isany$ is non-empty.
    From Lemma~\ref{lem:emptiness-hard} we get the result.
\end{proof}
}

%% file: important.tex
\subsection{Handling \texttt{!important} and shorthand property names}
\label{sec:important}

Our approach handles the \text{!important} keyword and shorthand property names.
In this section we explain the steps we take to account for them.

\subsubsection{The \text{!important} Keyword}

First, the keyword \verb+!important+
in property declaration as is used in the rule
\begin{minted}{css}
    div { color:red !important }
\end{minted}
can be
used to override the cascading behaviour of CSS, e.g., in our example, if
a node is matched by \verb+div+, as well as a later rule $R$ that assigns a
different color, then assign \verb+red+ to color (unless $R$ also has
the keyword \verb+!important+ next to its color property declaration).
To handle this, we can extend the notion of specificity of a selector
to the notion of specificity of a pair $(s,p)$ of selector and property
declaration, after which we may proceed as before (i.e. relating only two
edges with the same specificity).
Recall from \cite{CSS3sel} that the specificity of a selector is a
3-tuple $(a,b,c) \in \N \times \N \times \N$ where $a$, $b$, and $c$ can
be obtained by calculating the sum of the number of IDs, classes, tag
names, etc. in the selector. Since the lexicographic order is used to
rank the elements of $S := \N \times \N \times \N$, the specificity of
a pair $(s,p)$ can now be defined to be $(i,a,b,c)$, where $(a,b,c)$
is the specificity of $s$, and $i=1$ if \verb+!important+ can be found
in $p$ (otherwise, $i=0$). In particular, this also handles the case
where multiple occurrences of \verb+!important+ is found in the CSS file.

\subsubsection{Shorthand Property Names}

\emph{Shorthand property names} \cite{cssref} can be used to simultaneously set the values of related property names.
For example, \texttt{border: 5px solid red} is equivalent to
\begin{verbatim}
    border-width: 5px; border-style: solid; border-color:red
\end{verbatim}
In particular, this implies that $(s,p)$ and $(s,p')$ can be related in $\edgeOrder$ if $p$ defines \texttt{border}, while the other property $p'$ defines
\texttt{border-width}. One way to achieve this is to simply list all
pairs of comparable property names, which can be done since only
around 100 property names are currently officially related. [Incidentally, a close enough
approximation is that one property name is a \emph{prefix} of the other
property name (e.g., \texttt{border} is a prefix of \texttt{border-style}),
but this is not complete (e.g. \texttt{font} can be used to define
\texttt{line-height})]

%% file: intersection-appendix.tex
\def\refsecintersection{\ref{sec:intersection}}
\section{Additional Material for Section \protect\refsecintersection}
\label{sec:intersection-appendix}

\input{selector-to-automata-proof}

\input{intersection-proof}

\input{emptiness-proofs}

%% file: selector-to-automata-proof.tex
\def\refpropcsstoaut{\ref{prop:css-to-aut}}
\subsection{Correctness of $\selaut{\css}$ in Proposition
\protect\refpropcsstoaut}
\label{sec:css-automata-proof}

We show both soundness and completeness of $\selaut{\css}$.

\begin{lemma} \label{lem:aut-sound}
    For each CSS selector $\css$ and tree $\tree$, we have
    \[
        \aaccepts{\tree}{\node}{\selaut{\css}}
        \Rightarrow
        \tree, \node \models \css \ .
    \]
\end{lemma}
\begin{proof}
    Suppose
    $\aaccepts{\tree}{\node}{\selaut{\css}}$.
    By construction of $\selaut{\css}$ we know that the accepting run must pass through all states
    $\selstate{1}, \ldots, \selstate{\numof}$
    where
    $\css = \csssim_1\ \genop_1\ \cdots\ \genop_{\numof-1}\ \csssim_\numof$.
    Notice, in order to exit each state $\selstate{\idxi}$ a transition labelled by $\csssim_\idxi$ must be taken.
    Let $\node_\idxi$ be the node read by this transition, which necessarily satisfies $\csssim_\idxi$.
    Observe
    $\node_\numof = \node$.
    We proceed by induction.
    We have $\node_1$ satisfies $\csssim_1$.
    Hence, assume $\node_{\idxi}$ satisfies
    $\csssim_1\ \genop_1\ \cdots\ \genop_{\idxi-1}\ \csssim_{\idxi}$.
    We show $\node_{\idxi+1}$ satisfies
    $\csssim_1\ \genop_1\ \cdots\ \genop_{\idxi}\ \csssim_{\idxi+1}$.

    We case split on $\genop_{\idxi}$.
    \begin{compactitem}
    \item
        When $\genop_{\idxi} = \cssdescendant$ we need to show $\node_{\idxi+1}$ is a descendant of $\node_\idxi$.
        By construction of $\selaut{\css}$ the run reaches $\node_{\idxi+1}$ in one of two ways.
        If it is via a single transition
        $\selstate{\idxi} \atran{\arrchild}{\csssim_\idxi} \selstate{{\idxi+1}}$
        then $\node_{\idxi+1}$ is immediately a descendant of $\node_\idxi$.
        Otherwise the first transition is
        $\selstate{\idxi} \atran{\arrchild}{\csssim_\idxi} \midstate{\idxi}$.
        The reached node is necessarily a descendant of $\node_\idxi$.
        To reach $\node_{\idxi+1}$ a path is followed applying $\arrsibling$ and $\arrchild$ arbitrarily, which cannot reach a node that is not a descendant of $\node_\idxi$.
        Finally, the transition to $\node_{\idxi+1}$ is via $\arrneighbour$ or $\arrchild$ and hence $\node_{\idxi+1}$ must also be a descendant of $\node_\idxi$.

    \item
        When $\genop_{\idxi} = \csschild$ we need to show $\node_{\idxi+1}$ is a descendant of $\node_\idxi$.
        By construction of $\selaut{\css}$ the run reaches $\node_{\idxi+1}$ in one of two ways.
        If it is via a single transition
        $\selstate{\idxi} \atran{\arrchild}{\csssim_\idxi} \selstate{{\idxi+1}}$
        then $\node_{\idxi+1}$ is immediately a child of $\node_\idxi$.
        Otherwise the first transition is
        $\selstate{\idxi} \atran{\arrchild}{\csssim_\idxi} \midstate{\idxi}$.
        The reached node is necessarily a child of $\node_\idxi$.
        To reach $\node_{\idxi+1}$ only transitions labelled $\arrsibling$ and $\arrneighbour$ can be followed.
        Hence, the node reached must also be a child of $\node_\idxi$.

    \item
        When $\genop_{\idxi} = \cssneighbour$ we need to show $\node_{\idxi+1}$ is the next neighbour of $\node_\idxi$.
        Since the only path is a single transition labelled $\arrneighbour$ the result is immediate.

    \item
        When $\genop_{\idxi} = \csssibling$ we need to show $\node_{\idxi+1}$ is a sibling of $\node_\idxi$.
        By construction of $\selaut{\css}$ the run reaches $\node_{\idxi+1}$ in one of two ways.
        If it is via a single transition
        $\selstate{\idxi} \atran{\arrneighbour}{\csssim_\idxi} \selstate{{\idxi+1}}$
        then $\node_{\idxi+1}$ is immediately a sibling of $\node_\idxi$.
        Otherwise the first transition is
        $\selstate{\idxi} \atran{\arrneighbour}{\csssim_\idxi} \midstate{\idxi}$.
        The reached node is necessarily a sibling of $\node_\idxi$.
        To reach $\node_{\idxi+1}$ only transitions labelled $\arrsibling$ and $\arrneighbour$ can be followed.
        Hence, the node reached must also be a sibling of $\node_\idxi$.
    \end{compactitem}
    Thus, by induction, $\node_\numof = \node$ satisfies
    $\csssim_1\ \genop_1\ \cdots\ \genop_{\numof-1}\ \csssim_\numof = \css$.
\end{proof}

\begin{lemma} \label{lem:aut-comp}
    For each CSS selector $\css$ and tree $\tree$, we have
    \[
        \tree, \node \models \css \
        \Rightarrow
        \aaccepts{\tree}{\node}{\selaut{\css}}
    \]
\end{lemma}
\begin{proof}
    Assume
    $\tree, \node \models \css$.
    Thus, since
    $\css = \csssim_1\ \genop_1\ \cdots\ \genop_{\numof-1}\ \csssim_\numof$,
    we have a sequence of nodes
    $\node_1, \ldots, \node_\numof$
    such that for each $\idxi$ we have
    $\tree, \node_\idxi
     \models
     \csssim_1\ \genop_1\ \cdots\ \genop_{\idxi-1} \ \csssim_\idxi$.
    Note $\node_\numof = \node$.
    We build a run of $\selaut{\css}$ from $\csssim_1$ to $\selstate{\idxi}$ by induction.
    When $\idxi = 1$ we have the run constructed by taking the loops on the initial state $\selstate{1}$ labelled $\arrchild$ and $\arrsibling$ to navigate to $\node_1$.
    Assume we have a run to $\selstate{\idxi}$.
    We build a run to $\selstate{{\idxi+1}}$ we consider $\genop_\idxi$.
    \begin{compactitem}
    \item
        When $\genop_\idxi = \cssdescendant$ we know $\node_{\idxi+1}$ is a descendant of $\node_\idxi$.
        We consider the construction of $\selaut{\css}$.
        If $\node_{\idxi+1}$ is the first child of $\node_\idxi$, we construct the run to $\selstate{{\idxi+1}}$ via the transition
        $\selstate{\idxi} \atran{\arrchild}{\csssim_\idxi} \selstate{{\idxi+1}}$,
        noting that we know $\node_\idxi$ satisfies $\csssim_\idxi$.
        Otherwise we take
        $\selstate{\idxi} \atran{\arrchild}{\csssim_\idxi} \midstate{\idxi}$
        and arrive at either an ancestor or sibling of $\node_{\idxi+1}$.
        In the case of a neighbour, we can take the transition labelled $\arrneighbour$ to reach $\node_{\idxi+1}$.
        For an indirect sibling we can take the transition labelled $\arrsibling$ followed by the transition labelled $\arrneighbour$.
        For an ancestor, we take the transition labelled $\arrchild$ and arrive at another sibling or ancestor of $\node_{\idxi+1}$ that is closer.
        We continue in this way until we reach $\node_{\idxi+1}$ as needed.

    \item
        When $\genop_\idxi = \csschild$ we know $\node_{\idxi+1}$ is a child of $\node_\idxi$.
        We consider the construction of $\selaut{\css}$.
        If $\node_{\idxi+1}$ is the first child of $\node_\idxi$, we construct the run to $\selstate{{\idxi+1}}$ via the transition
        $\selstate{\idxi} \atran{\arrchild}{\csssim_\idxi} \selstate{{\idxi+1}}$,
        noting that we know $\node_\idxi$ satisfies $\csssim_\idxi$.
        Otherwise we take
        $\selstate{\idxi} \atran{\arrchild}{\csssim_\idxi} \midstate{\idxi}$
        and arrive at a preceding sibling of $\node_{\idxi+1}$.
        We can take the transition labelled $\arrsibling$ to reach the preceding neighbour of $\node_{\idxi+1}$ if required, and then the transition labelled $\arrneighbour$ to reach $\node_{\idxi+1}$ as required.

    \item
        When $\genop_\idxi = \cssneighbour$ we know $\node_{\idxi+1}$ is the neighbour of $\node_\idxi$.
        We consider the construction of $\selaut{\css}$ and take the only available transition
        $\selstate{\idxi} \atran{\arrneighbour}{\csssim_\idxi} \selstate{{\idxi+1}}$,
        noting that we know $\node_\idxi$ satisfies $\csssim_\idxi$.
        Thus, we reach $\node_{\idxi+1}$ as required.

    \item
        When $\genop_\idxi = \csssibling$ we know $\node_{\idxi+1}$ is a sibling of $\node_\idxi$.
        We consider the construction of $\selaut{\css}$.
        If $\node_{\idxi+1}$ is the neighbour of $\node_\idxi$, we construct the run to $\selstate{{\idxi+1}}$ via the transition
        $\selstate{\idxi} \atran{\arrchild}{\csssim_\idxi} \selstate{{\idxi+1}}$,
        noting that we know $\node_\idxi$ satisfies $\csssim_\idxi$.
        Otherwise we take
        $\selstate{\idxi} \atran{\arrneighbour}{\csssim_\idxi} \midstate{\idxi}$
        and arrive at a preceding sibling of $\node_{\idxi+1}$.
        We can take the transition labelled $\arrsibling$ to reach the preceding neighbour of $\node_{\idxi+1}$ is required, and then the transition labelled $\arrneighbour$ to reach $\node_{\idxi+1}$ as required.
    \end{compactitem}
    Thus, by induction, we construct a run to $\node_\numof$ ending in state $\selstate{\numof}$.
    We transform this to an accepting run by taking the transition
    $\selstate{\numof} \atran{\arrlast}{\csssim_\numof} \afinstate$,
    using the fact that $\node_\numof$ satisfies $\csssim_\numof$.
\end{proof}

%% file: intersection-proof.tex
\def\refpropautintersection{\ref{prop:aut-intersection}}
\subsection{Proof of Proposition  \protect\refpropautintersection}
\label{sec:intersection-proof}
    We show that
    \[
        \aaccepts{\tree}{\node}{\cssaut_1} \land
        \aaccepts{\tree}{\node}{\cssaut_2}
        \iff
        \aaccepts{\tree}{\node}{\cssaut_1 \cap \cssaut_2} \ .
    \]

    We begin by observing that that all runs of a CSS automaton showing acceptance of a node $\node$ in $\tree$ must follow a sequence of nodes
    $\node_1, \ldots, \node_\numof$
    such that
    \begin{compactitem}
    \item
        $\node_1$ is the root of $\tree$, and
    \item
        when
        $\node_\idxj = \node'\treedir$
        then either
        $\node_{\idxj+1} = \node'(\treedir+1)$
        or
        $\node_{\idxj+1} = \node_\idxj 1$ for all $\idxj$, and
    \item
        $\node_\numof = \node$
    \end{compactitem}
    that defines the path taken by the automaton.
    Each node is ``read'' by some transition on each run.
    Note a transition labelled $\arrsibling$ may read sequence nodes that is a factor of the path above.
    However, since these transitions are loops that do not check the nodes, without loss of generality we can assume each $\arrsibling$ in fact reads only a single node.
    That is, $\arrsibling$ behaves like $\arrneighbour$.
    Recall, $\arrsibling$ was only introduced to ensure the existence of ``short'' runs.

    Because of the above, any two runs accepting $\node$ in $\tree$ must follow the same sequence of nodes and be of the same length.

    We have
    $\aaccepts{\tree}{\node}{\cssaut_1} \land
     \aaccepts{\tree}{\node}{\cssaut_2}$
    iff there are accepting runs
    \[
        \astate^\idxi_1 \atran{\arrgen^\idxi_1}{\csssim^\idxi_1}
        \cdots
        \atran{\arrgen^\idxi_{\numof}}{\csssim^\idxi_{\numof}}
        \astate^\idxi_{\numof+1}
    \]
    of $\cssaut_\idxi$ over $\tree$ reaching node $\node$ for both
    $\idxi \in \set{1,2}$.
    We argue these two runs exist iff we have a run
    \[
        \tup{\astate^1_1, \astate^2_1}
        \atran{\arrgen_1}{\csssim_1}
        \cdots
        \atran{\arrgen_\numof}{\csssim_\numof}
        \tup{\astate^1_{\numof+1}, \astate^2_{\numof+1}}
    \]
    of $\cssaut_1 \cap \cssaut_2$
    where each $\arrgen_\idxj$ and $\csssim_\idxj$ depends on
    $\tup{\arrgen^1_\idxj, \arrgen^2_\idxj}$.
    \begin{compactitem}
    \item
        When $\tup{\arrchild, \arrchild}$ we have
        $\arrgen_\idxj = \arrchild$
        and
        $\csssim_\idxj = \csssim^1_\idxj \cap \csssim^2_\idxj$.

    \item
        When $\tup{\arrneighbour, \arrneighbour}$ we have
        $\arrgen_\idxj = \arrneighbour$
        and
        $\csssim_\idxj = \csssim^1_\idxj \cap \csssim^2_\idxj$.

    \item
        When $\tup{\arrneighbour, \arrsibling}$ we have
        $\arrgen_\idxj = \arrneighbour$
        and
        $\csssim_\idxj = \csssim^1_\idxj$.

    \item
        When $\tup{\arrsibling, \arrneighbour}$ we have
        $\arrgen_\idxj = \arrneighbour$
        and
        $\csssim_\idxj = \csssim^2_\idxj$.

    \item
        When $\tup{\arrsibling, \arrsibling}$ we have
        $\arrgen_\idxj = \arrsibling$
        and
        $\csssim_\idxj = \isany$.

    \item
        When $\tup{\arrlast, \arrlast}$ we have
        $\arrgen_\idxj = \arrlast$
        and
        $\csssim_\idxj = \csssim^1_\idxj \cap \csssim^2_\idxj$.

    \item
        The cases
        $\tup{\arrchild, \arrsibling}$,
        $\tup{\arrchild, \arrneighbour}$,
        $\tup{\arrchild, \arrlast}$,
        $\tup{\arrneighbour, \arrchild}$,
        $\tup{\arrneighbour, \arrlast}$,
        $\tup{\arrsibling, \arrchild}$,
        $\tup{\arrsibling, \arrlast}$,
        $\tup{\arrlast, \arrchild}$,
        $\tup{\arrlast, \arrneighbour}$, and
        $\tup{\arrlast, \arrsibling}$
        are not possible.
    \end{compactitem}
    The existence of the transitions comes from the definition of $\cssaut_1 \cap \cssaut_2$.
    We have to argue that $\node_\idxj$ satisfies both
    $\csssim^\idxi_\idxj$
    iff it also satisfies
    $\csssim_\idxj$.
    By observing
    $\csssim \cap \isany = \isany \cap \csssim = \csssim$
    we always have
    $\csssim_\idxj = \csssim^1_\idxj \cap \csssim^2_\idxj$.

    Let $\csssim^\idxi_\idxj = \csstype_\idxi\cssconds_\idxi$
    and
    $\csssim_\idxj = \csstype\cssconds$.
    It is immediate that $\node_\idxj$ satisfies
    $\cssconds = \cssconds_1 \cup \cssconds_2$
    iff it satisfies both $\cssconds_\idxi$.

    To complete the proof we need to show $\node_\idxj$ satisfies $\csstype$ iff it satisfies both $\csstype_\idxi$.
    Note, we must have some $\ns$ and $\ele$ such that
    $\csstype, \csstype_1, \csstype_2
     \in
     \set{\isany,\isanyns{\ns},\iselens{\ns}{\ele},\isele{\ele}}$
    else the type selectors cannot be satisfied (either
    $\csstype = \cssneg{\isany}$
    or
    $\csstype_1$ and $\csstype_2$ assert conflicting namespaces or elements).

    If some $\csstype_\idxi = \isany$ the property follows by definition.
    Otherwise, if $\csstype = \csstype_2$ then in all cases the conjunction of $\csstype_1$ and $\csstype_2$ is equivalent to $\csstype_2$ and we are done.
    The situation is similar when $\csstype = \csstype_1$.
    Otherwise
    $\csstype = \iselens{\ns}{\ele}$
    and
    $\csstype_1 = \isanyns{\ns}$
    and
    $\csstype_2 = \isele{\ele}$
    or vice versa, and it is easy to see $\csstype$ is equivalent to the intersection of $\csstype_1$ and $\csstype_2$.
    Thus, we are done.

%% file: emptiness-proofs.tex
\subsection{Proofs for Non-Emptiness of CSS Automata}
\label{sec:emptiness-proofs}

\subsubsection{Bounding Namespaces and Elements}
\label{sec:bounded-types-proof}

We show \refproposition{prop:boundedtypes}.
We need to define the finite sets $\finof{\eles}$ and $\finof{\ns}$.
To this end, we write
\begin{compactenum}
\item
    $\elesof{\cssaut}$
    to denote the set of namespaced elements
    $\qele{\ns}{\ele}$
    such that there is some transition
    $\astate \atran{\arrgen}{\csssim} \astate' \in \atrans$
    with
    $\csssim = \iselens{\ns}{\ele} \cssconds$
    for some $\ns$, $\ele$, and $\cssconds$,

\item
    $\selsof{\cssaut}$ is the set of transitions
    $\astate \atran{\arrgen}{\csssim} \astate' \in \atrans$
    with
    $\csssim \neq \isany$
    and $\numsels{\cssaut}$ denotes the cardinality of
    $\selsof{\cssaut}$.
\end{compactenum}

Let
$\set{\fakeele_1,
      \ldots,
      \fakeele_{\numsels{\cssaut}}}$
be a set of fresh namespaced elements and
\[
    \finelesof{\cssaut} = \elesof{\cssaut} \uplus
                          \set{\fakeele_1,
                               \ldots,
                               \fakeele_{\numsels{\cssaut}}} \uplus
                          \set{\nullele}
\]
where there is a bijection
$
    \elemap : \selsof{\cssaut}
               \rightarrow
               \set{\fakeele_1, \ldots, \fakeele_{\numsels{\cssaut}}}
$
such that for each
$\atrant \in \selsof{\cssaut}$
we have
$\ap{\elemap}{\atrant} = \fakeele$
and
\begin{compactenum}
\item
    $\fakeele = \qele{\ns}{\ele}$
    if $\csssim$ can only match elements $\qele{\ns}{\ele}$,
\item
    $\fakeele = \qele{\ns}{\ele}$
    for some fresh element $\ele$ if $\csssim$ can only match elements of the form
    $\qele{\ns}{\ele'}$
    for all elements $\ele'$, and
\item
    $\fakeele = \qele{\ns}{\ele}$
    for some fresh namespace $\ns$ if $\csssim$ can only match elements of the form
    $\qele{\ns'}{\ele}$
    for all namespaces $\ns'$, and
\item
    $\fakeele = \qele{\ns}{\ele}$
    for fresh $\ns$ and fresh $\ele$ if $\csssim$ places no restrictions on the element type.
\end{compactenum}

Thus, we can define bounded sets of namespaces and elements
\[
    \begin{array}{rcl} %
        \finof{\eles} %
        &=& %
        \setcomp{\ele} %
                {\exists \ns\ .\ %
                     \qele{\ns}{\ele} \in \finelesof{\cssaut}} %
        \\ %
        \finof{\nspaces} %
        &=& %
        \setcomp{\ns} %
                {\exists \ele\ .\ %
                     \qele{\ns}{\ele} \in \finelesof{\cssaut}} \ . %
    \end{array} %
\]

It remains to show $\finelesof{\cssaut}$ is sufficient.
That is, if some tree $\tree$ is accepted $\cssaut$, we can define another tree $\tree'$ that also is accepted by $\cssaut$ but only uses types in
$\finelesof{\cssaut}$.

We take
$\aaccepts{\tree}{\node}{\cssaut}$
with
$\tree = \tup{\treedom, \treelab}$
and we define
$\tree' = \tup{\treedom, \treelab'}$
satisfying the proposition.
Let
\[
    \astate_0, \node_0,
    \astate_1, \node_1,
    \ldots,
    \astate_\runlen, \node_\runlen,
    \astate_{\runlen+1}
\]
be the accepting run of $\cssaut$, by the sequence of transitions
$\atrant_0, \ldots, \atrant_\runlen$.
As noted above, we can assume each transition in $\atrans$ appears only once in this sequence.
Let
$\set{\csssim_1, \ldots, \csssim_{\numsels{\cssaut}}}$
be the set of selectors appearing in $\cssaut$.
We perform the following modifications to $\treelab$ to obtain $\treelab'$.

We obtain $\treelab'$ from $\treelab$ by changing the element labelling.
We first consider all
$0 \leq \idxi \leq \runlen$
such that $\node_\idxi$ is labelled by some element
$\qele{\ns}{\ele} \in \elesof{\cssaut}$.
Let
$\nodeset_{\qele{\ns}{\ele}}$
be the set of nodes labelled by
$\qele{\ns}{\ele}$
in $\treelab$.
In $\treelab'$ we label all nodes in
$\nodeset_{\qele{\ns}{\ele}}$
by $\qele{\ns}{\ele}$.
That is, we do not relabel nodes labelled by
$\qele{\ns}{\ele}$.
Let $\nodeset$ be the union of all such
$\nodeset_{\qele{\ns}{\ele}}$.

Next we consider all
$0 \leq \idxi \leq \runlen$
such that
$\node_\idxi \notin \nodeset$ (i.e. was not labelled in the previous case) and
$\atrant_\idxi = \astate_\idxi \atran{\arrgen}{\csssim} \astate_{\idxi+1}$
with
$\csssim \neq \isany$.
Let
$\qele{\ns}{\ele} \notin \elesof{\cssaut}$
be the element labelling of $\node_\idxi$ in $\treelab$.
Moreover, take $\fakeele$ such that $\ap{\elemap}{\atrant_\idxi} = \fakeele$.
In $\treelab'$ we label all nodes in
$\nodeset_{\qele{\ns}{\ele}}$
(i.e. labelled by $\qele{\ns}{\ele}$) in $\treelab$ by $\fakeele$.
That is, we globally replace $\qele{\ns}{\ele}$ by $\fakeele$.
Let $\nodeset'$ be $\nodeset$ union all such $\nodeset_\ele$.

Finally, we label all nodes not in $\nodeset'$ with the null element $\nullele$.

To see that
\[
    \astate_0, \node_0,
    \astate_1, \node_1,
    \ldots,
    \astate_\runlen, \node_\runlen,
    \astate_{\runlen+1}
\]
via
$\atrant_0, \ldots, \atrant_\runlen$
is an accepting run of
$\tup{\treedom, \treelab'}$
we only need to show that for each
$\atrant_\idxi = \astate_\idxi \atran{\arrgen}{\csssim} \astate_{\idxi+1}$
that $\node_\idxi$ satisfies $\csssim$.
This can be shown by induction over $\csssim$.
Most atomic cases are straightforward (e.g. the truth of $\pshover$ is not affected by our transformations).
The case of
$\isele{\ele}$,
$\isanyns{\ns}$, or
$\iselens{\ns}{\ele}$
appearing positively follows since in these cases the labelling remained the same or changed to some $\fakeele$ consistent with the selector.
When such selectors appear negatively, the result follows since we only changed elements and namespaces to fresh ones.
The truth of attribute selectors remains unchanged since we did not change the attribute labelling.
The cases of
$\psnthchild{\coefa}{\coefb}$
and
$\psnthlastchild{\coefa}{\coefb}$
follow since we did not change the number of nodes.
For the selectors
$\psnthoftype{\coefa}{\coefb}$
and
$\psnthlastoftype{\coefa}{\coefb}$
there are two cases.
If we did not change the element label $\qele{\ns}{\ele}$ of $\node_\idxi$, then we also did not change the label of its siblings.
Moreover, we did not add any $\qele{\ns}{\ele}$ labels elsewhere in the tree.
Hence the truth of the formulas remains the same.
If we did change the label from $\qele{\ns}{\ele}$ to $\fakeele$ for some $\fakeele$ then observe that we also relabelled all other nodes in the tree labelled by $\qele{\ns}{\ele}$.
In particular, all siblings of $\node_\idxi$.
Moreover, since $\elemap$ is a bijection and each transition appears only once in the run, we did not label any node not labelled $\qele{\ns}{\ele}$ with $\fakeele$.
Hence the truth of the formulas also remains the same.
Similar arguments hold for $\psonlychild$ and $\psonlyoftype$.

Thus, $\tup{\treedom, \treelab'}$ is accepted, and only uses elements in $\finelesof{\cssaut}$ as required.

\subsubsection{Proof of Polynomial Bound on Attribute Value Lengths}
\label{sec:poly-string-solution}

We prove \refproposition{prop:boundedatts}.
That is we argue the existence of a polynomial bound for
the solutions to any finite set $\consset$ of constraints of the
form $\opattns{\ns}{\att}{\attop}{\attval}$
or
$\cssneg{\opattns{\ns}{\att}{\attop}{\attval}}$,
for some fixed $\ns$ and $\att$.
We say that $\consset$ is a set of constraints over $\ns$ and $\att$.

In fact, the situation is a little more complicated because it may be the case that $\att$ is $\idatt$.
In this case we need to be able to enforce a global uniqueness constraint on the attribute values.
Thus, for constraints on an ID attribute, we need a bound that is large enough to allow to all constraints on the same ID appearing throughout the automaton to be satisfied by unique values.
Thus, for a given automaton, we might ask for a bound $\attvalbound$ such that \emph{if} there exists unique ID values for each transition, then there exist values of length bounded by $\attvalbound$.

However, the bound on the length must still work when we account for the fact that not all transitions in the automaton will be used during a run.
Consider the following illustrative example.
\begin{center}
    \begin{psmatrix}
        \\
        \rnode{N1}{$\ainitstate$} & &
        \rnode{N2}{$\astate$} &
        \rnode{N3}{$\afinstate$}
        \\
        \ncarc[arcangle=35]{->}{N1}{N2}\naput{$\arrneighbour$}\nbput{$\attisns{\ns}{\idatt}{\attval}$}
        \ncarc[arcangle=35]{<-}{N2}{N1}\nbput{$\arrchild$}\naput{$\attisns{\ns}{\idatt}{\attval}$}
        \ncline{->}{N2}{N3}\naput{$\arrlast$}\nbput{$\isany$}
    \end{psmatrix}
\end{center}
In this case we have two transitions with ID constraints, and hence two sets of constraints
$\consset_1 = \consset_2 = \set{\attisns{\ns}{\idatt}{\attval}}$.
Since these two sets of constraints cannot be satisfied simultaneously with unique values, even the bound
$\attvalbound = 0$
will satisfy our naive formulation of the required property (since the property had the existence of a solution as an antecedent).
However, it is easy to see that any run of the automaton does not use both sets of constraints, and that the bound
$\attvalbound = \sizeof{\attval}$
should suffice.
Hence, we formulate the property of our bound to hold for all \emph{sub-collections} of the collection of sets of constraints appearing in the automaton.

\begin{namedlemma}{lem:bound-atts}{Bounded Attribute Values}
    Given a collection of constraints
    $\consset_1, \ldots, \consset_\numof$
    over some $\ns$ and $\att$,
    there exists a bound $\attvalbound$ polynomial in the size of
    $\consset_1, \ldots, \consset_\numof$
    such that for any subsequence
    $\consset_{\idxi_1}, \ldots, \consset_{\idxi_\numofalt}$
    if there is a sequence of words
    $\attval_1, \ldots, \attval_\numofalt$
    such that all
    $\attval_\idxj$
    are unique and
    $\attval_\idxj$
    satisfies the constraints in
    $\consset_{\idxi_\idxj}$,
    then there is a sequence of words such that the length of each
    $\attval_\idxj$
    is bounded by $\attvalbound$,
    all
    $\attval_\idxj$
    are unique, and
    $\attval_\idxj$
    satisfies the constraints in
    $\consset_{\idxi_\idxj}$,
\end{namedlemma}

The proof uses ideas from Muscholl and Walukiewicz's NP fragment of LTL~\cite{MW05}.
We first, for each set of constraints $\consset$, construct a deterministic finite word automaton $\aut$ that accepts only words satisfying all constraints in $\consset$.
This automaton has a polynomial number of states and can easily be seen to have a short solution by a standard pumping argument.
Given automata
$\aut_1, \ldots, \aut_\numof$
with at most $\numstates$ states and $\numcons$ constraints in each set of constraints, we can again use pumping to show there is a sequence of distinct words
$\attval_1, \ldots, \attval_\numof$
such that each
$\attval_\idxi$
is accepted by
$\aut_\idxi$
and the length of
$\attval_\idxi$
is at most
$\numof\cdot\numstates\cdot\numcons$.

\paragraph{The Automata}

We define a type of word automata based on a model by Muscholl and Walukiewicz to show and NP upper bound for a variant of LTL.
These automata read words and keep track of which constraints in $\consset$ have been satisfied or violated.
They accept once all positive constraints have been satisfied and no negative constraints have been observed.

In the following, let
$\consprefixes{\consset}$
be the set of words $\attval'$ such that $\attval'$ is a prefix of some $\attval$ with
$\opattns{\ns}{\att}{\attop}{\attval} \in \consset$
or
$\cssneg{\opattns{\ns}{\att}{\attop}{\attval}} \in \consset$.
Moreover, let $\startchar$ and $\lastchar$ be characters not in $\alphabet$ that will mark the beginning and end of the word respectively.
Additionally, let $\eword$ denote the empty word.
Finally, we write
$\attval \factor \attval'$
if $\attval$ is a \defn{factor} of $\attval'$, i.e.,
$\attval' = \attval_1 \attval \attval_2$
for some $\attval_1$ and $\attval_2$.

\begin{nameddefinition}{def:consaut}{$\consaut{\consset}$}
    Given a set $\consset$ of constraints over $\ns$ and $\att$, we define
    $
        \consaut{\consset} = \tup{\astates, \atrans, \consset}
    $
    where
    \begin{compactitem}
    \item
        $\astates$ is the set of all words $\cha_1\attval\cha_2$ such that
        \begin{compactitem}
        \item
            $\attval \in \consprefixes{\consset}$, and
        \item
            $\cha_1, \cha_2 \in \alphabet \cup \set{\eword,
                                                    \startchar,
                                                    \lastchar}$.
        \end{compactitem}
    \item
        $\atrans \subseteq \astates
                           \times
                           \brac{\alphabet \cup \set{\startchar,
                                                     \lastchar}}
                           \times
                           \astates$
        is the set of transitions $\attval \atranch{\cha} \attval'$
        where $\attval'$ is the longest suffix of
        $\attval\cha$
        such that
        $\attval' \in \astates$.
    \end{compactitem}
\end{nameddefinition}
Observe that the size of the automaton $\consaut{\consset}$ is polynomial
in the size of $\consset$.

A \defn{run} of $\consaut{\consset}$
over a word with beginning and end marked
$\cha_1\ldots\cha_\numof \in \startchar \alphabet^\ast \lastchar$
is
\[
    \consrunst{\attval_0}{\conssat_0}{\consviol_0}
    \atranch{\cha_1}
    \consrunst{\attval_1}{\conssat_1}{\consviol_1}
    \atranch{\cha_2}
    \cdots
    \atranch{\cha_\numof}
    \consrunst{\attval_\numof}{\conssat_\numof}{\consviol_\numof}
\]
where
$\attval_0 = \eword$
and for all
$1 \leq \idxi \leq \numof$
we have
$\attval_{\idxi - 1} \atranch{\cha_\idxi} \attval_\idxi$
and
$\conssat_\idxi, \consviol_\idxi \subseteq \consset$
track the satisfied and violated constraints respectively.
That is
$\conssat_0 = \consviol_0 = \emptyset$,
and for all
$1 \leq \idxi \leq \numof$
we have
(noting
 $\startchar\attval \factor \attval_\idxi$
 implies
 $\startchar\attval$
 is a prefix of $\attval_\idxi$, and similar for
 $\attval\lastchar$)
$\conssat_\idxi =$
\[
    \begin{array}{l}
        \conssat_{\idxi-1} \cup
        \setcomp{\attisns{\ns}{\att}{\attval} \in \consset}
                {\startchar \attval \lastchar = \attval_\idxi}
        \ \cup
        \\
        \setcomp{\atthasns{\ns}{\att}{\attval} \in \consset}
                {\exists
                     \cha_1 \in \set{\startchar, \cspace},
                     \cha_2 \in \set{\cspace, \lastchar} \ .\ %
                 \cha_1 \attval \cha_2 \factor \attval_\idxi
                } \ \cup
        \\
        \setcomp{\attbeginns{\ns}{\att}{\attval} \in \consset}
                {\exists
                      \cha_2 \in \set{\lastchar, \mdash} \ .\ %
                 \startchar \attval \cha_2 \factor \attval_\idxi
                } \ \cup
        \\
        \setcomp{\attstrbeginns{\ns}{\att}{\attval} \in \consset}
                {\startchar \attval \factor \attval_\idxi} \cup
        \setcomp{\attstrendns{\ns}{\att}{\attval} \in \consset}
                {\attval \lastchar \factor \attval_\idxi} \ \cup
        \\
        \setcomp{\attstrsubns{\ns}{\att}{\attval} \in \consset}
                {\attval \factor \attval_\idxi}
    \end{array}
\]
and $\consviol_\idxi =$
\[
    \begin{array}{l} %
        \consviol_{\idxi-1} \cup %
        \setcomp{\cssneg{\attisns{\ns}{\att}{\attval}} \in \consset} %
                {\startchar\attval\lastchar = \attval_\idxi} %
        \ \cup %
        \\ %
        \setcomp{\cssneg{\atthasns{\ns}{\att}{\attval}} \in \consset} %
                {\exists %
                    \begin{array}{c} %
                        \cha_1 \in \set{\startchar, \cspace}, %
                        \\ %
                        \cha_2 \in \set{\cspace, \lastchar} %
                    \end{array} %
                 \ .\ 
                 \cha_1 \attval \cha_2 \factor \attval_\idxi %
                } \ \cup %
        \\ %
        \setcomp{\cssneg{\attbeginns{\ns}{\att}{\attval}} \in \consset} %
                {\exists %
                      \cha_2 \in \set{\lastchar, \mdash} \ .\ 
                 \startchar \attval \cha_2 \factor \attval_\idxi %
                } \ \cup %
        \\ %
        \setcomp{\cssneg{\attstrbeginns{\ns}{\att}{\attval}} \in \consset} %
                {\startchar \attval \factor \attval_\idxi} \ \cup %
        \\ %
        \setcomp{\cssneg{\attstrendns{\ns}{\att}{\attval}} \in \consset} %
                {\attval \lastchar \factor \attval_\idxi} \ \cup %
        \\ %
        \setcomp{\cssneg{\attstrsubns{\ns}{\att}{\attval}} \in \consset} %
                {\attval \factor \attval_\idxi} \ . %
    \end{array} %
\]
Such a run is \defn{accepting} if
$
    \conssat_\numof = \setcomp{\opattns{\ns}{\att}{\attop}{\attval}}
                              {\opattns{\ns}{\att}{\attop}{\attval}
                               \in
                               \consset}
$
and
$\consviol_\numof = \emptyset$.
That is, all positive constraints have been satisfied and no negative constraints have been violated.

\paragraph{Short Solutions}

We show the existence of short solutions via the following lemma.
The proof of this lemma is a simple pumping argument which appears below.
Intuitively, if a satisfying word is shorter than
$\numstates \cdot \numcons$
we do not change it.
If it is longer than
$\numstates \cdot \numcons$
any accepting run of the automaton on this word must contain a repeated
$\consrunst{\attval}{\conssat}{\consviol}$.
We can thus pump down this word to ensure that it is shorter than
$\numstates \cdot \numcons$.
Then, to ensure it is unique, we pump it up to some unique length of at most
$\numof \cdot \numstates \cdot \numcons$.

\begin{namedlemma}{lem:short-attvals}{Short Attribute Values}
    Given a sequence of sets of constraint automata
    $\consaut{\consset_1}, \ldots, \consaut{\consset_\numof}$
    each with at most $\numstates$ states and at most $\numcons$ constraints in each $\consset_\idxi$, if there is a sequence of pairwise unique words
    $\attval_1, \ldots, \attval_\numof$
    such that for all
    $1 \leq \idxi \leq \numof$
    there is an accepting run of
    $\consaut{\consset_\idxi}$
    over
    $\attval_\idxi$,
    then there exists such a sequence where the length of each
    $\attval_\idxi$
    is at most
    $\numof \cdot \numstates \cdot \numcons$.
\end{namedlemma}

To obtain
\reflemma{lem:bound-atts}
we observe that for any subsequence
$\consset_{\idxi_1}, \ldots, \consset{\idxi_\numofalt}$
we have
$\numofalt \cdot \numstates' \cdot \numcons'
 \leq
 \numof \cdot \numstates \cdot \numcons$
since
$\numofalt \leq \numof$
and the max number of states $\numstates'$ and constraints $\numcons'$ in the subsequence have
$\numstates' \leq \numstates$
and
$\numcons' \leq \numcons$.

We give the proof of Lemma~\ref{lem:short-attvals}.
That is, given a sequence of sets of constraint automata
$\consaut{\consset_1}, \ldots, \consaut{\consset_\numof}$
each with at most $\numstates$ states and at most $\numcons$ constraints in each $\consset_\idxi$, if there is a sequence of pairwise unique words
$\attval_1, \ldots, \attval_\numof$
such that for all
$1 \leq \idxi \leq \numof$
there is an accepting run of
$\consaut{\consset_\idxi}$
over
$\attval_\idxi$,
then there exists such a sequence where the length of each
$\attval_\idxi$
is at most
$\numof \cdot \numstates \cdot \numcons$.

To prove the lemma, take a sequence
$\attval_1, \ldots, \attval_\numof$
such that each $\attval_\idxi$ is unique and accepted by
$\consaut{\consset_\idxi}$.
We proceed by induction, constructing
$\attval'_1, \ldots, \attval'_\idxi$
such that each $\attval'_\idxj$ is unique, accepted by
$\consaut{\consset_\idxj}$,
and of length $\runlen$ such that either
\begin{compactitem}
\item
    $\runlen \leq \numstates \cdot \numcons$
    and
    $\attval'_\idxj = \attval_\idxj$, or
\item
    $\idxi \cdot \numstates \cdot \numcons
     \leq \runlen \leq
     (\idxi+1) \cdot \numstates \cdot \numcons$.
\end{compactitem}
When $\idxi = 0$ the result is vacuous.
For the induction there are two cases.

When the length $\runlen$ of $\attval_\idxi$ is such that
$\runlen \leq \numstates \cdot \numcons$
we set
$\attval'_\idxi = \attval_\idxi$.
We know $\attval'_\idxi$ is unique amongst
$\attval'_1, \ldots, \attval'_\idxi$
since for all $\idxj < \idxi$ either $\attval'_\idxj$ is longer than $\attval'_\idxi$ or is equal to $\attval_\idxj$ and thus distinct from $\attval_\idxi$.

When
$\runlen > \numstates \cdot \numcons$
we use a pumping argument to pick some $\attval'_\idxi$ of length $\runlen'$ such that
$\idxi \cdot \numstates \cdot \numcons
 \leq
 \runlen'
 \leq
 (\idxi + 1) \cdot \numstates \cdot \numcons$.
This ensures that $\attval'_\idxi$ is unique since it is the only word whose length lies within the bound.
We take the accepting run
\[
    \consrunst{\attvalalt_0}{\conssat_0}{\consviol_0}
    \atranch{\cha_1}
    \consrunst{\attvalalt_1}{\conssat_1}{\consviol_1}
    \atranch{\cha_2}
    \cdots
    \atranch{\cha_\numof}
    \consrunst{\attvalalt_\runlen}{\conssat_\numof}{\consviol_\runlen}
\]
of $\attval_\idxi$ and observe that the values of $\conssat_\idxj$ and $\consviol_\idxj$ are increasing by definition.
That is
$\conssat_\idxj \subseteq \conssat_{\idxj+1}$
and
$\consviol_\idxj \subseteq \consviol_{\idxj+1}$.
By a standard down pumping argument, we can construct a short accepting run containing only distinct configurations of length bound by
$\numstates \cdot \numcons$.
We construct this run by removing all cycles from the original run.
This maintains the acceptance condition.
Next we obtain an accepted word of length
$\idxi \cdot \numstates \cdot \numcons
 \leq
 \runlen'
 \leq
 (\idxi + 1) \cdot \numstates \cdot \numcons$.
Since
$\runlen > \numstates \cdot \numcons$
we know there exists at least one configuration
$\consrunst{\attvalalt}{\conssat}{\consviol}$
in the short run that appeared twice in the original run.
Thus there is a run of the automaton from
$\consrunst{\attvalalt}{\conssat}{\consviol}$
back to
$\consrunst{\attvalalt}{\conssat}{\consviol}$
which can be bounded by
$\numstates \cdot \numcons$
by the same downward pumping argument as before.
Thus, we insert this run into the short run the required number of times to obtain an accepted word
$\attval'_\idxi$
of the required length.

Thus, by induction, we are able to obtain the required short words
$\attval'_1, \ldots, \attval'_\numof$
as needed.

\subsubsection{Missing definitions for $\attspres{\csscond}{\wordposvec}$}
\label{sec:attspres-missing}

\[
    \begin{array}{rcl}
        \attspres{\atthasns{\ns}{\att}{\attval}}{\wordposvec} %
        &=& %
        \brac{
            \begin{array}{c}
                \bigwedge\limits_{1 \leq \idxj \leq \numofalt}\brac{ %
                    \begin{array}{c} %
                        \wordpos{\ns}{\att}{\idxi}{\idxj} = \cha_\idxj %
                        \ \land \\ %
                        \brac{ %
                            \begin{array}{c} %
                                \wordpos{\ns}{\att}{\idxi}{\numofalt+1} = \nullch %
                                \\ \lor \\ %
                                \wordpos{\ns}{\att}{\idxi}{\numofalt+1} = \cspace %
                            \end{array} %
                        } %
                    \end{array} %
                } \\ %
                \lor \\ %
                \bigvee\limits_{1 \leq \idxj \leq \attvalbound - \numofalt - 1} %
                \brac{ %
                    \begin{array}{c} %
                        \wordpos{\ns}{\att}{\idxi}{\idxj - 1} = \cspace %
                        \ \land \\ %
                        \bigwedge\limits_{1 \leq \idxj' \leq \numofalt} %
                        \wordpos{\ns}{\att}{\idxi}{\idxj + \idxj'} = \cha_\idxj %
                        \ \land \\ %
                        \brac{\wordpos{\ns}{\att}{\idxi}{\idxj + \numofalt + 1} = \nullch %
                          \lor %
                          \wordpos{\ns}{\att}{\idxi}{\idxj + \numofalt + 1} = \cspace} %
                    \end{array} %
                } %
            \end{array}
        }
        \\
        \attspres{\attbeginns{\ns}{\att}{\attval}}{\wordposvec} %
        &=& %
        \bigwedge\limits_{1 \leq \idxj \leq \numofalt} %
            \wordpos{\ns}{\att}{\idxi}{\idxj} = \cha_\idxj %
        \land %
        \brac{\wordpos{\ns}{\att}{\idxi}{\numofalt+1} = \nullch %
              \lor %
              \wordpos{\ns}{\att}{\idxi}{\numofalt+1} = \mdash} %
        \\
        \attspres{\attstrendns{\ns}{\att}{\attval}}{\wordposvec} %
        &=& %
        \bigvee\limits_{0 \leq \idxj \leq \attvalbound - \numofalt - 1} %
        \brac{ %
            \begin{array}{c} %
                \bigwedge\limits_{1 \leq \idxj' \leq \numofalt} %
                \wordpos{\ns}{\att}{\idxi}{\idxj + \idxj'} = \cha_\idxj %
                \ \land \\ %
                \wordpos{\ns}{\att}{\idxi}{\idxj + \numofalt + 1} = \nullch %
            \end{array} %
        } %
    \end{array} %
\]

\subsubsection{Negating Positional Formulas}
\label{sec:negating-pos-formulas}

We need to negate selectors like $\psnthchild{\coefa}{\coefb}$,
For completeness, we give the definition of the negation below.

We decompose $\coefb$ according to the period $\coefa$.
I.e.
$\coefb = \coefa \coefb_1 + \coefb_2$,
where $\coefb_1$ and $\coefb_2$ are the unique integers such that
$\abs{\coefb_2} < \abs{\coefa}$
and
$\coefb_1 \coefa < 0$ implies $\coefb_2 \leq 0$
and
$\coefb_1 \coefa > 0$ implies $\coefb_2 \geq 0$.

\begin{definition}[$\nomatch{\xvar}{\coefa}{\coefb}$]
    Given constants $\coefa, \coefb, \coefb_1$, and $\coefb_2$ as above, we define the formula
    $\nomatch{\xvar}{\coefa}{\coefb}$
    to be
    \[
        \begin{array}{c} %
            \brac{0 \geq \coefa \land \xvar < \coefb} %
            \lor %
            \brac{0 \leq \coefa \land \xvar > \coefb} %
            \ \lor \\ %
            \brac{ %
                \exists \nvar \ .\ \exists \denotevar{\coefb}'_2 . %
                \brac{ %
                    \begin{array}{c} %
                        \abs{\denotevar{\coefb}'_2} < \abs{\coefa} %
                        \ \land \\ %
                        \brac{\coefb_1\coefa < 0 %
                            \Rightarrow \denotevar{\coefb}'_2 \leq 0} %
                        \ \land \\ %
                        \brac{\coefb_1\coefa > 0 %
                            \Rightarrow \denotevar{\coefb}'_2 \geq 0} %
                        \ \land \\ %
                        \denotevar{\coefb}'_2 \neq \coefb_2 %
                        \ \land \\ %
                        \xvar = \coefa\nvar + %
                                \coefa\coefb_1 + %
                                \denotevar{\coefb}'_2 %
                    \end{array} %
                } %
            } %
        \end{array} %
    \]
\end{definition}
In the following, whenever we negate a formula of the form
$\neg \brac{\exists \nvar . \xvar = \coefa \nvar + \coefb}$
we will use \acmeasychair{}{the formula}
$\nomatch{\xvar}{\coefa}{\coefb}$.
One can verify that the resulting formula is existential Presburger.
We show that our negation of periodic constraints is correct.

\begin{namedproposition}{prop:nomatch}{Correctness of $\nomatch{\xvar}{\coefa}{\coefb}$}
    Given constants $\coefa$ and $\coefb$, we have
    \[
        \neg\brac{\exists \nvar . \xvar = \coefa \nvar + \coefb}
        \iff
        \nomatch{\xvar}{\coefa}{\coefb} \ .
    \]
\end{namedproposition}
\begin{proof}
    We first consider $\coefa = 0$.  Since there is no $\coefb'_2$ with
    $\abs{\coefb'_2} < 0$ we have to prove
    \[
        \neg\brac{\xvar = \coefb}
        \iff
        \brac{\xvar < \coefb} \lor \brac{\xvar > \coefb}
    \]
    which is immediate.

    In all other cases, the conditions on $\coefb_2$ and $\coefb'_2$ ensure that we always have
    $0 < \abs{\coefb_2 - \coefb'_2} < \abs{\coefa}$.

    If
    $\exists \nvar . \xvar = \coefa \nvar + \coefa \coefb_1 + \coefb_2$
    then if $\coefa > 0$ it is easy to verify that we don't have $\xvar < \beta$ (since $\nvar = 0$ gives $\xvar = \beta$ as the smallest value of $\xvar$) and similarly when $\coefa < 0$ we don't have $\xvar < \beta$.
    To disprove the final disjunct of of
    $\nomatch{\xvar}{\coefa}{\coefb}$
    we observe there can be no $\nvar'$ s.t.
    $\xvar = \coefa \nvar' + \coefa \coefb_1 + \coefb'_2$
    since
    $\xvar = \coefa \nvar + \coefa \coefb_1 + \coefb_2$
    and
    $0 < \abs{\coefb_2 - \coefb'_2} < \abs{\coefa}$.

    In the other direction, we have three cases depending on the satisfied disjunct of $\nomatch{\xvar}{\coefa}{\coefb}$.
    Consider $\coefa > 0$ and $\xvar < \coefb$.
    In this case there is no $\nvar$ such that
    $\xvar = \coefa \nvar + \coefa \coefb_1 + \coefb_2 = \coefa \nvar + \coefb$
    since $\nvar = 0$ gives the smallest value of $\xvar$, which is $\coefb$.
    The case is similar for the second disjunct with $\coefa < 0$.

    The final disjunct gives some $\nvar'$ such that
    $\xvar = \coefa \nvar' + \coefa \coefb_1 + \coefb'_2$
    with
    $0 < \abs{\coefb_2 - \coefb'_2} < \abs{\coefa}$.
    Hence, there can be no $\nvar$ with
    $\xvar = \coefa \nvar + \coefa \coefb_1 + \coefb_2$.
\end{proof}

\subsubsection{Correctness of Presburger Encoding}
\label{sec:aut-emptiness-proof}

We prove soundness and completeness of the Presburger encoding of CSS automata non-emptiness in the two lemmas below.

\begin{lemma}\label{lem:aut-enc-complete}
    For a CSS automaton $\cssaut$, we have
    \[
        \ap{\Lang}{\cssaut} \neq \emptyset
        \Rightarrow
        \presof{\cssaut}
        \text{ is satisfiable.}
    \]
\end{lemma}
\begin{proof}
    We take a run of $\cssaut$ and construct a satisfying assignment to the variables in $\presof{\cssaut}$.
    That is take a document tree
    $\tree = \tup{\treedom, \treelab}$,
    node
    $\node \in \treedom$,
    and sequence
    \[
        \astate_0, \node_0,
        \astate_1, \node_1,
        \ldots,
        \astate_\runlen, \node_\runlen,
        \astate_{\runlen+1}
        \in
        \brac{\astates \times \treedom}^\ast \times \set{\afinstate}
    \]
    that is an accepting run.
    We know from \refproposition{prop:boundedtypes} that $\tree$ can only use namespaces from $\finof{\nspaces}$ and elements from $\finof{\eles}$.
    Let $\atrant_0, \ldots, \atrant_\runlen$ be the sequence of transitions used in the accepting run.
    We assume (w.l.o.g.) that no transition is used twice.
    We construct a satisfying assignment to the variables as follows.
    \begin{compactitem}
    \item
        $\astatevar{\idxi} = \astate_\idxi$ for all $\idxi \leq \runlen + 1$ and $\astatevar{\idxi} = \afinstate$ for all $\idxi > \runlen + 1$.

    \item
        $\nsvar{\idxi} = \ap{\treelabns}{\node_\idxi}$ for all $\idxi \leq \runlen+1$ ($\nsvar{\idxi}$ can take any value for other values of $\idxi$).

    \item
         $\elevar{\idxi} = \ap{\treelabele}{\node_\idxi}$ for all $\idxi \leq \runlen+1$ ($\elevar{\idxi}$ can take any value for other values of $\idxi$).

    \item
        $\pclsvar{\idxi}{\pcls}
         =
         \brac{\pcls \in \ap{\treelabpclss}{\node_\idxi}}$,
        for each pseudo-class
        $\pcls \in \pclss \setminus \set{\psroot}$
        and
        $\idxi \leq \runlen+1$
        (these variables can take any value for $\idxi > \runlen+1$).

    \item
        $\numvar{\idxi} = \treedir$, when $\node_\idxi = \node'\treedir$ for some $\node'$ and $\treedir$ and $1 \leq \idxi \leq \runlen + 1$, otherwise $\numvar{\idxi}$ can take any value.

    \item
        $\numtypevar{\idxi}{\qele{\ns}{\ele}} = \idxj$,
        where $\idxj$ is the number of nodes of type $\qele{\ns}{\ele}$ preceding $\node_\idxi$ in the sibling order.
        That is
        $\node_\idxi = \node'\treedir$
        for some $\node'$ and $\treedir$ and
        \[
            \idxj = \sizeof{ %
                \setcomp{\node'\treedir'} %
                        {\begin{array}{c} %
                             \treedir' < \treedir %
                             \ \land %
                             \node'\treedir' \in \treedom %
                             \ \land \\ %
                             \ap{\treelabns}{\node'\treedir'} %
                             = %
                             \ap{\treelabns}{\node} %
                             \ \land \\ %
                             \ap{\treelabele}{\node'\treedir'} %
                             = %
                             \ap{\treelabele}{\node} %
                         \end{array}} %
            } \ . %
        \]
        When $\idxi = 0$ or $\idxi > \runlen + 1$ we can assign any value to $\numtypevar{\idxi}{\qele{\ns}{\ele}}$.

    \item
        $\lnumvar{\idxi} = \totnumof - \treedir$ where $\idxi \leq \runlen + 1$ and $\node = \node'\treedir$ for some $\node'$ and $\treedir$ and $\totnumof$ is the smallest number such that $\node'\totnumof \notin \treedom$.
        For $\idxi = 0$ or $\idxi > \runlen + 1$ the variable $\lnumvar{\idxi}$ can take any value.

    \item
        $\lnumtypevar{\idxi}{\qele{\ns}{\ele}} = \idxj$,
        where $\idxj$ is the number of nodes of type $\qele{\ns}{\ele}$ suceeding $\node_\idxi$ in the sibling order.
        That is
        $\node_\idxi = \node'\treedir$
        for some $\node'$ and $\treedir$ and
        \[
            \idxj = \sizeof{ %
                \setcomp{\node'\treedir'} %
                        {\begin{array}{c} %
                             \treedir' > \treedir %
                             \ \land %
                             \node'\treedir' \in \treedom %
                             \ \land \\ %
                             \ap{\treelabns}{\node'\treedir'} %
                             = %
                             \ap{\treelabns}{\node} %
                             \ \land \\ %
                             \ap{\treelabele}{\node'\treedir'} %
                             = %
                             \ap{\treelabele}{\node} %
                         \end{array}} %
            } \ . %
        \]
        When $\idxi = 0$ or $\idxi > \runlen + 1$ we can assign any value to $\lnumtypevar{\idxi}{\qele{\ns}{\ele}}$.

    \item
        Assignments to
        $\wordpos{\ns}{\att}{\idxi}{\idxj}$
        are discussed below.
    \end{compactitem}
    It remains to prove that the given assignment satisfies the formula.

    Recall
    \[
        \presof{\cssaut} = \brac{ %
            \begin{array}{c} %
                \astatevar{0} = \ainitstate %
                \land %
                \astatevar{\numof} = \afinstate %
                \land \\ %
                \bigwedge\limits_{0 \leq \idxi < \numof} %
                \brac{ %
                    \tranpres{\idxi} \lor %
                    \astatevar{\idxi} = \afinstate %
                } %
                \land \\ %
                \consistent %
            \end{array} %
        } \ . %
    \]
    The first two conjuncts follow immediately from our assignment to $\astatevar{\idxi}$ and that the chosen run was accepting.
    Next we look at the third conjunct and simultaneously prove $\consistentnums$.
    When $\idxi \geq \runlen + 1$ we assigned $\afinstate$ to $\astatevar{\idxi}$ and can choose any assignment that satisfies $\consistentnums$.
    Otherwise we show we satisfy $\tranpres{\idxi}$ by showing we satisfy $\tranprest{\idxi}{\atrant_\idxi}$.
    We also show $\consistentnums$ is satsfied by induction, noting it is immediate for $\idxi = 0$ and that for $\idxi = 1$ we must have either the first orlast case which do not depend on the induction hypothesis.
    Consider the form of $\atrant_\idxi$.
\begin{compactenum}
\item
    When
    $\atrant_\idxi = \astate_\idxi \atran{\arrchild}{\csssim} \astate_{\idxi+1}$
    we immediately confirm the values of $\astatevar{\idxi}$, $\astatevar{\idxi+1}$, $\numvar{\idxi+1}$, $\numtypevar{\idxi+1}{\qele{\ns}{\ele}}$ satisfy the constraint.
    Similarly for
    $\neg{\pclsvar{\idxi}{\psempty}}$
    since we know $\psempty \notin \ap{\treelabpclss}{\node_\idxi}$.
    We defer the argument for
    $\csssimpres{\csssim}{\idxi}$
    until after the case split.
    That $\consistentnums$ is satisfied can also be seen directly.

\item
    When
    $\atrant_\idxi = \astate_\idxi \atran{\arrneighbour}{\csssim} \astate_{\idxi+1}$
    we know $\node_\idxi = \node'\treedir$ and $\node_{\idxi+1} = \node'(\treedir+1)$ for some $\node'$ and $\treedir$.
    We can easily check the values of $\astatevar{\idxi}$, $\astatevar{\idxi+1}$, $\numvar{\idxi+1}$, $\lnumvar{\idxi}$, $\numtypevar{\idxi+1}{\qele{\ns}{\ele}}$, and $\lnumtypevar{\idxi+1}{\qele{\ns}{\ele}}$ satisfy the constraint.
    We defer the argument for
    $\csssimpres{\csssim}{\idxi}$
    until after the case split.
    To show $\consistentnums$ we observe $\numvar{\idxi+1}$ is increased by $1$ and only one $\numtypevar{\idxi+1}{\qele{\ns}{\ele}}$ is increased by $1$, the others being increased by $0$.
    Similarly for $\lnumvar{\idxi}$ and $\numtypevar{\idxi+1}{\qele{\ns}{\ele}}$.
    Hence the result follows from induction.

\item
    When
    $\atrant_\idxi = \astate_\idxi \atran{\arrsibling}{\isany} \astate_{\idxi+1}$
    we know $\node_\idxi = \node'\treedir$ and $\node_{\idxi+1} = \node'(\treedir')$ for some $\node'$, $\treedir$, and $\treedir < \treedir'$.
    We can easily check the values of $\astatevar{\idxi}$, $\astatevar{\idxi+1}$, $\numvar{\idxi+1}$, $\lnumvar{\idxi}$, $\numtypevar{\idxi+1}{\qele{\ns}{\ele}}$, and $\lnumtypevar{\idxi+1}{\qele{\ns}{\ele}}$ satisfy the constraint.
    We defer the argument for
    $\csssimpres{\csssim}{\idxi}$
    until after the case split.
    To satisfy the constraints over the position variables, we observe that values for
    $\shiftvar$
    and
    $\shiftvartype{\qele{\ns}{\ele}}$
    can be chosen easily for the specified assignment.
    Combined with induction this shows $\consistentnums$ as required.

\item
    When
    $\atrant_\idxi = \astate_\idxi \atran{\arrlast}{\csssim} \astate_{\idxi+1}$

    We can easily check the values of $\astatevar{\idxi}$ and $\astatevar{\idxi+1}$.
    We defer the argument for
    $\csssimpres{\csssim}{\idxi}$
    until after the case split.
    By induction we immediately obtain $\consistentnums$.
\end{compactenum}

We show
$\csssimpres{\csssim}{\idxi}$
is satisfied for each $\node_\idxi$ and $\csssim$ labelling $\atrant_\idxi$.
Take a node $\node$ and $\csssim = \csstype\cssconds$ from this sequence.
Note $\node$ satisfies $\csssim$ since thw run is accepting.
Recall
\[
    \csssimpres{\csstype\cssconds}{\idxi} %
    = %
    \brac{ %
        \begin{array}{c} %
            \csssimpres{\csstype}{\idxi} %
            \land \\ %
            \brac{ %
                \bigwedge\limits_{\csscond \in \noatts{\cssconds}} %
                    \csssimpres{\csscond}{\idxi} %
            } %
            \land \\ %
            \attspres{\csstype\cssconds}{\idxi} %
        \end{array} %
    } \ . %
\]
From the type information of $\node$ we immediately satisfy $\csssimpres{\csstype}{\idxi}$.

For a positive $\csscond \in \cssconds$ there are several cases.
If $\csscond = \psroot$ then we know we are in $\node_0$ and the encoding is $\ptrue$.
If $\csscond$ is some other pseudo class $\pcls$ then the encoding of $\csscond$ is $\pclsvar{\idxi}{\pcls}$ and we assigned true to this variable.
For
$\psnthchild{\coefa}{\coefb}$
and
$\psnthlastchild{\coefa}{\coefb}$
satisfaction of the encoding follows immediately from $\node$ satisfying $\csscond$ and our assignment to $\numvar{\idxi}$ and $\lnumvar{\idxi}$.
We satisfy the encodings of
$\psnthoftype{\coefa}{\coefb}$,
$\psnthlastoftype{\coefa}{\coefb}$,
$\psonlychild$,
and
$\psonlyoftype$
similarly.
The latter follow since an only child is position $1$ from the start and end, and an only of type node has $0$ strict predecessors or successors of the same type.

For a negative $\csscond \in \cssconds$ there are several cases.
If $\csscond = \cssneg{\psroot}$ then we know we are not in $\node_0$ and the encoding is $\ptrue$.
If $\csscond$ is the negation of some other pseudo class $\pcls$ then the encoding of $\csscond$ is $\neg\pclsvar{\idxi}{\pcls}$ and we assigned false to this variable.
For the selectors
$\cssneg{\psnthchild{\coefa}{\coefb}}$
and the opposite selector
$\cssneg{\psnthlastchild{\coefa}{\coefb}}$
satisfaction of the encoding follows immediately from $\node$ satisfying $\csscond$, our assignment to $\numvar{\idxi}$ and $\lnumvar{\idxi}$ as well as \refproposition{prop:nomatch}.
We satisfy encodings of
$\cssneg{\psnthoftype{\coefa}{\coefb}}$
and of
$\cssneg{\psnthlastoftype{\coefa}{\coefb}}$
in a likewise fashion.
For the remaining cases of
$\cssneg{\psonlychild}$,
and
$\cssneg{\psonlyoftype}$
the property follows since, for the former the node must either not be position $1$ from the start or end, and for the latter a not only of type node has more than $0$ strict predecessors or successors of the same type.

Next, to satisfy
$\attspres{\csstype\cssconds}{\idxi}$
we have to satisfy a number of conjuncts.
First, if we have a word
$\cha_1 \ldots \cha_\numof$
we assign it to the variables
$\wordpos{\ns}{\att}{\idxi}{\idxj}$
(where $\node$ is the $\idxi$th in the run and $\idxj$ ranges over all word positions within the cimputed bound)
ny assigning
$\wordpos{\ns}{\att}{\idxi}{\idxj} = \cha_\idxj$
when
$\idxj \leq \numof$
and
$\wordpos{\ns}{\att}{\idxi}{\idxj} = \nullch$
otherwise.

In all cases below, it is straightforward to observe that it a word (within the computed length bound) satisfies
$\opattns{\ns}{\att}{\attop}{\attval}$
or
$\cssneg{\opattns{\ns}{\att}{\attop}{\attval}}$
then the encoding
$\attspresns{\ns}{\att}{\opattns{\ns}{\att}{\attop}{\attval}}{\idxi}$
or
$\neg\attspresns{\ns}{\att}{\opattns{\ns}{\att}{\attop}{\attval}}{\idxi}$
is satisfied by our variable assignment.
Similarly $\attspresnulls{\wordposvec}$ is straightforwardly satisfied.
Hence, if a word satisfies $\consset$ then our assignment to the variables means
$\attspresns{\ns}{\att}{\consset}{\idxi}$
is also satisfied.

There are a number of cases of conjuncts for attribute selectors.
The simplest is for sets
$\cssconds^\ns_\att$
where we see immediately that all constraints are satisfied for
$\ap{\ap{\treelabatts}{\node}}{\ns, \att}$
and hence we assign this value to the appropriate variables and the conjuct is satisfied also.
For each
$\hasatt{\att}$
and
$\opatt{\att}{\attop}{\attval} \in \cssconds$
we have in the document some namespace $\ns$ such that
$\ap{\ap{\treelabatts}{\node}}{\ns, \att}$
satisfies the attribute selector and all negative selectors applying to all namespaces.
Let $\ns'$ be the fresh name space assigned to the selector during the encoding and $\consset$ be the full set of constraints belonging to the conjunct (i.e. including negative ones).
We assign to the variable
$\wordpos{\ns}{\att}{\idxi}{\idxj}$
the $\idxj$th character of
$\ap{\ap{\treelabatts}{\node}}{\ns, \att}$
(where $\node$ is the $\idxi$th in the run)
and satisfy the conjuct as above.
Note here that a single value of $\qatt{\ns}{\att}$ is assigned to several $\qatt{\ns'}{\att}$.
This is benign with respect to the global uniqueness required by ID attributes because each copy has a different namespace.

Finally, we have to satisfy the consistency constraints.
We showed $\consistentnums$ above.
The remaining consistency constraints are easily seen to be satisfied: $\consistentids$ because each ID is unique causing at least one pair of characters to differ in every value; $\consistentpseudo$ since it encodes basic consistency constraints on the appearence of pseudo elements in the tree.

Thus, we have satisfied the encoded formula, completing the first direction of the proof.
\end{proof}

\begin{lemma}\label{lem:aut-enc-sound}
    For a CSS automaton $\cssaut$, we have
    \[
        \presof{\cssaut}
        \text{ is satisfiable.}
        \Rightarrow
        \ap{\Lang}{\cssaut} \neq \emptyset
    \]
\end{lemma}
\begin{proof}
    Take a satisfying assignment $\assignment$ to the free variables of $\presof{\cssaut}$.
    We construct a tree and node $\tup{\tree, \node}$ as well as a run of $\cssaut$ accepting $\tup{\tree, \node}$.

    We begin by taking the sequence of states
    $\astate_0, \ldots, \astate_{\runlen+1}$
    which is the prefix of the assignment to
    $\astatevar{0}, \ldots, \astatevar{\numof}$
    where
    $\astate_{\runlen}$
    is the first occurrence of $\afinstate$.
    We will construct a series of transitions
    $\atrant_0, \ldots, \atrant_\runlen$
    with
    $\atrant_\idxi = \astate_\idxi \atran{\arrgen_\idxi}{\csssim_\idxi} \astate_{\idxi+1}$
    for all $0 \leq \idxi \leq \runlen$.
    We will define each $\arrgen_\idxi$ and $\csssim_\idxi$, as well as construct $\tree$ and $\node$ by induction.
    We construct the tree inductively, then show $\csssim_\idxi$ is satisfied for each $\idxi$.

    At first let $\tree_0$ contain only a root node.
    Thus $\node_0$ is necessarily this root.
    Throughout the proof we label each $\node_\idxi$ as follows.
    \begin{compactitem}
    \item
        $\ap{\treelabns}{\node_\idxi} = \ap{\assignment}{\nsvar{\idxi}}$ (i.e. we assign the value given to $\nsvar{\idxi}$ in the satisfying assignment).

    \item
        $\ap{\treelabele}{\node_\idxi} = \ap{\assignment}{\elevar{\idxi}}$.

    \item
        $\ap{\treelabpclss}{\node_\idxi}
         =
         \brac{
             \begin{array}{c}
                 \setcomp{\pcls}
                         {\pcls \in \pclss \setminus \set{\psroot}
                          \land
                          \ap{\assignment}{\pclsvar{\idxi}{\pcls}} = \ptrue}
                 \cup \\
                 \setcomp{\psroot}{\idxi = 0}
             \end{array}
         }$.

    \item
        $\ap{\ap{\treelabatts}{\node_\idxi}}
            {\ns, \att}
         =
         \ap{\assignment}{\wordpos{\ns}{\att}{\idxi}{1}
                          \ldots
                          \wordpos{\ns}{\att}{\idxi}{\attvalbound}}$
        where
        $\ap{\assignment}{\wordpos{\ns}{\att}{\idxi}{1}
                          \ldots
                          \wordpos{\ns}{\att}{\idxi}{\attvalbound}}$
        is the word obtained by stripping all of the null characters from
        $\ap{\assignment}{\wordpos{\ns}{\att}{\idxi}{1}}
         \ldots
         \ap{\assignment}{\wordpos{\ns}{\att}{\idxi}{\attvalbound}}$.
    \end{compactitem}
    We pick $\atrant_\idxi$ as the transition corresponding to a satisfied disjunct of $\tranpres{\idxi}$ (of which there is at least one since $\astate_\idxi \neq \afinstate$ when $\idxi \leq \runlen$).
    Thus, take $\atrant_\idxi = \astate_\idxi \atran{\arrgen_\idxi}{\csssim_\idxi} \astate_{\idxi}$.
    We proceed by a case split on $\arrgen_\idxi$.
    Note only cases $\arrgen_\idxi = \arrchild$ and $\arrgen_\idxi = \arrlast$ may apply when $\idxi = 0$.
    \begin{compactitem}
    \item
        When $\arrgen_\idxi = \arrchild$ we build $\tree_{\idxi+1}$ as follows.
        First we add the leaf node $\node_{\idxi+1} = \node_\idxi 1$.
        Then, if $\idxi > 0$, we add siblings appearing after $\node_\idxi$ with types required by the last of type information.
        That is, we add
        $\ap{\assignment}{\lnumvar{\idxi}} - 1
         =
         \sum\limits_{\qele{\ns}{\ele} \in \eles}
            \ap{\assignment}{\lnumtypevar{\idxi}{\qele{\ns}{\ele}}}$
        siblings appearing after $\node_\idxi$.
        In particular, for each $\ns$ and $\ele$ we add
        $\ap{\assignment}{\lnumtypevar{\idxi}{\qele{\ns}{\ele}}}$
        new nodes.
        Letting $\node_\idxi = \node \treedir$ each of these new nodes $\node'$ will have the form $\node\treedir'$ with $\treedir' > \treedir$.
        We set
        $\ap{\treelabns}{\node'} = \ns$,
        $\ap{\treelabele}{\node'} = \ele$,
        $\ap{\treelabpclss}{\node'} = \emptyset$,
        $\ap{\treelabatts}{\node'} = \emptyset$.

    \item
        When $\arrgen_\idxi = \arrneighbour$ we build $\tree_{\idxi+1}$ by adding a single node to $\tree_\idxi$.
        When $\node_\idxi = \node\treedir$ we add $\node_{\idxi+1} = \node \brac{\treedir+1}$ with the labelling as above.

    \item
        When $\arrgen_\idxi = \arrsibling$ we build $\tree_{\idxi+1}$ as follows.
        We add
        $\ap{\assignment}{\shiftvar} =
         \brac{\ap{\assignment}{\numvar{\idxi+1}}
               -
               \ap{\assignment}{\numvar{\idxi}}}$
        new nodes of the form $\node\treedir'$ where
        $\node_\idxi = \node\treedir$
        and
        $\ap{\assignment}{\numvar{\idxi}}
         = \treedir < \treedir' \leq
         \ap{\assignment}{\numvar{\idxi}}$.
        Let $\node_{\idxi+1}$ be $\node\ap{\assignment}{\numvar{\idxi}}$ labelled as above.
        For the remaining new nodes, for each $\ns$ and $\ele$, we label
        $\ap{\assignment}{\shiftvartype{\qele{\ns}{\ele}}}$
        of the new nodes $\node'$ with
        $\ap{\treelabns}{\node'} = \ns$,
        $\ap{\treelabele}{\node'} = \ele$,
        $\ap{\treelabpclss}{\node'} = \emptyset$,
        $\ap{\treelabatts}{\node'} = \emptyset$.
        Note $\consistentnums$ ensures we have enough new nodes to partition like this.

    \item
        When $\arrgen_\idxi = \arrlast$ and $\idxi = 0$ we have completed building the tree.
        If $\idxi > 0$, we add siblings appearing after $\node_\idxi$ with types required by the last of type information exactly as in the case of $\arrgen = \arrchild$ above.
    \end{compactitem}
    The tree and node we require are the tree and node obtained after reaching some $\arrgen_\idxi = \arrlast$, for which we necessairily have $\idxi = \runlen$ since $\arrlast$ must be and can only be used to reach $\afinstate$.
    In constructing this tree we have almost demonstrated an accepting run of $\cssaut$.
    To complete the proof we need to argue that all $\csssim_\idxi$ are satisfied by $\node_\idxi$ and that the obtained is valid.
    Let $\csstype\cssconds = \csssim_\idxi$.

    To check $\csstype$ we observe that $\csssimpres{\csstype}{\idxi}$ constrains $\nsvar{\idxi}$ and $\elevar{\idxi}$ to values, which when assigned to $\node_\idxi$ as above mean $\node_\idxi$ directly satisfies $\csstype$.

    Now, take some $\csscond \in \cssconds$.
    In each case we argue that $\csssimpres{\csstype}{\idxi}$ ensures the needed properties.
    Note this is straightforward for the attribute selectors due to the directness of the Presburger encoding.
    Consider the remaining selectors.

    First assume $\csscond$ is positive.
    If it is $\psroot$ then we must have $\idxi = 0$ and $\node_\idxi$ is the root node as required.
    For other pseudo classes $\pcls$ we asserted $\pclsvar{\idxi}{\pcls}$ hence we have $\pcls \in \ap{\treelabpclss}{\node_\idxi}$.
    The encoding of the remaining positive constraints can only be satisfied when $\idxi > 0$.
    That is, $\node_\idxi$ is not the root node.

    For
    $\psnthchild{\coefa}{\coefb}$
    observe we constructed $\tree$ such that $\node_\idxi = \node \ap{\assignment}{\numvar{\idxi}}$ for some $\node$.
    From the defined encoding of
    $\csssimpres{\psnthchild{\coefa}{\coefb}}{\idxi}$
    we directly obtain that $\node_\idxi$ satisfies \acmeasychair{}{the selector}
    $\psnthchild{\coefa}{\coefb}$.
    Similarly for
    $\psnthlastchild{\coefa}{\coefb}$ as we always pad the end of the sibling order to ensure the correct number of succeeding siblings.

    For
    $\psnthoftype{\coefa}{\coefb}$
    and
    $\psnthlastoftype{\coefa}{\coefb}$
    selectors, by similar arguments to the previous selectors, we have ensured that there are enough preceeding or succeeding nodes (along with the directness of their Presburger encoding) to ensure these selectors are satisfied by $\node_\idxi$ in $\tree$.

    For $\psonlychild$ we know there are no other children since
    $\ap{\assignment}{\numvar{\idxi}} =
     \ap{\assignment}{\lnumvar{\idxi}} =
     1$.
    Finally for the selector $\psonlyoftype$ we know there are no other children of the same type since
    $\ap{\assignment}{\numtypevar{\idxi}{\qele{\ns}{\ele}}} =
     \ap{\assignment}{\lnumtypevar{\idxi}{\qele{\ns}{\ele}}} =
     0$
    where $\node_\idxi$ has type
    $\qele{\ns}{\ele}$.

    When $\csscond$ is negative there are several cases.
    If it is $\cssneg{\psroot}$ then we must have $\idxi > 0$ and $\node_\idxi$ is not the root node.
    For other pseudo classes $\pcls$ we asserted $\neg\pclsvar{\idxi}{\pcls}$ hence we have $\pcls \notin \ap{\treelabpclss}{\node_\idxi}$.
    The encoding of the remaining positive constraints are always satisfied on the root node.  That is, $\idxi = 0$.
    When $\node_\idxi$ is not the root node we have $\idxi > 0$.

    For
    $\cssneg{\psnthchild{\coefa}{\coefb}}$
    observe we constructed $\tree$ such that $\node_\idxi = \node \ap{\assignment}{\numvar{\idxi}}$ for some $\node$.
    From the definition of
    $\csssimpres{\cssneg{\psnthchild{\coefa}{\coefb}}}{\idxi}$
    we obtain that $\node_\idxi$ does not satisfy the required selector
    $\psnthchild{\coefa}{\coefb}$ via \refproposition{prop:nomatch}.
    Similarly for the last child selector
    \[
        \cssneg{\psnthlastchild{\coefa}{\coefb}}.
    \]

    For
    $\cssneg{\psnthoftype{\coefa}{\coefb}}$
    and
    $\cssneg{\psnthlastoftype{\coefa}{\coefb}}$,
    by similar arguments to the previous selectors, we have ensured that there are enough preceeding or succeeding nodes (along with their Presburger encodings and \refproposition{prop:nomatch}) to ensure these selectors are satisfied by $\node_\idxi$ in $\tree$.

    For $\cssneg{\psonlychild}$ we know there are some other children since
    $\ap{\assignment}{\numvar{\idxi}} > 1$
    or
    $\ap{\assignment}{\lnumvar{\idxi}} > 1$.
    Finally for $\cssneg{\psonlyoftype}$ we know there are other children of the same type since
    $\ap{\assignment}{\numtypevar{\idxi}{\qele{\ns}{\ele}}} > 0$
    or
    $\ap{\assignment}{\lnumtypevar{\idxi}{\qele{\ns}{\ele}}} > 0$
    where $\node_\idxi$ has type
    $\qele{\ns}{\ele}$.

    Thus we have an accepting run of $\cssaut$ over some $\tup{\tree, \node}$.
    However, we finally have to argue that $\tree$ is a valid document tree.
    This is enforced by $\consistentids$ and $\consistentpseudo$.

    First, $\consistentids$ ensures all IDs satisfying the Presburger encoding are unique.
    Since we transferred these values directly to $\tree$ our tree also has unique IDs.

    Next, we have to ensure properties such as no node is both active and inactive.
    These are all directly taken care of by $\consistentpseudo$.
    Thus, we are done.
\end{proof}

%% file: experiments-appendix.tex
\section{Additional Material for the Experiments Section}

\input{experiments-appendix-optimised-emp}

\input{experiments-details.tex}

%% file: experiments-appendix-optimised-emp.tex
\subsection{Optimised CSS Automata Emptiness Check}
\label{sec:optimised-aut-emp}

The reduction presented in Section~\ref{sec:edgeOrder} proves membership in NP.
However, the formula constructed is quite large even for the intersection of two relatively small selectors.
Moreover, selectors generally do no assert complex properties, so for most transitions, the full power of existential Presburger arithmetic is not needed.
Hence, only a small part of each formula requires complex reasoning, while the remainder of the problem is better and easily solved with direct knowledge of the automata.

In this section we present an alternative algorithm.
In essence it is a backwards reachability algorithm for deciding non-emptiness of a CSS automaton.
Instead of constructing a single large query that requires a non-trivial solve time, the backwards reachability algorithm only makes small queries to the SAT solver to enforce constraints that are not simply enforced by a standard automaton algorithm.

The idea is that the automaton collects constraints on the node positions required to satisfy the nth-child (sibling) constraints as it performs its backwards search.
It also tracks extra information to ensure it does not get stuck in a loop, at most one node is labelled $\pstarget$, and all $\idatt$s are unique.
Each time the automaton takes a transition labelled $\arrchild$ it checks whether the current set of sibling constraints is satisfiable.
If so, the automaton can move up to the parent node and begin with a fresh set of sibling constraints.
If not, the automaton cannot execute the transition.
Once the initial state has been reached, it just remains to check whether the $\idatt$ constraints are satisfiable.
If they are, a witness to non-emptiness has been found.

The algorithm is a worklist algorithm, where the worklist consists of tuples of the form
\[
    \tup{\astate,
         \broot,
         \bsibling,
         \btarget,
         \tset,
         \idcons,
         \poscons,
         \idxi}
\]
where
\begin{itemize}
\item
    $\astate$ is the state reached so far,
\item
    $\broot$ is a boolean indicating whether the current node has to be the root,
\item
    $\bsibling$ is a boolean indicating whether the current node has to have siblings,
\item
    $\btarget$ is a boolean indicating whether a node marked $\pstarget$ has been seen on the run so far,
\item
    $\tset$ is the set of transitions seen on the current state (recall all loops are self-loops, so cycle detection can be implemented using $\tset$),
\item
    $\idcons$ is the set of constraints on $\idatt$ attributes in the run so far,
\item
    $\poscons$ is the set of constraints on node positions on the current level of the tree (that is, for the assertion of nth-child constraints),
\item
    $\idxi$ is the index position in the run (akin to the use of indices in the Presburger encoding).
\end{itemize}

The initial worklist contains a single element
\[
    \tup{\afinstate, \pfalse, \pfalse, \pfalse, \emptyset, \emptyset, \emptyset, \numof}
\]
where $\numof$ is the number of transitions of the CSS automaton.
Note, the final element of the tuple will always range between $1$ and $\numof$ since we decrement this counter whenever we take a new transition, and each transition may only be visited once.

In the following, we partition sets of node selector elements into sets containing pseudo-classes, attribute selectors, and positional selectors.
That is, given a node selector
$\csstype\cssconds$
we write
\begin{itemize}
\item
    $\attsof{\cssconds}$
    for the elements of $\cssconds$ of the form $\csscond$ or
    $\cssneg{\csscond}$
    where $\csscond$ is of the form
    $\hasattns{\ns}{\att}$
    or
    $\opattns{\ns}{\att}{\attop}{\attval}$,
\item
    $\posof{\cssconds}$
    for the elements of $\cssconds$ of the form $\csscond$ or
    $\cssneg{\csscond}$
    where $\csscond$ is of the form
    \[
        \begin{array}{c}
            \psnthchild{\coefa}{\coefb},\\
            \psnthlastchild{\coefa}{\coefb},\\
            \psnthoftype{\coefa}{\coefb},\\
            \psnthlastoftype{\coefa}{\coefb}\\,
            \psonlychild, \text{or}\\
            \psonlyoftype,
        \end{array}
    \]
\item
    $\psof{\cssconds}$
    for the elements of $\cssconds$ of the form $\csscond$ or
    $\cssneg{\csscond}$
    where $\csscond$ is of the form
    \[
        \begin{array}{c}
            \pslink, %
            \psvisited, %
            \pshover, %
            \psactive, %
            \psfocus, %
            \pstarget, \\%
            \psenabled, %
            \psdisabled, %
            \pschecked, %
            \psroot, %
            \text{ or }
            \psempty \ .
        \end{array}
    \]
\end{itemize}

If the worklist is empty, we terminate, and return that the automaton is empty.

If it is not empty, we take an arbitrary element
\[
    \tup{\astate',
         \broot,
         \bsibling,
         \btarget,
         \tset,
         \idcons,
         \poscons,
         \idxi}
\]
and for each transition
\[
    \atrant = \astate \atran{\arrgen}{\csssim} \astate'
\]
with
$\csssim = \csstype\cssconds$
we add to the worklist
\[
    \tup{\astate,
         \broot',
         \bsibling',
         \btarget',
         \tset',
         \idcons',
         \poscons',
         \idxi'}
\]
where $\idxi' = \idxi - 1$ is a fresh index, and when certain conditions are satisfied.
We detail these conditions and the definition of the new tuple below.
We begin with general conditions and definitions, then describe those specific to the value of $\arrgen$.

In all cases, we can only add a new tuple if
\begin{enumerate}
\item
    $\atrant \notin \tset$,
\item
    $\attspres{\csstype\cssconds}{\idxi'}$ is satisfiable,
\item
    $\psof{\cssconds}$ is satisfiable -- that is, we do not have
    $\pcls \in \cssconds$
    and
    $\cssneg{\pcls} \in \cssconds$
    for some pseudo-class $\pcls$, and, moreover, we do not have
    \begin{itemize}
    \item
        $\pslink \in \cssconds$
        and
        $\psvisited \in \cssconds$, or
    \item
        $\psenabled \in \cssconds$
        and
        $\psdisabled \in \cssconds$, or
    \item
        $\psroot \in \cssconds$
        and
        $\bsibling = \ptrue$.
    \end{itemize}
\end{enumerate}
In all cases, we define
\begin{itemize}
\item
    $\broot' =
     \begin{cases}
        \ptrue & \psroot \in \cssconds
        \\
        \broot & \arrgen = \arrlast
        \\
        \pfalse & \text{otherwise,}
    \end{cases}$
\item
    $\bsibling' =
     \begin{cases}
        \ptrue & \arrgen \in \set{\arrneighbour, \arrsibling}
        \\
        \bsibling & \arrgen = \arrlast
        \\
        \pfalse & \text{otherwise,}
     \end{cases}$
\item
    $\btarget' = \btarget \lor \brac{\pstarget \in \cssconds}$,
\item
    $\tset' =
     \begin{cases}
        \tset \cup \set{\atrant} & \astate = \astate'
        \\
        \emptyset & \text{otherwise,}
     \end{cases}$
\item
    $\idcons' = \idcons \cup \idcons''$
    where
    $\idcons''$
    is the set of all clauses
    $\attspresns{\ns}{\idatt}{\consset}{\idxi'}$
    appearing in
    $\attspres{\csstype\cssconds}{\idxi'}$
    for some $\ns$ and $\consset$.
\end{itemize}

Next, we give the conditions and definitions dependent on $\arrgen$.
To do so we need to define the set of positional constraints derived from
$\posof{\cssconds}$.
We use a slightly different encoding to the previous section.
We use a variable
$\numvar{\idxi}$
encoding that the node is the
$\numvar{\idxi}$th
child of the parent,
and
$\numtypevar{\idxi}{\qele{\ns}{\att}}$
counting the number of nodes of type
$\qele{\ns}{\att}$
to the left of the current node (exclusive).
To encode ``last of'' constraints, we use the variable
$\totnumvar$
to encode the total number of siblings of the current node (inclusive),
and
$\totnumtypevar{\qele{\ns}{\att}}$
to encode the total number of siblings of the given type (inclusive).

That is, when
$\broot' = \ptrue$
let
\begin{itemize}
    \item
        $\poscons'' = \set{\pfalse}$
        if
        $\bsibling' = \ptrue$,
    \item
        $\poscons'' = \set{\pfalse}$
        if there is some
        $\csscond \in \posof{\cssconds}$
        that is not of the form
        $\cssneg{\csscond'}$
        for some $\csscond'$, and
    \item
        $\poscons'' = \set{\ptrue}$
        otherwise,
\end{itemize}
and when
$\broot' = \pfalse$
let
$\poscons'' =$
\[
    \begin{array}{l}
        \setcomp{\exists \nvar\ .\ %
                     \numvar{\idxi'} = \coefa \nvar + \coefb}
                {\psnthchild{\coefa}{\coefb} \in \cssconds} \ \cup
        \\
        \setcomp{\exists \nvar\ .\ %
                     \totnumvar - \numvar{\idxi'} - 1 = \coefa \nvar + \coefb}
                {\psnthlastchild{\coefa}{\coefb} \in \cssconds} \ \cup
        \\
        \setcomp{\begin{array}{c}
                     \qele{\nsvar{\idxi'}}{\elevar{\idxi'}}
                     =
                     \qele{\ns}{\ele} \Rightarrow
                     \\
                     \exists \nvar\ .\ %
                         \numtypevar{\idxi'}{\qele{\ns}{\ele}}
                         =
                         \coefa \nvar + \coefb
                 \end{array}}
                {\begin{array}{c}
                     \ns \in \finof{\nspaces}
                     \land
                     \ele \in \finof{\eles}
                     \ \land
                     \\
                     \psnthoftype{\coefa}{\coefb} \in \cssconds
                 \end{array}} \ \cup
        \\
        \setcomp{\begin{array}{c}
                     \qele{\nsvar{\idxi'}}{\elevar{\idxi'}}
                     =
                     \qele{\ns}{\ele} \Rightarrow
                     \\
                     \exists \nvar\ .\ %
                         \totnumtypevar{\qele{\ns}{\ele}}
                         - \numtypevar{\idxi'}{\qele{\ns}{\ele}}
                         - 1
                         =
                         \coefa \nvar + \coefb
                 \end{array}}
                {\begin{array}{c}
                     \ns \in \finof{\nspaces}
                     \land
                     \ele \in \finof{\eles}
                     \ \land
                     \\
                     \psnthlastoftype{\coefa}{\coefb} \in \cssconds
                 \end{array}} \ \cup
        \\
        \setcomp{\nomatch{\numvar{\idxi'}}{\coefa}{\coefb}}
                {\cssneg{\psnthchild{\coefa}{\coefb}} \in \cssconds} \ \cup
        \\
        \setcomp{\nomatch{\totnumvar - \numvar{\idxi'} - 1}{\coefa}{\coefb}}
                {\cssneg{\psnthlastchild{\coefa}{\coefb}} \in \cssconds} \ \cup
        \\
        \setcomp{\begin{array}{c}
                     \qele{\nsvar{\idxi'}}{\elevar{\idxi'}}
                     =
                     \qele{\ns}{\ele} \Rightarrow
                     \\
                     \nomatch{\numtypevar{\idxi'}{\qele{\ns}{\ele}}}{\coefa}{\coefb}
                 \end{array}}
                {\begin{array}{c}
                     \ns \in \finof{\nspaces}
                     \land
                     \ele \in \finof{\eles}
                     \ \land
                     \\
                     \cssneg{\psnthoftype{\coefa}{\coefb}} \in \cssconds
                 \end{array}} \ \cup
        \\
        \setcomp{\begin{array}{c}
                     \qele{\nsvar{\idxi'}}{\elevar{\idxi'}}
                     =
                     \qele{\ns}{\ele} \Rightarrow
                     \\
                     \nomatch{\totnumtypevar{\qele{\ns}{\ele}}
                              - \numtypevar{\idxi'}{\qele{\ns}{\ele}}
                              - 1}
                             {\coefa}
                             {\coefb}
                 \end{array}}
                {\begin{array}{c}
                     \ns \in \finof{\nspaces}
                     \land
                     \ele \in \finof{\eles}
                     \ \land
                     \\
                     \cssneg{\psnthlastoftype{\coefa}{\coefb}} \in \cssconds
                 \end{array}} \ \cup
        \\
        \set{\totnumvar > \numvar{\idxi'}}
        \cup
        \set{
            \numvar{\idxi'} =
            \sum\limits_{\substack{
                \ns \in \finof{\nspaces}
                \\
                \ele \in \finof{\eles}
            }}
            \numtypevar{\idxi'}{\qele{\ns}{\ele}}
        }
        \ \cup
        \\
        \setcomp{\qele{\nsvar{\idxi'}}{\elevar{\idxi'}}
                 =
                 \qele{\ns}{\ele}
                 \Rightarrow
                 \totnumtypevar{\qele{\ns}{\ele}} > \numtypevar{\idxi'}{\qele{\ns}{\ele}}}
                {\ns \in \finof{\nspaces} \land \ele \in \finof{\eles}}
        \ \cup
        \\
        \setcomp{\qele{\nsvar{\idxi'}}{\elevar{\idxi'}}
                 \neq
                 \qele{\ns}{\ele}
                 \Rightarrow
                 \totnumtypevar{\qele{\ns}{\ele}} \geq \numtypevar{\idxi'}{\qele{\ns}{\ele}}}
                {\ns \in \finof{\nspaces} \land \ele \in \finof{\eles}}
        \ .
    \end{array}
\]
Then we have the following.
\begin{itemize}
\item
    When $\arrgen = \arrlast$, we set
    \[
        \begin{array}{rcl}
            \poscons'
            &=&
            \poscons \cup \poscons''\ \cup
            \\
            & &
            \set{\numvar{\idxi'} = \numvar{\idxi}} \cup
            \set{\qele{\nsvar{\idxi'}}{\elevar{\idxi'}}
                 =
                 \qele{\nsvar{\idxi}}{\elevar{\idxi}}} \ \cup
            \\
            & &
            \setcomp{\numtypevar{\idxi'}{\qele{\ns}{\ele}} =
                          \numtypevar{\idxi}{\qele{\ns}{\ele}}}
                         {\ns \in \finof{\nspaces}
                          \land
                          \ele \in \finof{\eles}} \ .
        \end{array}
    \]

\item
    When $\arrgen = \arrchild$,
    \begin{itemize}
    \item
        we require $\poscons$ is satisfiable,
        $\neg \broot$,
        and
        $\psempty \notin \cssconds$,
    \item
        we set
        $\poscons' = \poscons''$.
    \end{itemize}

\item
    When $\arrgen = \arrneighbour$,
    \begin{itemize}
    \item
        we require
        $\neg \broot$
        and
        $\psroot \notin \cssconds$,
    \item
        we set
        \[
            \begin{array}{rcl}
                \poscons'
                &=&
                \poscons \cup \poscons'' \cup
                \set{\numvar{\idxi} = \numvar{\idxi'} + 1}
                \ \cup
                \\
                & &
                \setcomp{
                    \qele{\nsvar{\idxi'}}{\elevar{\idxi'}} = \qele{\ns}{\ele}
                    \Rightarrow
                    \numtypevar{\idxi}{\qele{\ns}{\ele}}
                    =
                    \numtypevar{\idxi'}{\qele{\ns}{\ele}} + 1
                }{
                    \ns \in \finof{\nspaces}
                    \land
                    \ele \in \finof{\eles}
                } \ \cup
                \\
                & &
                \setcomp{
                    \qele{\nsvar{\idxi'}}{\elevar{\idxi'}} \neq \qele{\ns}{\ele}
                    \Rightarrow
                    \numtypevar{\idxi}{\qele{\ns}{\ele}}
                    =
                    \numtypevar{\idxi'}{\qele{\ns}{\ele}}
                }{
                    \ns \in \finof{\nspaces}
                    \land
                    \ele \in \finof{\eles}
                } \ .
            \end{array}
        \]
    \end{itemize}

\item
    When $\arrgen = \arrsibling$
    \begin{itemize}
    \item
        we require
        $\neg \broot$
        and
        $\psroot \notin \cssconds$,
    \item
        we set
        \[
            \begin{array}{rcl}
                \poscons'
                &=&
                \poscons \cup \poscons''\ \cup
                \\
                & &
                \set{\exists \shiftvar,
                             \brac{\shiftvartype{\qele{\ns}{\ele}}}_{\substack{
                                 \ns \in \finof{\nspaces}
                                 \\
                                 \ele \in \finof{\eles}
                             }} \ .
                        \brac{\begin{array}{c}
                            \numvar{\idxi} = \numvar{\idxi'} + \shiftvar
                            \land
                            \shiftvar \geq 1
                            \ \land \\
                            \bigwedge\limits_{\substack{
                                \ns \in \finof{\nspaces}
                                \\
                                \ele \in \finof{\eles}
                            }}\brac{
                                \begin{array}{c}
                                    \numtypevar{\idxi}{\qele{\ns}{\ele}}
                                    =
                                    \numtypevar{\idxi}{\qele{\ns}{\ele}}
                                    +
                                    \shiftvartype{\qele{\ns}{\ele}}
                                    \ \land \\
                                    \brac{
                                        \qele{\nsvar{\idxi'}}{\elevar{\idxi'}}
                                        =
                                        \qele{\ns}{\ele}
                                    }
                                    \Rightarrow
                                    \shiftvartype{\qele{\ns}{\ele}} \geq 1
                                \end{array}
                            }
                            \ \land \\
                            \shiftvar = \sum\limits_{\substack{
                                \ns \in \finof{\nspaces}
                                \\
                                \ele \in \finof{\eles}
                            }}
                            \shiftvartype{\qele{\ns}{\ele}}
                        \end{array}}
                } \ .
            \end{array}
        \]
    \end{itemize}
\end{itemize}

Moreover, if we were able to add the tuple to the worklist, and we have
\begin{itemize}
\item
    $\astate = \ainitstate$,
\item
    $\poscons'$ is satisfiable, and
\item
    The ID constraints are satisfiable,
\end{itemize}
then the algorithm terminates, reporting that the automaton is non-empty.
To check the ID constraints are satisfiable, we test satisfiability of the following formula.
Let $\attvalbound$ be the bound on the length of ID values, as derived in the previous section.
We assert
\[
    \bigwedge\limits_{\csscond \in \idcons'} \csscond
    \land
    \bigwedge\limits_{\substack{
        \ns \in \nspaces
        \\
        1 \leq \idxi_1, \idxi_2 \leq \numof
    }}
        \bigvee\limits_{1 \leq \idxj \leq \attvalbound}
            \wordpos{\ns}{\idatt}{\idxi_1}{\idxj}
            \neq
            \wordpos{\ns}{\idatt}{\idxi_2}{\idxj} \ .
\]
That is, we assert all ID conditions are satisfied, and all IDs are unique.

%% file: experiments-details.tex
\subsection{Sources of the CSS files used in our experiments}
\label{sec:benchmark-details}

We have collected \numbenchmarks\ CSS files from 41 global websites for our experiments. These websites cover the 20 most popular sites listed on Alexa~\cite{website-ranking}, which are Google, YouTube, Facebook, Baidu, Wikipedia, Yahoo!, Reddit, Google India, Tencent QQ, Taobao, Amazon, Tmall, Twitter, Google Japan, Sohu, Windows Live, VK, Instagram, Sina, and 360 Safeguard. Note that we have excluded Google India and Google Japan from our collection as we found the two sites share the same CSS files with Google. We have further collected CSS files from 12 well-known websites ranked between 21 and 100 on the same list, including LinkedIn, Yahoo!~Japan, Netflix, Imgur, eBay, WordPress, MSN, Bing, Tumblr, Microsoft, IMDb, and GitHub. Our examples also contain CSS files from several smaller websites, including Arch Linux, arXiv, CNN, DBLP, Google News, Londonist, The Guardian, New York Times, NetworkX, OpenStreetMap, and W3Schools. These examples were used in the testing and development of our tool.